\documentclass[12pt]{article}

\usepackage{amsfonts}
\usepackage{amsmath}
\usepackage{amssymb}
\usepackage{lscape}
\usepackage{float}
\usepackage[margin=.75in]{geometry}
\usepackage{graphicx}
\usepackage{natbib}
\usepackage{ntheorem}
\usepackage{rotating}
\usepackage{theorem}
\usepackage[usenames,dvipsnames]{color}
\usepackage{caption}

\theoremseparator{.}
\newtheorem{proposition}{Proposition}
\newcommand{\Indicator}{1 \hskip -2.5 pt \hbox{I}}
\newcommand{\Perp}{\perp\!\!\!\perp}
\newenvironment{proof}[1][Proof]{\begin{trivlist}
\newcommand{\qed}{\hfill $\blacksquare$}
\item[\hskip \labelsep {\bfseries #1.}]}{\end{trivlist}}
\DeclareMathOperator*{\argmax}{\arg\!\max}

\begin{document}

\title{Fact or friction: Jumps at ultra high frequency\thanks{We would like to thank Bj\o rn Eraker (the referee), David Bates, Tim Bollerslev, Michael Brennan, Alvaro Cartea, Frank Diebold, Dobrislav Dobrev, David Leigh, Richard Louth, Andrew Patton, Kevin Sheppard, George Tauchen, Valeri Voev and the participants at the Fourth CSDA (Computational Statistics \& Data Analysis) International Conference on Computational and Financial Econometrics, the Sixth Nordic Econometric Meeting, the Marie Curie High Frequency Research Training Workshop, the SITE workshop on Measuring and Modeling Risk with High Frequency Data, the 65th European Meeting of the Econometric Society, the 2012 North American Winter Meeting of the Econometric Society, and at seminars at CREATES, University of Aarhus, Cambridge University, Deutsche Bank, Queen Mary University of London, London School of Economics, Imperial College London, Duke University, Erasmus University Rotterdam, and Universidad Carlos III for helpful comments and insightful suggestions that greatly improved the paper. We are grateful to ICAP and the CME for giving permission to use the EBS foreign exchange and CME futures market data in this study. Kim Christensen and Mark Podolskij acknowledge financial support from CREATES, which is funded by the Danish National Research Foundation.}}

\author{Kim Christensen\thanks{CREATES, Aarhus University, Department of Economics and Business, Fuglesangs All\'{e} 4, 8210 Aarhus V, Denmark. Corresponding author. E-mail: kchristensen@creates.au.dk. Tel.: +45 8716 5367.} \and Roel C. A. Oomen\thanks{Deutsche Bank, 1 Great Winchester Street, London EC2N 2DB, United Kingdom and affiliated with University of Amsterdam, Department of Quantitative Economics, Valckenierstraat 65-67, 1018 XE Amsterdam, The Netherlands. E-mail: roel.ca.oomen@gmail.com} \and Mark Podolskij\thanks{Aarhus University, Department of Mathematics, Ny Munkegade 118, 8000 Aarhus C, Denmark. E-mail: mpodolskij@math.au.dk}}

\date{November 6, 2014}

\maketitle

\begin{abstract}
This paper shows that jumps in financial asset prices are often erroneously identified and are, in fact, rare events accounting for a very small proportion of the total price variation. We apply new econometric techniques to a comprehensive set of ultra high-frequency equity and foreign exchange tick data recorded at millisecond precision, allowing us to examine the price evolution at the individual order level. We show that in both theory and practice, traditional measures of jump variation based on lower-frequency data tend to spuriously assign a burst of volatility to the jump component. As a result, the true price variation coming from jumps is overstated. Our estimates based on tick data suggest that the jump variation is an order of magnitude smaller than typical estimates found in the existing literature.

\bigskip


\end{abstract}

\vfill

\thispagestyle{empty}

\pagebreak

\section{Introduction} \setcounter{page}{1}

A long-standing consensus in the literature on asset pricing is that a realistic dynamic model should incorporate several, if not all, of the following stylized facts: random walk behavior \citep*[e.g.,][]{fama:65a} at a macroscopic level, market microstructure effects \citep*[e.g.,][]{niederhoffer-osborne:66a} at a microscopic level, stochastic volatility \citep*[e.g.,][]{mandelbrot:63a}, leverage \citep*[e.g.,][]{black:76a}, and jumps \citep*[e.g.,][]{press:67a}. Extensive support for these factors can be found both in the theory of finance and in the abundantly available financial market data. In this paper, we examine the role of the jump component by applying new econometric techniques to a comprehensive set of the finest resolution equity and foreign exchange tick data. The use of individual order-level tick data in the context of jump identification is novel to this paper and necessary to arrive at our main finding that the jump component is substantially smaller than currently thought. Specifically, we show that jumps account for about 1\% of total price variability (i.e., quadratic variation) in contrast to accepted estimates from lower-frequency data, which are an order of magnitude larger. Our microscopic view at the tick data provides the intuition for this result: A burst of volatility is often spuriously identified as a jump at the lower frequencies commonly used in the literature. No doubt exists that, in times of stress, asset prices do move sharply over short periods of time, and while occasionally genuine price jumps do occur, we find that more often than not price continuity is preserved even when accompanied with a severe deterioration of liquidity.

\begin{table}
\setlength{ \tabcolsep}{0.05cm}
\caption{Selection of literature reporting estimates of the jump variation component.}
\label{Table:Literature}
\begin{center}
\begin{scriptsize}
\begin{tabular}{llllll}
\hline
Article & Data & Period & Frequency & Model & Jump variation \\
\hline
\citet*{press:67a}                          & Ten DJIA constituents     & 1926 -- 1960      & Monthly               & JD    & $20\%^a$          \\
\citet*{beckers:81a}                        & 47 US large-cap stocks    & 1975 -- 1977      & Daily                 & JD    & $25\%^b$          \\
\citet*{ball-torous:83a}                    & 47 US large-cap stocks    & 1975 -- 1977      & Daily                 & JD    & $50\%^c$          \\
\citet*{ball-torous:85a}                    & 30 US large-cap stocks    & 1981 -- 1982      & Daily                 & JD    & $47\%^d$          \\
\citet*{jorion:88a}                         & DM/$\$$ spot              & 1974 -- 1985      & Weekly                & JD    & $96\%~~$          \\
                                            & CRSP index                & 1974 -- 1985      & Weekly                & JD    & $36\%~~$          \\
\citet*{bates:96a}                          & DM/$\$$ options           & 1984 -- 1991      & Weekly                & SVJ   & $20\%~~$          \\
\citet*{bakshi-cao-chen:97a}                & S\&P 500 options          & 1988 -- 1991      & Daily                 & SVJ   & $18.9\%~~$        \\
\citet*{bates:00a}                          & S\&P 500 options          & 1988 -- 1993      & Weekly                & SVJ   & $30.3\%-38.5\%~~$ \\
\citet*{andersen-benzoni-lund:02a}          & S\&P 500 spot             & 1953 -- 1996      & Daily                 & SVJ   & $5.5\%~~$         \\
\citet*{bollerslev-zhou:02a}                & DM/$\$$                   & 1986 -- 1996      & Five minutes          & SVJ   & $7.4\%~~$         \\
\citet*{pan:02a}                            & S\&P 500 spot and options & 1989 -- 1996      & Weekly                & SVJ   & $55.7\%~~$        \\
Chernov et al. (2003)                       & DJIA spot                 & 1953 -- 1999      & Daily                 & SVJ   & $9.4\%~~$         \\
\citet*{eraker-johannes-polson:03a}         & S\&P 500 spot             & 1980 -- 1999      & Daily                 & SVJ   & $8.2\%-14.7\%~~$  \\
                                            & Nasdaq 100 spot           & 1985 -- 1999      & Daily                 & SVJ   & $6.0\%-17.0\%~~$  \\
\citet*{eraker:04a}                         & S\&P 500 spot and options & 1987 -- 1990      & Daily                 & SVJ   & $17.1\%~~$        \\
\citet*{johannes:04a}                       & US Treasury bills         & 1965 -- 1999      & Daily                 & SVJ   & $50\%~~$          \\
\citet*{maheu-mccurdy:04a}                  & 11 US large-cap stocks    & 1962 -- 2001      & Daily                 & GARCH & $29\%^e$          \\
                                            & DJIA spot                 & 1960 -- 2001      & Daily                 & GARCH & $16.6\%~~$        \\
                                            & Nasdaq 100 spot           & 1985 -- 2001      & Daily                 & GARCH & $14.4\%~~$        \\
                                            & TXX spot                  & 1995 -- 2001      & Daily                 & GARCH & $22.8\%~~$        \\
\citet*{barndorff-nielsen-shephard:04b}     & DM/$\$$                   & 1986 -- 1996      & Five minutes          & RV    & $3.1\%~~$         \\
\citet*{huang-tauchen:05a}                  & S\&P 500 futures          & 1982 -- 2002      & Five minutes          & RV    & $4.4\%~~$         \\
                                            & S\&P 500 spot             & 1997 -- 2002      & Five minutes          & RV    & $7.3\%~~$         \\
\citet*{barndorff-nielsen-shephard:06a}     & DM/$\$$, USDJPY           & 1986 -- 1996      & Five to 120 minutes   & RV    & $5.0\%-21.9\%~~$  \\
\citet*{bates:06a}                          & S\&P 500 spot             & 1953 -- 1996      & Daily                 & SVJ   & $12.7\%~~$        \\
\citet*{andersen-bollerslev-diebold:07a}    & DM/$\$$                   & 1986 -- 1999      & Five minutes          & RV    & $7.2\%~~$         \\
                                            & S\&P 500 spot             & 1990 -- 2002      & Five minutes          & RV    & $14.4\%~~$        \\
                                            & US Treasury bonds         & 1990 -- 2002      & Five minutes          & RV    & $12.6\%~~$        \\
\citet*{bollerslev-law-tauchen:08a}         & 40 US large-cap stocks    & 2001 -- 2005      & $17.5$ minutes        & RV    & $12\%~~$          \\
                                            & Equally weighted index    & 2001 -- 2005      & $17.5$ minutes        & RV    & $9\%~~$           \\
\citet*{jiang-oomen:07a}                    & S\&P 500 spot             & 1987 -- 1995      & Five minutes          & SVJ   & $18.6\%-19.5\%~~$ \\
\citet*{ait-sahalia-jacod:09b}              & INTC and MSFT             & 2006              & Five to 120 seconds   & RV    & $25\%^f$          \\
Bollerslev et al. (2009)                    & S\&P 500 futures          & 1985 -- 2004      & Five minutes          & RV    & $6.8\%~~$         \\
\citet*{todorov:09a}                        & S\&P 500 futures          & 1990 -- 2002      & Five minutes          & RV    & $15\%~~$          \\
\citet*{tauchen-zhou:11a}                   & S\&P 500 spot             & 1986 -- 2005      & Five minutes          & RV    & $5.4\%~~$         \\
                                            & US Treasury bonds         & 1991 -- 2005      & Five minutes          & RV    & $19.1\%~~$        \\
                                            & USDJPY                    & 1997 -- 2004      & Five minutes          & RV    & $6.5\%~~$         \\
\citet*{andersen-bollerslev-huang:11a}      & S\&P 500 futures          & 1990 -- 2005      & Five minutes          & RV    & $4.9\%~~$         \\
                                            & US Treasury bonds         & 1990 -- 2005      & Five minutes          & RV    & $14.6\%~~$        \\
\citet*{bates:12a}                          & CRSP index, S\&P 500      & 1926 -- 2006      & Daily                 & SVJ   & $6.4\%-7.2\%~~$   \\
\citet*{corsi-reno:12a}                     & S\&P 500 futures          & 1990 -- 2007      & Five minutes          & RV    & $6\%^g$           \\
\citet*{patton-sheppard:15a}                & S\&P 500 ETF              & 1997 -- 2008      & Five minutes$^h$      & RV    & $15\%~~$          \\
\hline
\end{tabular}
\end{scriptsize}
\medskip
\begin{scriptsize}
\parbox{\textwidth}{\emph{Note.} ``JD,'' ``SVJ,'' and ``GARCH'' refer to the \citet*{press:67a}- and \citet*{merton:76a}-type jump-diffusion model, the class of stochastic volatility plus jump models, and GARCH-type models, respectively. ``RV'' refers to the class of model-free measures of variation, including realized variance and (bi-) power variation. \\
\emph{a}. Individual JV between $0\%-100\%$; \emph{b}. Individual JV between $10\%-55\%$; \emph{c}. Individual JV between $10\%-80\%$; \emph{d}. Individual JV between $8\%-93\%$; \emph{e}. Individual JV between $20\%-41\%$. Two stocks have shorter sample from 1980s -- 2001; \emph{f}. Includes an infinite activity jump component; \emph{g}. Smoothed JV between $2\%-20\%$; \emph{h}. Five-minute sampling frequency equivalent in trade time.}
\end{scriptsize}
\end{center}
\end{table}
\nocite{bollerslev-kretschmer-pigorsch-tauchen:09a} \nocite{chernov-gallant-ghysels-tauchen:03a}

The foundations of most asset pricing models can be cast in the class of arbitrage-free It\^{o} semimartingales. These processes are naturally decomposed into a continuous diffusive Brownian component and a discontinuous jump part. The importance of being able to distinguish between these two fundamentally different sources of risk is emphasized in \citet*{ait-sahalia:04a}. Specifically, jumps have a profoundly distinct impact on option pricing \citep*[e.g.,][]{cox-ross:76a, merton:76a, duffie-pan-singleton:00a}, risk management \citep*[e.g.,][]{duffie-pan:01a, bakshi-panayotov:10a}, and asset allocation \citep*[e.g.,][]{jarrow-rosenfeld:84a, liu-longstaff-pan:03a}. Empirical work on identifying and modeling the jump component now spans nearly half a century. Table \ref{Table:Literature} provides a representative but necessarily incomplete overview. Starting with the influential paper of \citet*{press:67a}, and continuing up to \citet*{jorion:88a}, a number of papers estimate a (constant volatility) jump-diffusion model and report levels of jump variation (JV, expressed as a fraction of total return variation) in excess of $20\%$. An important shortcoming of the \citet*{press:67a}- or \citet{merton:76a}-style jump-diffusion model is that the jump component is the only mechanism that can account for fat tails of the empirical return distribution so that---in the presence of stochastic volatility---the JV measurements are potentially inflated. From the 1990s onwards, a large body of work considers numerous generalizations of the jump-diffusion model to include one or several stochastic volatility factors as well as state-dependent jump components. Estimation methods for such models are often highly complex and numerically intensive \citep*[e.g.,][]{eraker-johannes-polson:03a} but have the nice feature that they can exploit the information in the spot price \citep*[e.g.,][]{andersen-benzoni-lund:02a}, its associated derivative prices \citep*[e.g.,][]{bates:96a}, or both \citep*[e.g.,][]{pan:02a}. The majority of this literature concentrates on the US large-cap Standard \& Poor's (S\&P) 500 equity index and typically finds that the JV is around $10\% - 20\%$. The corresponding figure for foreign exchange rates is comparable and that of Treasury bills and individual stocks is higher still.

The most recent work on jumps has seen a shift away from model-based inference on low-frequency data to model-free inference based on intraday data. In an influential series of papers, \citet*{barndorff-nielsen-shephard:04b, barndorff-nielsen-shephard:06a} and \citet*{barndorff-nielsen-shephard-winkel:06a} introduce the concept of (bi-) power variation---a simple but effective technique to identify and measure the variation of jumps from intraday data [see \citet*{ait-sahalia-jacod:09b, ait-sahalia-jacod:09a} and \citet*{mancini:04a, mancini:09a} for a related jump-robust threshold estimator]. Using this, or variations thereof, a number of recent articles report model-free JV estimates of around 10\% for the S\&P 500 index and thus reinforce the earlier literature that the jump component is important \citep*[e.g.,][]{huang-tauchen:05a, andersen-bollerslev-diebold:07a}.

\begin{figure}[ht!]
\begin{center}
\caption{The S\&P 500 flash-crash: Jump or burst in volatility?}
\label{Figure:Crash}
\begin{tabular}{cc}
\small{Panel A: Five-minute data} & \small{Panel B: Trade-by-trade}\\
\includegraphics[height=8cm,width=0.48\textwidth]{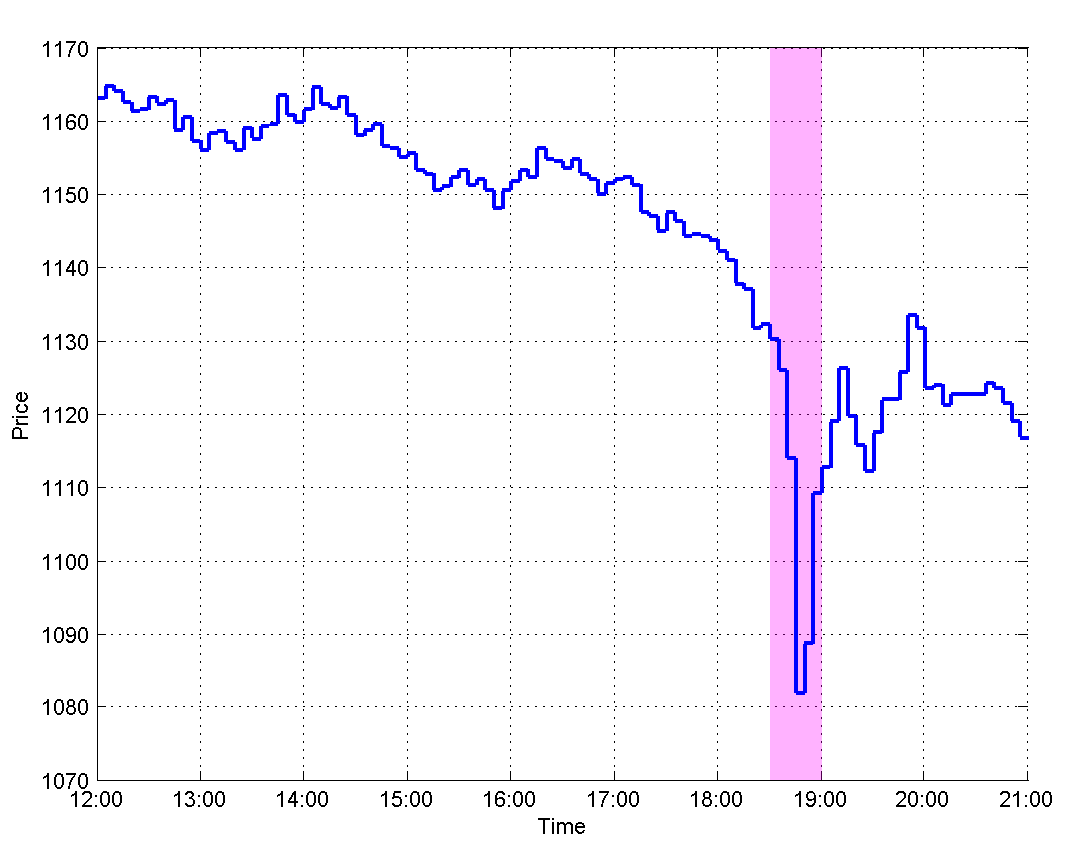} &
\includegraphics[height=8cm,width=0.48\textwidth]{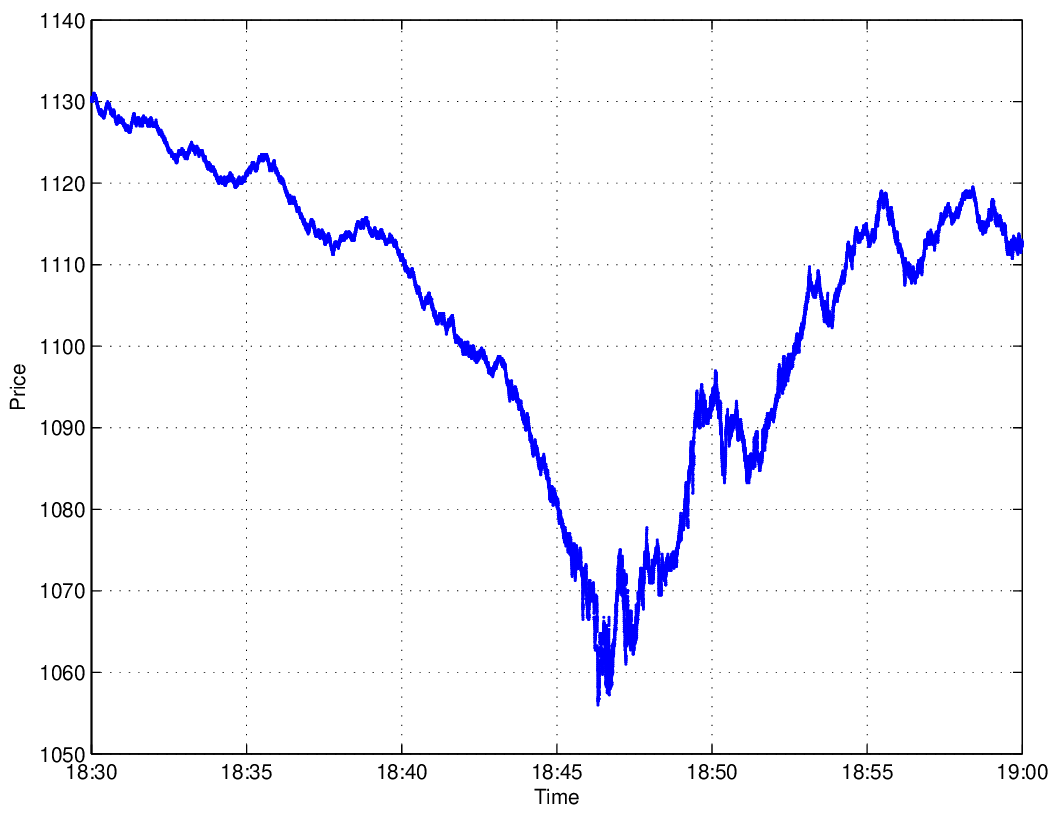} \\
\end{tabular}
\begin{scriptsize}
\parbox{\textwidth}{\emph{Note.} This figure draws the price evolution of the Chicago Mercantile Exchange's E-mini S\&P 500 futures contract on May 6, 2010. Panel A displays the data for the whole day at a five-minute frequency, while Panel B is at the trade frequency over the flash-crash episode. The reported time is Greenwich Mean Time, while local time in the US (EDT) is four hours earlier.}
\end{scriptsize}
\end{center}
\end{figure}

By their nature, jumps are instantaneous and discrete moves in the price. Therefore, not surprisingly, identification is aided by the finest resolution price data. This is precisely what motivates the recent literature to use intraday data. However, the consensus five-minute frequency at which returns are typically sampled, instead of at the finest tick-by-tick resolution, reflects a compromise to ensure that the market microstructure effects are sufficiently benign for the theory to remain valid. Figure \ref{Figure:Crash} illustrates the cost of doing so: a diminished ability to distinguish jumps from bursts in volatility. From Panel A, in which returns are sampled at a five-minute frequency, one would probably conclude that the notorious flash-crash episode contains a number of very large jumps [see, e.g., \citet*{easley-prado-ohara:11a} for a discussion of the event]. The widely used bi-power variation (BV) jump measure is highly significant indicating the presence of jumps on formal statistical grounds. Yet, when focusing on the relevant subperiod in Panel B, one can see that at tick frequency jumps are elusive. In fact, in a recent study using audit trail data, \citet*{kirilenko-kyle-samadi-tuzun:17a} characterize the flash-crash as a ``brief period of extreme market volatility.'' The March 2011 earthquake in Japan led to similar scenarios in the US dollar-Japanese yen (USDJPY) exchange rate, which experienced a flash-crash-type episode on March 16 and a rapid depreciation following a coordinated intervention by the Bank of Japan and other G7 central banks on March 18. Again, both events would be classified as exhibiting large jumps by conventional realized measures based on five-minute data, while the tick data reveal a period of heightened volatility instead of discrete price jumps (see Figure \ref{Figure:Quake}).

This paper provides a comprehensive study into the magnitude of the jump component based on the highest resolution data available. To the best of our knowledge, we are the first to do so, and we find that the tick data have a fundamentally different story to tell. Our analysis employs the pre-averaging theory of \citet*{podolskij-vetter:09a} and \citet*{jacod-li-mykland-podolskij-vetter:09a} to construct noise-robust jump measures from tick data. We apply these to a representative data set composed of US large-cap stocks (the 30 Dow Jones Industrial Average (DJIA) constituents), equity indexes (the S\&P 500 and Nasdaq 100), and currency pairs (the Euro-US dollar (EURUSD), US dollar-Japanese yen (USDJPY), and US dollar-Swiss franc (USDCHF)). We start by confirming that when sampling at low frequencies the jump component appears substantial and in line with the extant literature at around 10\%. In passing, we observe a near-perfect relation in which the jump component decreases in magnitude as the sampling frequency is increased from 15 minutes to five minutes. Next, we obtain our main result, in which most of the previously identified jumps vanish as we move to the tick frequency and we are left with highly volatile episodes instead. The overall JV measured across all instruments we consider is just over 1\%.

The finding of a much-diminished role for jumps clearly carries important implications for asset pricing models and their associated finance applications. The observational equivalence---at lower frequencies---between a sharp but continuous move in the price and a sudden discontinuous jump highlights an important point. When market participants view prices at periodic but infrequent intervals (e.g., every hour, or at the end of each trading day), they could be exposed to price moves that they experience as genuine jumps when a period of extreme volatility unfolds in between observation points. For instance, an options trader that follows the market on an hourly basis for the purpose of hedging can experience jumps and could incur a profoundly different pay-off pattern compared with an automated hedging algorithm operating in tick-time. Similarly, traders that require a certain capacity (e.g., a derivatives dealer that is hedging a large book of options, or a pension fund managing a large portfolio) can also experience price jumps as it takes time and price discovery to execute large trades, particularly in times of stress. So even when the underlying price process is free of jumps, the trader's pay-off and execution costs can be impacted in a manner consistent with the presence of genuine price jumps. As such, models with jumps still have an important role to play in scenarios such as these. This does not, however, contradict the findings in this paper. In fact, we discuss in some detail the intimate link between price continuity, liquidity, volatility dynamics, and observation frequency as it provides the intuition for our results and reconciles their relation to the extant literature.

The rest of the paper is organized as follows. Section \ref{Sec:Theory} describes the theoretical model and presents the required estimators and asymptotic theory needed to measure the jump component. Section \ref{Sec:Empirical} summarizes the empirical findings and Section \ref{Sec:Conclusion} concludes. The proofs and more theoretical discussions are contained in the Appendices. Additional empirical results and case studies are collected in an online Appendix \citep*[see][]{christensen-oomen-podolskij:13b}.

\section{Jump measurement from noisy tick data} \label{Sec:Theory}

To study the role of jumps in financial asset prices and develop the necessary theory, we need a model that captures the salient and hypothesized aspects of tick data: stochastic volatility, jumps in the price, market microstructure noise, and outliers. We model the first two components as a latent continuous-time process that represents the unobserved efficient price process. The last two features are modeled on top of the discretely sampled efficient price to represent the observed tick data process. Within this setting, the econometric challenge we address is how to infer the magnitude of the jump component from noisy and discretely sampled measurements of the efficient price process. What is novel to this paper is the consideration for microstructure noise and outliers.

\subsection{A model for noisy tick data}

Starting with the formulation of the efficient price process, let $X=(X_t)_{t\geq 0}$ denote the logarithmic price of a financial asset that is defined on a filtered probability space $\bigl( \Omega, \mathcal{F}, \left( \mathcal{F}_{t} \right)_{t \geq 0}, \mathbb{P} \bigr)$ and is adapted to the filtration $\mathcal{F}_{t}$ that represents the information available to market participants at time $t$, $t \geq 0$. We assume that $X$ operates in an arbitrage-free frictionless market, which implies that $X$ belongs to the class of semimartingale processes \citep*[e.g.,][]{back:91a,delbaen-schachermayer:94a}. To allow for stochastic volatility and (finite activity) price jumps, we further assume that $X$ can be represented by a jump-diffusion model of the form
\begin{equation} \label{Eqn:Xprocess}
X_t = X_0 + \int_{0}^{t} a_{s} \text{d}s + \int_{0}^{t} \sigma_{s} \text{d}W_{s} + \sum_{i = 1}^{N_t^J} J_i, \qquad t \geq 0,
\end{equation}
where $X_{t}$ is the log-price at time $t$, $a = \left( a_{t} \right)_{t \geq 0}$ is a locally bounded and predictable drift term, $\sigma = \left( \sigma_{t} \right)_{t \geq 0}$ is an adapted c\`{a}dl\`{a}g volatility process, $W = \left( W_{t} \right)_{t \geq 0}$ is a standard Brownian motion, $N^J = \left( N^J_{t} \right)_{t \geq 0}$ is a counting process, and $J = \left( J_{i} \right)_{i = 1, \ldots, N_{t}^J}$ is a sequence of nonzero random variables. Here, $N^J$ represents the total number of jumps in $X$ that have occurred up to time $t$ and $J$ denotes the corresponding jump sizes.

The above jump-diffusion model is completely nonparametric and, apart from a regularity condition imposed on the variance process to facilitate the proof of the central limit theorem, we do not require the driving processes to be specified. As such, it nests many well-known continuous-time models as special cases: for example, the geometric Brownian motion with an Ornstein-Uhlenbeck process for log-volatility as in \citet*{alizadeh-brandt-diebold:02a}, the stochastic volatility model with log-normal jumps generated by a non-homogeneous Poisson process as in \citet*{andersen-benzoni-lund:02a}, and the affine class of models as in \citet*{duffie-pan-singleton:00a}. Moreover, as discussed in more detail in \ref{Sec:ApndxProof}, the model and our theory allow for leverage effects \citep*[e.g.,][]{christie:82a} and the proofs can be extended to incorporate jumps in the volatility process as in, e.g., \citet*{eraker-johannes-polson:03a}.

\begin{figure}[ht!]
\begin{center}
\caption{An illustration of noisy tick data.}
\label{Figure:Noise}
\begin{tabular}{cc}
\small{Panel A: Microstructure noise} & \small{Panel B: Outliers} \\
\includegraphics[height=8cm,width=0.48\textwidth]{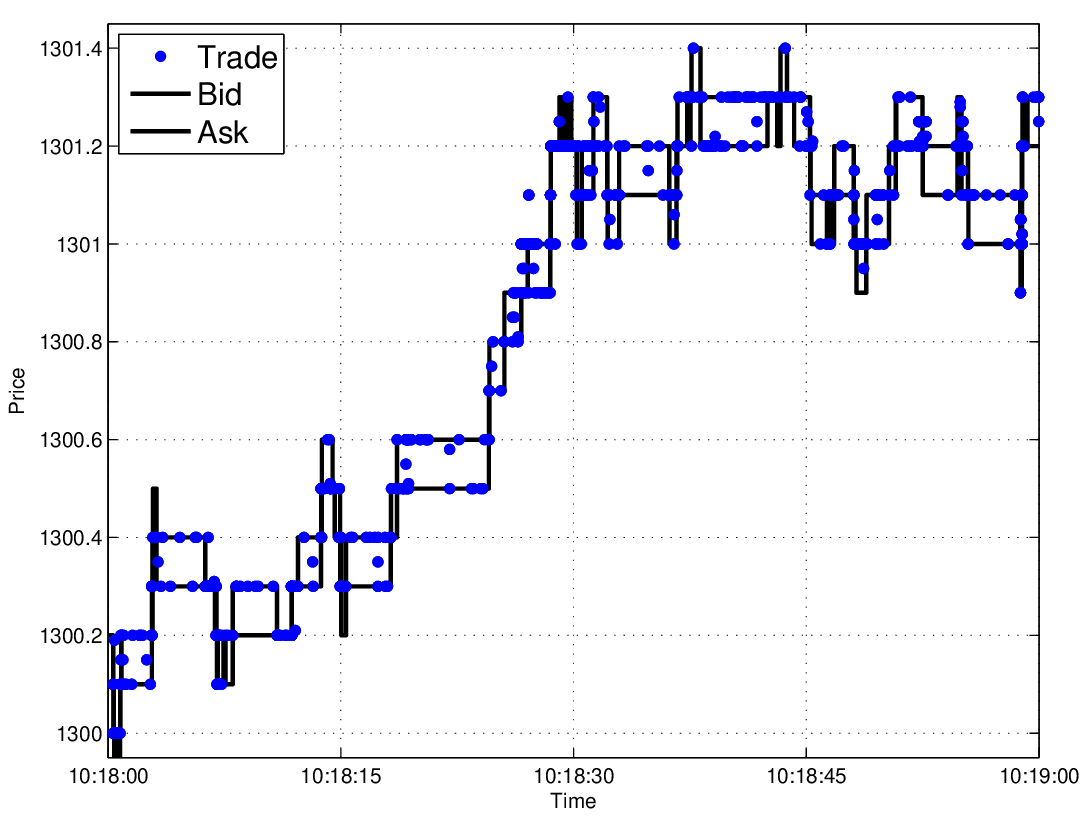} &
\includegraphics[height=8cm,width=0.48\textwidth]{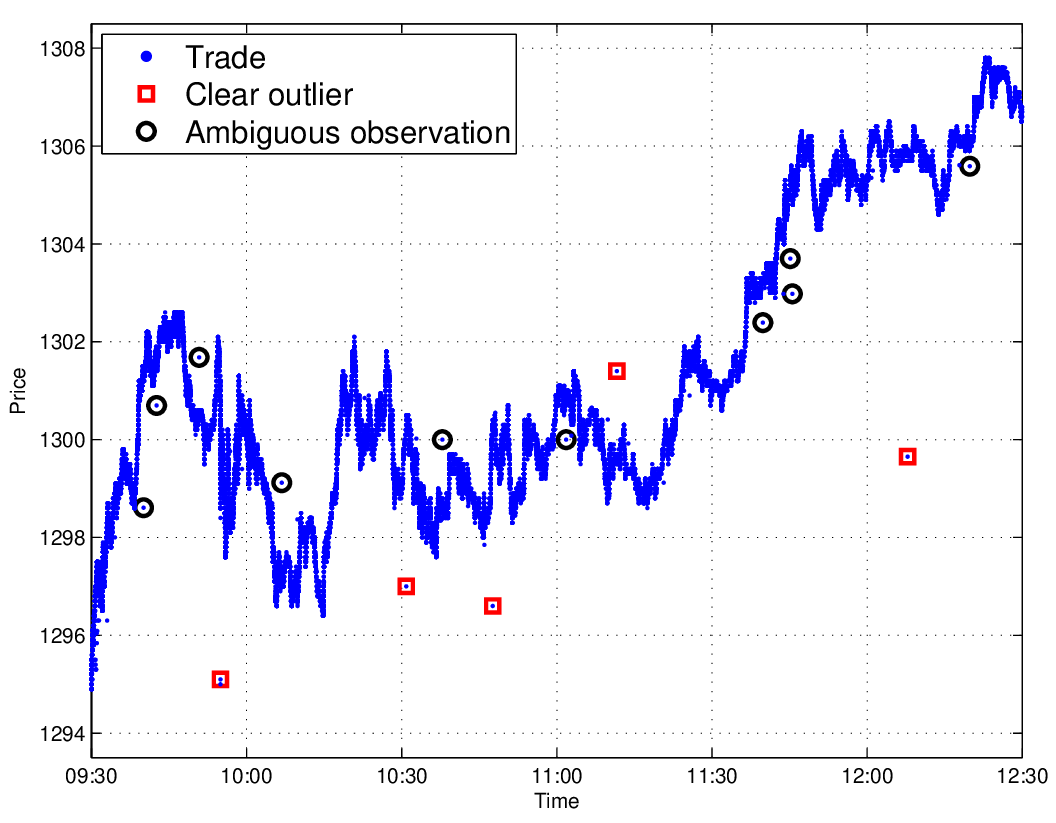} \\
\end{tabular}
\begin{scriptsize}
\parbox{\textwidth}{\emph{Note.} Panel A plots one minute of National Best Bid and Offer (NBBO) quote data overlaid with trades for SPY on March 11, 2011. It serves to illustrate the microstructure noise effects such as price discreteness and bid-ask bounce. Panel B plots SPY trades for the full morning session of the same day and highlights clear and suspected outliers. Amongst the clear outliers, four out of five are trades for 100,000 shares or more (and therefore qualify for delayed reporting of block trades). The first outlier is a trade for only two hundred shares or about half the average trade size for this period.}
\end{scriptsize}
\end{center}
\end{figure}

In practice we do not observe the efficient price process in continuous time but instead have discretely sampled measurements available that---at an individual order level---can contain a substantial amount of noise. One source of noise is microstructure effects that arise from market imperfections such as bid-ask spreads and price discreteness \citep*[e.g.,][]{niederhoffer-osborne:66a, roll:84a, black:86a}. Outliers are another source of noise but have received little attention in the literature thus far. These short-lived off-market prices often corrupt tick data and can be attributed to delayed trade reporting on block trades, fat-finger errors, bugs in the data feed, misprints, decimal misplacement, incorrect ordering of data, etc.\footnote{Pre-filtering tick data for outliers is standard practice. In some instances, most if not all the outliers can be filtered based on exchange-provided information such as a trade condition. In cases in which this is unreliable or not available, data-driven techniques as in \citet*{barndorff-nielsen-hansen-lunde-shephard:09a} can be used. A drawback of such an approach is that the filtering rules typically rely on ad hoc tuning parameters that control the tolerance levels to outliers and the user consequently risks over- or under-cleaning the data. The method for jump measurement we present here is robust to (residual) outliers and can thus be used in conjunction with or instead of a pre-filtering method.} Figure \ref{Figure:Noise} provides an illustration of both types of noise. With this in mind, and restricting attention to the unit time interval $t \in [0,1]$, we model the observed tick data process as
\begin{equation} \label{Eqn:Yprocess}
Y_{i/N} = X_{i/N} + u_{i/N} + O_{i/N}, \qquad i = 0,1,\ldots,N
\end{equation}
where $u$ is an identically and independently distributed or i.i.d. (microstructure) noise process with $E(u) = 0$ and $E(u^2) = \omega^2$, and $u$ is independent of $X$, or $u \Perp X$. Further, $O_{i/N} = \Indicator_{ \{ i / N \in \mathcal{A}_N \} } \mathcal{S}_i$ where $\mathcal{A}_N$ is a random set holding the appearance times of the outliers and their sizes are given by $(\mathcal{S}_i)_{i = 1, \ldots, N_1^O}$. We assume that $\mathcal{A}_N$ is a.s. finite and model it by
\begin{equation} \label{An}
\mathcal{A}_N = \left\{ \frac{[N T_i]}{N} : 0 \leq T_i \leq 1 \right\},
\end{equation}
where $\left( T_i \right)_{i = 1,\ldots, N_1^O}$ are the arrival times of another counting process $N^O = \left(  N_t^O \right)_{t \geq 0}$. We assume that $O$ is mutually independent of $X$ and $u$, $O \Perp (X, u)$. This implies that $N^J \Perp N^O$, i.e., the two counting processes generating jumps and outliers, are also independent, which further means that the probability of observing both a jump and an outlier in a small time interval is asymptotically negligible.

The formulation of the microstructure component is standard in the literature on high-frequency data [see \citet*{andersen-bollerslev-diebold:10a} and \citet*{barndorff-nielsen-shephard:07a} for comprehensive reviews]. It provides analytic convenience and also has some empirical support \citep*[see, e.g.,][]{hansen-lunde:06a, diebold-strasser:13a} because it captures the empirically observed negative first-order autocorrelation in tick returns [e.g., induced by bid-ask bouncing, see \citet*{roll:84a}]. The i.i.d. assumption, however, is not binding [\citet*{jacod-li-mykland-podolskij-vetter:09a} provide a more general treatment of the noise. Our results extend along those lines as well.] The outlier process is novel to this paper. It is sufficiently general for our purposes and can capture the first-order effects of the observed off-market prices.

\subsection{Measurement of the jump component}

Within the above setup, the magnitude of the jump component expressed as a fraction of total (quadratic) price variation is given by
\begin{equation} \label{Eqn:JV}
JV = \frac{[X]_1 - \int_0^1 \sigma_s^2 \text{d}s }{[X]_1}, \qquad \text{where} \qquad [X]_1 = \int_0^1 \sigma_s^2 \text{d}s + \sum_{i = 1}^{N_1^J} J_i^2.
\end{equation}
Measurement of these components based on discretely sampled observations of the efficient price process $X$ has been widely studied and is well understood. Specifically, a consistent estimator of $[X]_1$ is given by the realized variance (RV), see, e.g., \citet*{andersen-bollerslev-diebold-labys:01a}
\begin{equation} \label{Eqn:RVn}
RV = \sum_{i = 1}^N |r_i|^2 \overset{p}{\to} [X]_1 \qquad \text{as} \qquad N \to \infty,
\end{equation}
where $r_i = X_{i/N} - X_{(i - 1)/N}$ are logarithmic returns of the efficient price process.\footnote{Consistency follows as a direct consequence from the definition of the quadratic variation process
\begin{equation*}
[X]_t = \underset{N \to \infty}{\text{p-lim}} \sum_{i = 1}^{N} \left( X_{t_{i}} - X_{t_{i - 1}} \right)^2,
\end{equation*}
for any sequence of partitions $0 = t_0  < t_1  < ... < t_N  = t$ with $\sup_{i} \left\{ t_{i} - t_{i - 1} \right\} \to 0$ as $N \to \infty$ \citep*[e.g.,][]{protter:04a}. It is this fundamental result from stochastic calculus that has motivated the increasing use of high-frequency data to estimate financial volatility.} An intuitive and elegant way to separate the diffusive- and jump-variation components from $[X]_1$ is via the use of the \citet*{barndorff-nielsen-shephard:04b} bi-power variation measure: a jump-robust and consistent estimator of the integrated variance\footnote{Alternative ways of estimating the integrated variance in the presence of jumps include the threshold estimators of \citet*{ait-sahalia-jacod:09b, ait-sahalia-jacod:09a} and \citet*{mancini:04a, mancini:09a}, and the quantile-based estimators of \citet*{christensen-oomen-podolskij:10a} and \citet*{andersen-dobrev-schaumburg:12a}. Recent developments also show how to improve the finite sample jump robustness of the BV or make it
more efficient; see \citet*{corsi-pirino-reno:10a} and \citet*{mykland-shephard-sheppard:12a}.}
\begin{equation} \label{Eqn:BVn}
BV = \frac{N}{N - 1} \frac{\pi}{2} \sum_{i = 2}^{N} |r_{i-1}| |r_i| \overset{p}{\to}  \int_0^1 \sigma_s^2 \text{d}s
\qquad \text{as} \qquad N \to \infty.
\end{equation}
The adjustment factor $N / (N - 1)$ applied to BV is used to improve its finite sample properties. Using Eqs. \eqref{Eqn:RVn} and \eqref{Eqn:BVn}, the difference between realized variance and bi-power variation can be used to isolate the JV, i.e.,
\begin{equation}
RV - BV \overset{p}{\to} \sum_{i = 1}^{N_1^J} J_i^2.
\end{equation}
This result has been widely used in the literature cited in Table \ref{Table:Literature} that studies the role of the jump component using intraday five-minute data.

An inherent limitation associated with discretely sampled data is that the ability to distinguish a jump from a diffusive movement diminishes as the sampling frequency decreases. An illustration of this was given in Figure \ref{Figure:Crash}. In this paper we therefore set out to utilize the finest resolution tick data available to study the jump component. At this level of granularity, however, microstructure effects invalidate the standard RV and BV measures described above. So to make inference about $[X]_1$ and its components using noisy tick data, we make use of the pre-averaging approach introduced by \citet*{jacod-li-mykland-podolskij-vetter:09a} and \citet*{podolskij-vetter:09a, podolskij-vetter:09b}.\footnote{When the object to be estimated is quadratic variation, the pre-averaging approach is to first-order equivalent to the realized kernel-based estimator of \citet*{barndorff-nielsen-hansen-lunde-shephard:08a} and the two-scale or multi-scale subsampler of \citet*{zhang-mykland-ait-sahalia:05a} and \citet*{zhang:06a}.} Intuitively, this method locally smooths the observed price series $Y$ so that the microstructure component $u$ (almost) disappears under averaging. The outlier component $O$ is dealt with via the same mechanism. Returns on this pre-averaged price series can then be used to construct consistent measures of the diffusive- and jump-variation components.

To implement pre-averaging, we calculate returns on a price series that is pre-averaged in a local neighborhood of $K$ observations, i.e.,
\begin{equation} \label{Eq:PreavgP}
r^\ast_{i,K} = \frac{1}{K} \left(\sum_{j = K/2}^{K - 1} Y_{(i+j)/N} - \sum_{j = 0}^{K/2 - 1} Y_{(i+j)/N}\right),
\end{equation}
where $K\geq 2$ and even. For the asymptotics to work, it is required that $K = \theta \sqrt{N} + o(N^{- 1/4})$, i.e., the pre-averaging horizon grows as the sampling frequency increases, but at a slower rate. We can then construct noise- and outlier-robust versions of RV and BV as
\begin{align}
RV^\ast &= \frac{N}{N - K + 2} \frac{1}{K \psi_K} \sum_{i = 0}^{N - K + 1} |r^\ast_{i,K}|^2 - \frac{\widehat{\omega}^{2}}{\theta^{2} \psi_K} ,\label{Eqn:RVast} \\
BV^\ast &= \frac{N}{N - 2K + 2} \frac{1}{K\psi_K}\frac{\pi}{2} \sum_{i = 0}^{N - 2K + 1} |r^\ast_{i,K}| |r^\ast_{i+K,K}| - \frac{\widehat{\omega}^{2}}{\theta^{2} \psi_K},\label{Eqn:BVast}
\end{align}
where $\psi_K = (1+2K^{-2})/12$.

The term $\frac{\widehat{\omega}^{2}}{\theta^{2} \psi_K}$ is a bias-correction, which compensates for the residual microstructure noise that remains after pre-averaging. It drops out when we compute $RV^\ast-BV^\ast$ and so it is of limited importance. $\widehat{\omega}^2$ is an estimator of the noise variance $\omega^2$. It can be estimated in a number of ways [see, e.g., \citet*{gatheral-oomen:10a} for a comparison of estimators]. Here, we use the estimator proposed by \citet*{oomen:06b}, which is defined as
\begin{equation} \label{Eqn:omega}
\hat{\omega}^{2}_{\text{AC}} = - \frac{1}{N - 1} \sum_{i = 2}^{N} r_{i}^{\ast} r_{i-1}^{\ast},
\end{equation}
where $r_{i}^{\ast} = Y_{i/N} - Y_{(i-1)/N}$.

Proposition \ref{Thm:CLT} generalizes the results of \citet*{podolskij-vetter:09a} to allow for outliers. It gives the probability limit of the pre-averaging estimators and their joint asymptotic distribution under the null hypothesis of no jumps.
\begin{proposition} \label{Thm:CLT}
Assume that $Y$ follows Eq. \eqref{Eqn:Yprocess} and that $E(u^4) < \infty$. As $N \to \infty$, it holds that
\begin{equation} \label{Eqn:RV*BV*consistency}
RV^\ast \overset{p}{\to} [X]_1, \qquad BV^\ast \overset{p}{\to} \int_0^1 \sigma_s^2 \text{\upshape{d}}s.
\end{equation}
Moreover, suppose that $E(u^8)<\infty $, and that $X$ is a continuous semimartingale, i.e., $X$ follows Eq. \eqref{Eqn:Xprocess} but with $N_t^J \equiv 0$ for all $t$, and with condition (V) as listed in the proof fulfilled. As $N \to \infty$, it then further holds that
\begin{equation} \label{Eqn:RV*BV*CLT}
N^{1/4} \left(
\begin{array}{c}
RV^\ast - \int_0^1 \sigma_s^2 \text{\upshape{d}}s \\[0.25cm]
BV^\ast - \int_0^1 \sigma_s^2 \text{\upshape{d}}s
\end{array}
\right) \overset{d_s}{\to} MN(0, \Sigma^\ast),
\end{equation}
a mixed normal distribution with conditional covariance matrix $\Sigma^\ast$, where $\Sigma^\ast$ is defined in \ref{Sec:ApndxSigma}.\footnote{Throughout the paper, the symbol ``$\overset{d_{s}}{\to}$'' is used to denote convergence in law stably. We refer to \citet*{barndorff-nielsen-hansen-lunde-shephard:08a} for a formal definition of stable convergence in law and the motivation for using this type of convergence in the high-frequency volatility setting. The $N^{-1/4}$ rate of convergence is slow compared to the noiseless case, but it is the fastest possible that can be achieved in noisy diffusion models \citep*[see][]{gloter-jacod:01a, gloter-jacod:01b}.}
\end{proposition}
\begin{proof}
See \ref{Sec:ApndxCLT}.
\end{proof}
The consistency result in Eq. \eqref{Eqn:RV*BV*consistency} implies that in the presence of noise and outliers, we can consistently estimate the magnitude of the jump component as
\begin{equation} \label{Eqn:NoiseJumpEstimation}
RV^\ast - BV^\ast \overset{p}{\to} \sum_{i = 1}^{N_1^J} J_i^2.
\end{equation}
Eq. \eqref{Eqn:RV*BV*CLT} supplies the basis for constructing a nonparametric noise-robust jump test. The online Appendix \citep*{christensen-oomen-podolskij:13b} provides the derivation and proposes a consistent jump-robust estimator of $\Sigma^{*}$, such that the test can be implemented in practice. Moreover, it highlights the finite sample properties of the test using simulations and applies it to the empirical data that are analyzed below. Because the test is not required to arrive at our main result, we leave it out of the paper.

\subsection{Simulation study}

We now perform a small Monte Carlo study to gauge the finite sample performance of the noise- and outlier-robust pre-averaging theory and focus specifically on its ability to back out the diffusive- and jump-variation components.

We simulate the logarithm of the efficient price $X$ from four distinct models, namely (1) a Brownian motion (BM) model where $\text{d}X_t = \sigma_t \text{d}W_t$ with $\sigma_t^2 = 0.0391$ corresponding to a return volatility of about 20\% annualized, (2) a \citet*{heston:93a}-type stochastic volatility (SV) model, using parameter values as in \citet*{bakshi-ju-ou-yang:06a}, (3) a two-factor stochastic volatility model with leverage (SV2F-LEV) proposed by \citet*{chernov-gallant-ghysels-tauchen:03a}, using calibrated parameter values reported in \citet*{huang-tauchen:05a}, (4) a Brownian motion plus jump (BMJ) model, in which we position a jump of random size at a random point in the series, ensuring that, on average, it accounts for $20\%$ of total variation. To study the impact of outliers, we also consider a fifth Brownian motion plus outlier (BMO) model, in which we insert an outlier of random size at a random point in the series, ensuring that, on average, it accounts for $20\%$ of total variation.

We simulate ten thousand independent price paths for each model using an Euler discretisation scheme with $N = 40,000$ and add i.i.d. microstructure noise where we select $\omega^2$ by fixing the noise ratio \citep*{oomen:06b} $\gamma = \sqrt{N \omega^2 / \int_0^1 \sigma_s^2 \text{d}s} = 0.50$. This ensures that the magnitude of the noise is in proportion to the efficient price variation. The values of $N$ and $\gamma$ are representative of the tick data in the empirical section (see Table \ref{Table:Descriptive}). While the i.i.d. noise assumption is standard in the literature, some evidence exists that noise at tick frequency can be dependent. We therefore consider an additional scenario
\begin{equation}
u_i = \beta u_{i-1} + \varepsilon_i,
\end{equation}
where $\varepsilon_i \overset{\text{i.i.d.}}{\sim} N(0,\omega^2(1-\beta^2))$. We set $\beta$ using an ad hoc approach by matching it to the coefficient of an AR(1) model calibrated to the trade sign of the S\&P 500 futures contract on the day of the flash-crash (see Figure \ref{Figure:Crash}).\footnote{The Chicago Mercantile Exchange (CME) E-mini S\&P 500 futures contract data have signed trade sizes to indicate if the trade was the result of a buy or sell order.} This gives $\beta = 0.77$.\footnote{If we aggregate trades by millisecond timestamp---a sensible approach for the market data distribution mechanism of the CME in which a collection of individual passive orders that get filled instantaneously against a single aggressor trade all get reported separately---the estimate is substantially lower at $\beta=0.41$. Thus, the value used in the simulations is conservative and, if anything, overstates the actual dependence in microstructure noise.} The setup allows us to gauge the ability of the pre-averaged estimators to isolate the jump component in the presence of stochastic volatility, microstructure noise, and outliers.

\begin{table}
\setlength{\tabcolsep}{0.08cm}
\caption{Properties of the pre-averaging estimators using simulated data.}
\label{Table:RVrobust}
\begin{center}
\begin{tabular}{lccccclccccclccccc}
\hline
& \multicolumn{5}{c}{$RV^\ast$} & & \multicolumn{5}{c}{$BV^\ast$} & & \multicolumn{5}{c}{$BV^\ast_\tau$}\\
\cline{2-6} \cline{8-12} \cline{14-18}
~~~~~~~~~~~~~~~~~~~~$\theta=$ & 0.1 & 0.5 & 1.0 & 2.0 & 5.0 && 0.1 & 0.5 & 1.0 & 2.0 & 5.0 && 0.1 & 0.5 & 1.0 & 2.0 & 5.0 \\
\hline
\multicolumn{18}{l}{\emph{Panel A : i.i.d. noise ($\beta = 0$)}}\\
BM        & 1.00 & 1.00 & 1.00 & 1.00 & 1.01 && 1.00 & 1.00 & 1.00 & 1.01 & 1.01 && 1.00 & 1.00 & 1.00 & 1.01 & 1.01 \\
SV        & 1.00 & 1.00 & 1.01 & 1.01 & 1.01 && 1.00 & 1.00 & 1.00 & 1.00 & 1.00 && 1.00 & 1.00 & 1.00 & 1.00 & 1.00 \\
SV2F-LEV  & 1.00 & 1.00 & 1.00 & 1.00 & 1.00 && 1.00 & 1.00 & 1.00 & 1.00 & 1.00 && 1.00 & 1.00 & 1.00 & 1.00 & 1.00 \\
BMJ       & 1.25 & 1.25 & 1.25 & 1.25 & 1.25 && 1.02 & 1.04 & 1.05 & 1.08 & 1.11 && 1.00 & 1.00 & 1.01 & 1.02 & 1.04 \\
BMO       & 1.00 & 1.00 & 1.00 & 1.00 & 1.01 && 1.00 & 1.00 & 1.00 & 1.01 & 1.01 && 1.00 & 1.00 & 1.00 & 1.01 & 1.01 \\
\multicolumn{18}{l}{}\\
\multicolumn{18}{l}{\emph{Panel B : Dependent noise ($\beta = 0.77$)}}\\
BM        & 1.02 & 1.00 & 1.00 & 1.00 & 1.01 && 1.03 & 1.00 & 1.00 & 1.01 & 1.01 && 1.03 & 1.00 & 1.00 & 1.01 & 1.01 \\
SV        & 1.02 & 1.00 & 1.01 & 1.01 & 1.01 && 1.03 & 1.00 & 1.00 & 1.00 & 1.00 && 1.03 & 1.00 & 1.00 & 1.00 & 1.00 \\
SV2F-LEV  & 1.03 & 1.00 & 1.00 & 1.00 & 1.00 && 1.03 & 1.00 & 1.00 & 1.00 & 1.00 && 1.03 & 1.00 & 1.00 & 1.00 & 1.00 \\
BMJ       & 1.26 & 1.25 & 1.25 & 1.25 & 1.25 && 1.05 & 1.04 & 1.06 & 1.08 & 1.11 && 1.03 & 1.00 & 1.01 & 1.02 & 1.04 \\
BMO       & 1.02 & 1.00 & 1.00 & 1.00 & 1.01 && 1.02 & 1.00 & 1.00 & 1.01 & 1.01 && 1.03 & 1.00 & 1.00 & 1.01 & 1.01 \\
\hline
\end{tabular}
\medskip
\begin{scriptsize}
\parbox{\textwidth}{\emph{Note.} The table reports the sample average of the pre-averaging estimators $RV^\ast$, $BV^\ast$, and $BV^\ast_\tau$ across $10,000$ independent simulation runs and for five choices of pre-averaging horizon $\theta = \{0.1; 0.5; 1 ; 2; 5\}$. $RV^\ast$ and $BV^\ast$ are defined in Eqs. \eqref{Eqn:RVast} -- \eqref{Eqn:BVast}, while $BV^\ast_\tau$ is a truncated version of $BV^\ast$. Truncation details are in the simulation study. We simulate Eq. \eqref{Eqn:Yprocess} using five models: (1) a Brownian motion (BM), where $\text{d}X_t = \sigma_t \text{d}W_t$ with $\sigma_t^2 = 0.0391$ (i.e., about 20\% volatility pro anno (p.a.)), (2) a \citet*{heston:93a}-type stochastic volatility (SV) model, (3) a two-factor SV model with leverage (SV2F-LEV), (4) a Brownian motion plus jump (BMJ) model, and (5) a Brownian motion plus outlier (BMO) model. In model BMJ (BMO), a single jump (outlier) is placed at random in the data and it accounts for $20\%$ of total variation. To add microstructure noise we use $u_i = \beta u_{i-1} + \varepsilon_i$, where $\varepsilon_i \overset{\text{i.i.d.}}{ \sim} N(0, \omega^2(1- \beta^2))$. We consider i.i.d. noise ($\beta = 0$) and dependent noise ($\beta = 0.77$). $\omega^{2}$ is determined via the noise ratio parameter $\gamma = \sqrt{N \omega^{2} / \int_0^1 \sigma_s^2 \text{d}s}$. We set $\gamma = 0.50$ and let the total number of observations be $N = 40,000$, which matches our empirical data. In each simulation, all estimates are normalized by $\int_0^1 \sigma_s^2 \text{d}s$. An unbiased estimator should therefore be close to one on average. The only exception is model BMJ, where $E(RV^\ast) = 1.25$.}
\end{scriptsize}
\end{center}
\end{table}

Table \ref{Table:RVrobust} reports the mean of the pre-averaging estimators $RV^\ast$ and $BV^\ast$, normalized by the diffusive variation component $\int_0^1 \sigma_s^2 \text{d}s$, for five choices of pre-averaging horizon $\theta = \{0.1; 0.5; 1 ; 2; 5\}$.\footnote{\citet*{christensen-oomen-podolskij:13b} provide additional simulation results for sample sizes of $N=10,000$ and $N=100,000$.} As desired, the (normalized) estimates of $RV^\ast$ and $BV^\ast$ are close to unity for all models without jumps. For the BMJ model, $E(RV^\ast) = 1.25 \int_0^1 \sigma_s^2 \text{d}s$ and we note that this is closely matched by the simulated values irrespective of the choice of $\theta$. None of these results change qualitatively for the dependent noise scenario. The jump-robust $BV^\ast$ estimates, however, are biased upward and this bias increases with $\theta$.

To reduce the finite sample bias in the ordinary BV estimator defined by Eq. \eqref{Eqn:BVn}, \citet*{corsi-pirino-reno:10a} propose to augment the estimator with a threshold filter. The idea is to remove any large jumps first (by removing returns that exceed some predefined threshold) before calculating BV. It is natural to use such an approach for the pre-averaged BV estimator as well. To set a good threshold in our setup, under a scaled Brownian motion with i.i.d. noise, as in model BM, the asymptotic distribution of $r_{i,K}^\ast$ is given by:
\begin{equation}
N^{1/4} r_{i,K}^\ast \mid \mathcal{F}_{i / N} \overset{a}{\sim} N \left( 0, \frac{\theta \sigma^2}{12} + \frac{\omega^2}{\theta} \right).
\end{equation}
Thus, we can define a threshold by taking
\begin{equation}
\tau = \frac{q_{\alpha}}{N^{\varpi}} \sqrt{\psi_K \theta \sigma^2 +  \frac{\omega^2}{\theta}} ,
\end{equation}
where $q_{\alpha}$ is the $\alpha$-quantile from the $N(0,1)$ distribution and $\varpi \in (0,0.25)$.

To implement the threshold, we set $\alpha = 0.999$, $\varpi = 0.20$, and replace $\omega^2$ and $\sigma^2$ by $\widehat{\omega}^2_{\text{AC}}$ and $BV^\ast$, respectively. Because the returns are calculated on pre-averaged prices, the data truncation needs some care. Simply removing all observations in which $|r_{i,K}^\ast|> \tau$ would be inefficient as a single jump in the tick data inflates a whole sequence of returns due to the averaging. As an alternative, we adopt a simple algorithm in which for every breach of $\tau$ we compute noisy tick returns from all the prices used in the pre-averaging calculations for that sequence of breaches and we then remove the largest one. After removing the largest noisy tick return, we compute a reconnected price series and repeat the pre-cleaning until all extreme observations are gone. $BV^\ast_\tau$ is then calculated from the price path of the remaining data. A graphical illustration of the procedure can be found in the online Appendix.

The third panel in Table \ref{Table:RVrobust} shows that this thresholding procedure works well and that $BV^\ast_\tau$ largely removes the finite sample bias. Regarding the choice of $\theta$, setting it to large values lowers the effective sampling frequency and slowly reintroduces the finite sample bias. However, setting it to small values exposes it to the dependent noise that operates at high sampling frequencies. Given these trade-offs, we settle on $\theta \in [0.5,2]$ as a reasonable choice.

\section{The role of the jump component in currency and US equity data}\label{Sec:Empirical}

In this section, we use the pre-averaged RV and BV measures to provide an in-depth study into the role of the jump component for a representative set of US equity and foreign exchange rate data. To the best of our knowledge this is the first study to take a comprehensive look at the magnitude of the jump variation component as measured from noisy tick data sampled at ultra-high frequencies. We also report results based on five- and 15-minute sampling frequencies as these are widely used in the literature and provide a natural reference point.

\subsection{Data description}

We have available a large set of tick data covering a representative set of foreign exchange rate, large-cap US equity, and US equity index data. The sample period is from January 2007 through March 2011 (or 1,170 trading days) and includes several episodes of exceptional turbulence such as the global housing and credit crisis, the S\&P 500 flash-crash, the European sovereign debt crisis, the bail-out of Greece, and the Japanese earthquake.

For the equities, we consider all 30 DJIA index constituents as of October 2010, as well as two market-wide indexes traded as highly liquid exchange traded funds (ETFs), namely QQQ tracking the Nasdaq 100 index and SPY tracking the S\&P 500 index. The latter is included because it is used by many other studies (see Table \ref{Table:Literature}) and thus provides a good benchmark. The equity data are extracted from the NYSE Trade and Quote (TAQ) database and include both quote and trade data with millisecond precision timestamps allowing for a detailed view of the price evolution. We restrict attention to the official trading hours from 9:30 to 16:00 local New York time.

For the foreign currency data, we have the three major rates of euro, Japanese yen, and Swiss franc all traded against the US dollar, i.e., EURUSD, USDJPY, and USDCHF, respectively. The data are taken from the EBSLive feed that also provides both trade and quote data with millisecond timestamps. We restrict attention to the most liquid London and New York trading hours from 7:00 to 19:00 Greenwich Mean Time.

\begin{table}
\setlength{\tabcolsep}{0.30cm}
\caption{Jump variation estimates for equity and foreign exchange rate high-frequency data.}
\label{Table:Descriptive}
\begin{center}
\begin{scriptsize}
\begin{tabular}{lrrrrrrrrrrrrrrrr}
\hline
&       &               && \multicolumn{3}{c}{Tick frequency} && \multicolumn{3}{c}{Five-minute frequency} && \multicolumn{3}{c}{15-minute frequency}\\
\cline{5-7} \cline{9-11} \cline{13-15}
& $N$~~~  & $\widehat{\gamma}$~~~       && $RV^\ast$ & $BV^\ast_\tau$ & JV && $RV$ & $BV$ & JV && $RV$ & $BV$ & JV \\
\hline
\multicolumn{6}{l}{\emph{Panel A : Equity indexes}}\\
QQQ             &  48,333 & 0.37 && 22.8 & 23.0 & -1.7 && 22.8 & 22.4 &  3.7 && 22.7 & 21.6 &  8.9\\
SPY             & 113,770 & 0.34 && 21.1 & 21.2 & -0.8 && 20.9 & 20.5 &  4.1 && 20.7 & 19.8 &  8.1\\
\emph{Average}  &  81,052 & 0.36 && 22.0 & 22.1 & -1.3 && 21.9 & 21.5 &  3.9 && 21.7 & 20.7 &  8.5\\
\multicolumn{6}{l}{}\\
\multicolumn{6}{l}{\emph{Panel B : Foreign exchange rates}}\\
EURUSD          &  17,703 & 0.36 &&  9.9 &  9.9 & -0.7 &&  9.8 &  9.4 &  7.1 &&  9.6 &  9.2 &  7.5\\
USDCHF          &   4,629 & 0.19 && 10.4 & 10.3 &  1.8 && 10.7 & 10.2 &  8.8 && 10.4 &  9.8 & 10.3\\
USDJPY          &   9,079 & 0.31 && 10.8 & 10.8 &  0.1 && 10.8 & 10.3 &  9.0 && 10.4 &  9.8 &  9.7\\
\emph{Average}  &  10,470 & 0.29 && 10.4 & 10.3 &  0.4 && 10.4 & 10.0 &  8.3 && 10.1 &  9.6 &  9.2\\
\multicolumn{6}{l}{}\\
\multicolumn{6}{l}{\emph{Panel C : DJIA constituents}}\\
AA              &  31,361 & 0.50 && 48.4 & 48.0 &  1.5 && 48.9 & 47.6 &  5.2 && 47.2 & 44.8 & 10.0\\
AXP             &  31,371 & 0.27 && 46.7 & 46.4 &  1.3 && 46.6 & 45.0 &  6.9 && 45.8 & 43.2 & 11.2\\
BA              &  21,122 & 0.26 && 30.5 & 30.3 &  1.5 && 30.6 & 29.5 &  7.0 && 30.4 & 28.5 & 12.0\\
BAC             &  82,628 & 0.81 && 59.0 & 59.3 & -1.0 && 58.9 & 56.2 &  8.9 && 58.2 & 54.7 & 11.7\\
CAT             &  27,422 & 0.26 && 36.7 & 36.6 &  0.2 && 36.7 & 35.7 &  5.4 && 36.3 & 34.7 &  8.8\\
CSCO            &  54,129 & 0.63 && 31.9 & 31.7 &  1.3 && 32.2 & 31.2 &  5.9 && 31.4 & 30.0 &  9.1\\
CVX             &  35,064 & 0.26 && 30.0 & 30.2 & -1.1 && 29.7 & 29.0 &  4.6 && 28.7 & 27.4 &  8.9\\
DD              &  20,302 & 0.28 && 32.3 & 32.2 &  0.6 && 32.8 & 31.7 &  6.8 && 31.9 & 30.3 &  9.9\\
DIS             &  24,914 & 0.36 && 29.9 & 29.6 &  1.7 && 30.3 & 29.1 &  7.6 && 29.6 & 28.0 & 10.5\\
GE              &  56,193 & 0.76 && 38.2 & 38.1 &  0.9 && 38.2 & 36.6 &  8.0 && 38.1 & 35.4 & 13.7\\
HD              &  30,111 & 0.37 && 34.5 & 34.4 &  1.0 && 34.4 & 33.3 &  6.4 && 33.6 & 31.6 & 12.0\\
HPQ             &  35,631 & 0.33 && 29.8 & 29.4 &  2.8 && 29.9 & 28.9 &  6.3 && 29.1 & 27.2 & 12.5\\
IBM             &  25,483 & 0.27 && 26.0 & 25.9 &  0.8 && 25.4 & 24.6 &  6.3 && 24.7 & 23.2 & 11.9\\
INTC            &  54,657 & 0.64 && 32.2 & 32.1 &  0.6 && 32.4 & 31.2 &  7.4 && 31.3 & 29.3 & 12.4\\
JNJ             &  29,929 & 0.37 && 18.8 & 18.6 &  1.4 && 18.9 & 18.1 &  8.2 && 18.3 & 16.9 & 14.0\\
JPM             &  67,858 & 0.33 && 48.7 & 48.8 & -0.4 && 48.5 & 47.2 &  5.3 && 47.8 & 45.2 & 10.7\\
KFT             &  20,494 & 0.44 && 22.8 & 22.3 &  4.2 && 23.3 & 22.0 & 10.6 && 22.2 & 20.5 & 14.3\\
KO              &  24,862 & 0.33 && 21.2 & 21.0 &  1.3 && 21.5 & 20.3 & 10.9 && 20.5 & 18.9 & 14.9\\
MCD             &  23,610 & 0.29 && 23.7 & 23.6 &  0.9 && 23.8 & 22.7 &  8.8 && 22.7 & 21.2 & 12.6\\
MMM             &  17,005 & 0.27 && 26.2 & 26.0 &  1.4 && 25.9 & 25.1 &  6.6 && 24.7 & 23.3 & 11.5\\
MRK             &  30,107 & 0.37 && 30.1 & 29.7 &  2.7 && 31.4 & 29.8 & 10.0 && 29.8 & 27.5 & 14.4\\
MSFT            &  59,937 & 0.58 && 28.3 & 28.2 &  0.7 && 28.3 & 27.4 &  5.8 && 27.5 & 25.8 & 12.0\\
PFE             &  37,409 & 0.81 && 26.5 & 26.2 &  2.2 && 26.7 & 25.6 &  8.5 && 25.5 & 24.0 & 11.1\\
PG              &  30,271 & 0.33 && 23.0 & 22.5 &  3.6 && 21.2 & 20.4 &  8.1 && 20.1 & 18.8 & 13.1\\
T               &  38,131 & 0.55 && 29.0 & 28.6 &  2.4 && 29.8 & 28.9 &  5.8 && 28.0 & 26.4 & 10.8\\
TRV             &  15,961 & 0.26 && 37.6 & 36.6 &  5.0 && 38.9 & 36.9 &  9.9 && 37.4 & 33.8 & 18.6\\
UTX             &  18,727 & 0.26 && 27.1 & 26.9 &  1.2 && 26.8 & 25.9 &  6.5 && 25.9 & 24.6 & 10.1\\
VZ              &  30,053 & 0.45 && 27.5 & 27.3 &  1.7 && 28.1 & 27.0 &  7.7 && 27.1 & 25.5 & 11.5\\
WMT             &  35,785 & 0.35 && 23.5 & 23.4 &  0.3 && 23.6 & 22.7 &  7.7 && 22.6 & 21.2 & 11.9\\
XOM             &  55,702 & 0.30 && 28.1 & 28.2 & -0.7 && 27.8 & 27.0 &  5.2 && 26.8 & 25.5 &  9.8\\
\emph{Average}  &  35,541 & 0.41 && 31.6 & 31.4 &  1.3 && 31.7 & 30.6 &  7.3 && 30.8 & 28.9 & 11.9\\
\hline
\end{tabular}
\end{scriptsize}
\medskip
\begin{scriptsize}
\parbox{\textwidth}{\emph{Note.} This table reports the average daily quadratic variation (QV) and integrated variance (IV) estimates together with the implied jump variation (JV) at different sampling frequencies for the sample period January 2007 through March 2011. At the tick (five- and 15-minute) frequency, QV is estimated by $RV^\ast$ in Eq. \eqref{Eqn:RVast} ($RV$ in Eq. \eqref{Eqn:RVn}) and IV is estimated by $BV^\ast_\tau$ ($BV$ in Eq. \eqref{Eqn:BVn}). $BV^\ast_\tau$ is a truncated version of $BV^\ast$ in Eq. \eqref{Eqn:BVast}. Truncation details are in the simulation study. The JV from Eq. \eqref{Eqn:JV} is expressed in percentages, while the realized variation measures are expressed as percentage annualized standard deviations. $N$ is the average daily number of observations (after millisecond aggregation) and $\widehat{\gamma}$ the estimated noise ratio given by Eq. \eqref{Eqn:omega}.}
\end{scriptsize}
\end{center}
\end{table}

While the foreign exchange rate data are exceptionally clean, the equity data require some pre-filtering. For the quote data we adopt the routine proposed in \citet*[][hereafter BNHLS]{barndorff-nielsen-hansen-lunde-shephard:09a}.\footnote{In particular, we remove quotes with irregular quote condition, with bid or ask equal to zero, with negative spreads, with spreads larger than ten times the median spread for the day, and when the mid-quote is outside five mean absolute deviations from a centered mean of 50 observations (excluding the quote under investigation).} That paper also suggests removing trades that when synchronized to the quote data are well above (below) the most recent best offer (bid) price. While this procedure works well in most cases, we did come across a few scenarios in which it leads to excessive removal of data. Over a volatile episode, the trade data appear to lag the quote data significantly and, as a direct consequence, the BNHLS filter removes most data (see Figure \ref{Figure:tradefilter} in \ref{Sec:ApndxGraphs}). To address such a situation, we introduce a simple backward-forward matching (BFM) algorithm that seeks to match trades that fall outside the bid-offer to quotes over a short forward-looking window and, if not successful, over a longer backward-looking window. We set the forward-looking window to one second as it merely aims to account for small inaccuracies in the timestamps and relative ordering of the quotes and trades. The backward-looking window is set to 20 minutes so that block trades qualifying for delayed trade reporting are covered by the filter procedure. Table \ref{Table:filter} in \ref{Sec:ApndxGraphs} reports detailed summary statistics for the BFM filter. After cleaning, we have a total sample size of well over four billion observations.

The results we present are based on trade data sampled at every tick, i.e., $N$ is equal to the number of trades for a given security on a particular day.\footnote{Results based on quote data are qualitatively very similar and are not reported to conserve space.} This sampling is standard in the literature and tends to homogenize the data in the common scenario in which the frequency of trade arrivals is correlated with volatility, i.e., heteroskedasticity is reduced through time deformation \citep*{ane-geman:00a}. This has the advantage---as made explicit in Eq. \eqref{Eqn:loglog} below---that it reduces the finite sample bias in the $BV^\ast$ measure. Furthermore, the consistency result in Proposition \ref{Thm:CLT} is robust to irregular sampling (see \ref{Sec:ApndxSampling}).

\begin{figure}[ht!]
\begin{center}
\caption{$\theta$-signature plot of average annualized volatility and jump proportion.}
\label{Figure:Theta_Sigplot} 
\begin{tabular}{cc}
\footnotesize{Panel A: Equity data} & \footnotesize{Panel B: Foreign exchange rate data} \\
\includegraphics[height=8cm,width=0.48\textwidth]{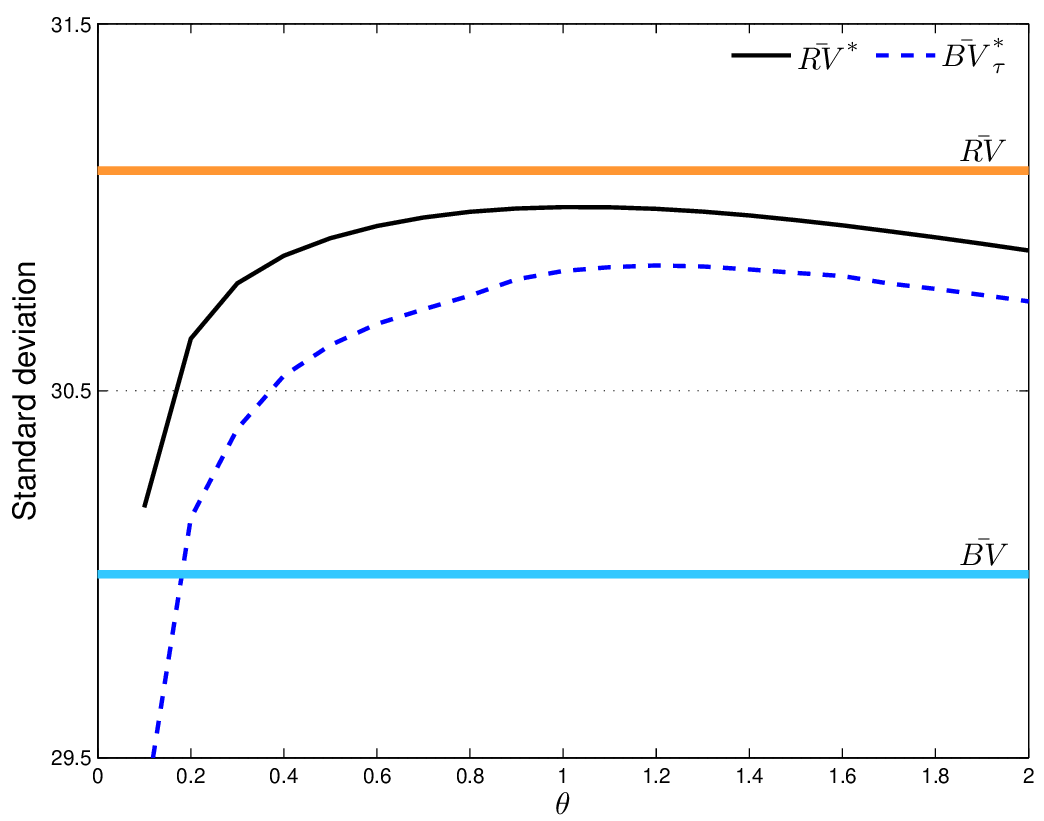} &
\includegraphics[height=8cm,width=0.48\textwidth]{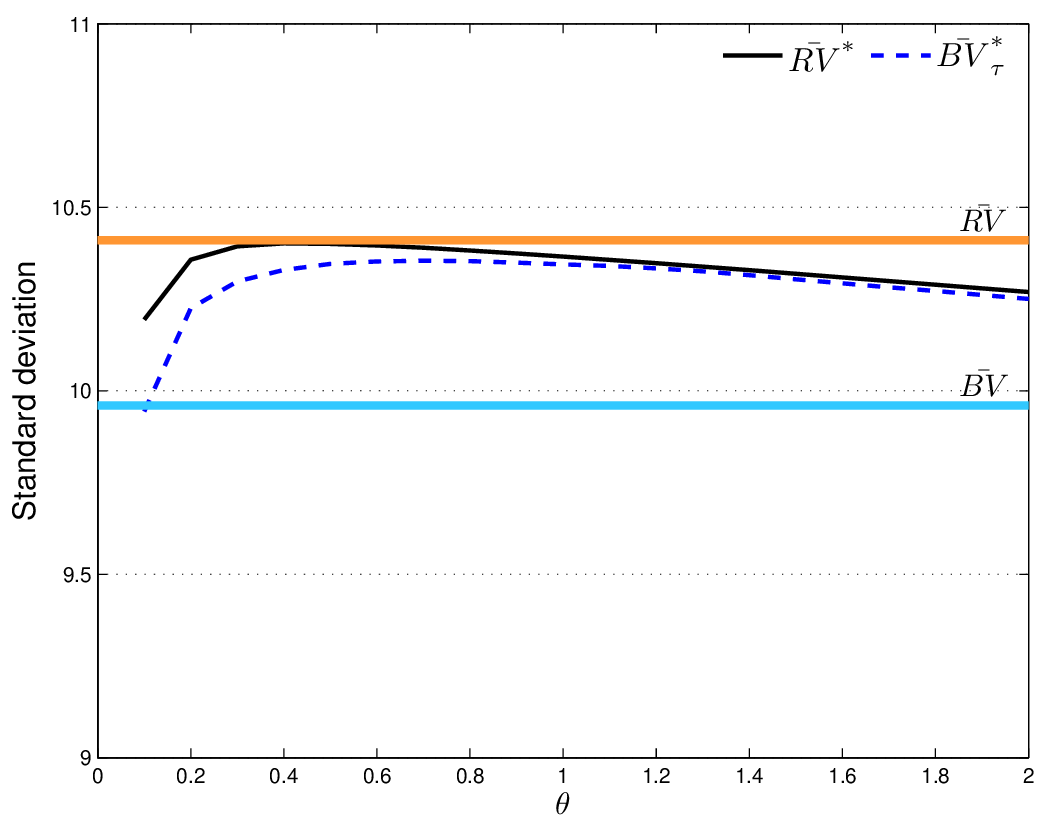} \\
\footnotesize{Panel C: Equity data} & \footnotesize{Panel D: Foreign exchange rate data} \\
\includegraphics[height=8cm,width=0.48\textwidth]{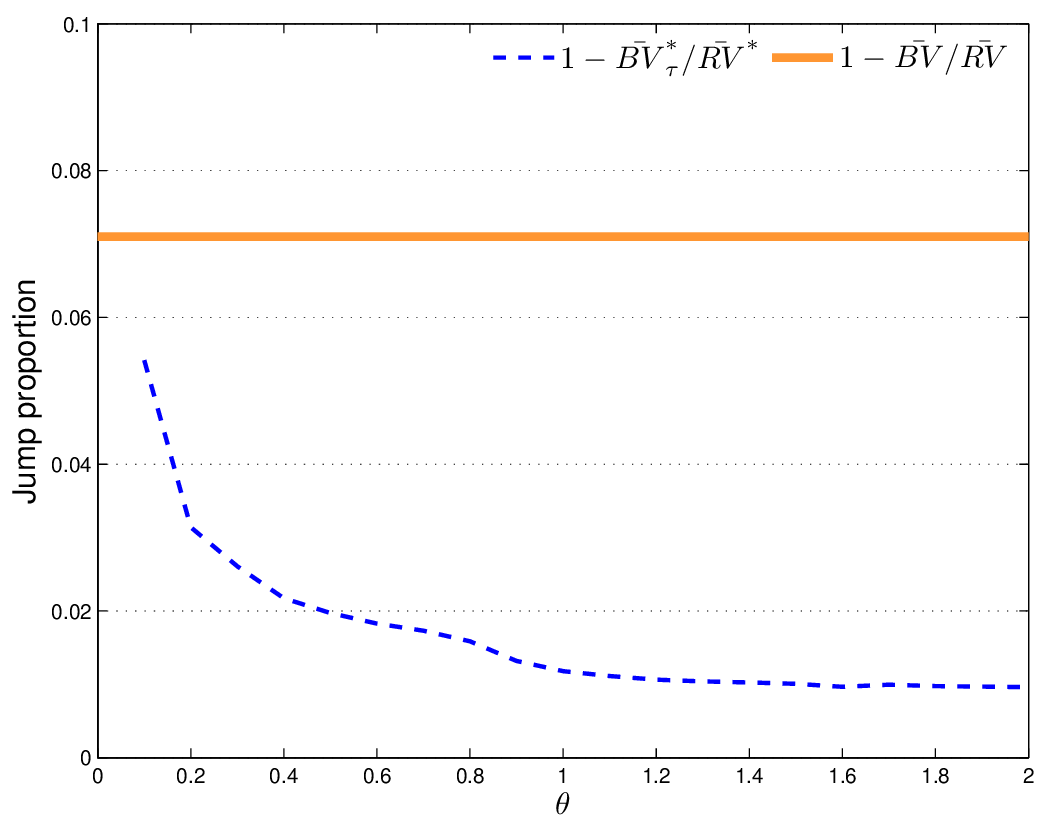} &
\includegraphics[height=8cm,width=0.48\textwidth]{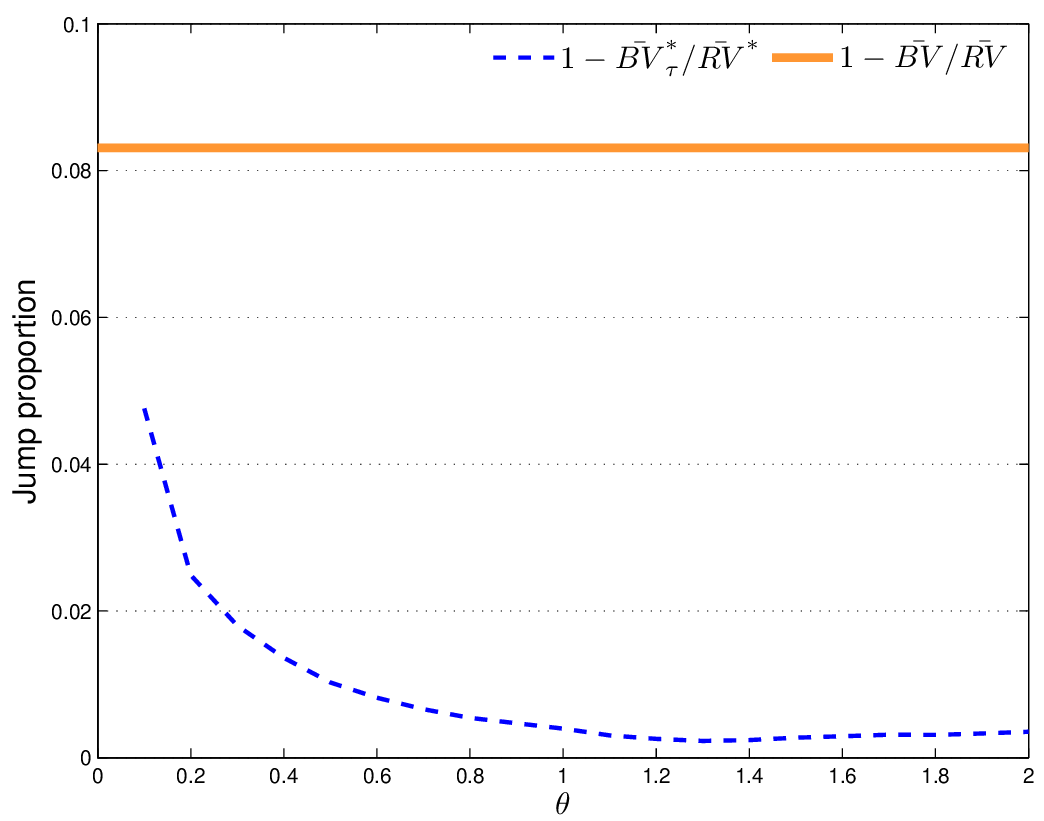} \\
\end{tabular}
\begin{scriptsize}
\parbox{\textwidth}{\emph{Note.} Panels A and B report the annualized volatility of the noise-robust $RV^\ast$ defined in Eq. \eqref{Eqn:RVast} and $BV^\ast_\tau$, which is a truncated version of $BV^\ast$ from Eq. \eqref{Eqn:BVast}. Truncation details are in the simulation study. We first average the time series of estimates for individual assets in our sample and then average over the cross-section of stocks and foreign exchange rate pairs. To explore the sensitivity of our results to pre-averaging, we compute the estimates for many different values of pre-averaging horizon $\theta$. As a comparison, the corresponding average values of the realized variance and bi-power variation from Eqs. \eqref{Eqn:RVn} -- \eqref{Eqn:BVn} based on five-minute data are indicated with a solid thick line. In Panels C and D, the numbers are converted into an estimate of the jump proportion using Eq. \eqref{Eqn:JV}.}
\end{scriptsize}
\end{center}
\end{figure}

\subsection{The magnitude of the jump component}

To start our investigation into the magnitude of the jump component, Table \ref{Table:Descriptive} reports the average daily quadratic variation and integrated variance, together with the implied jump variation, as estimated by $RV^\ast$ and $BV^\ast_\tau$ from tick data. For comparison, we also report the results using RV and BV applied to the conventional five- and 15-minute sampled data. The single most important observation to be made is that the jump component inferred from tick data is small. Much smaller than what is computed from the five- and 15-minute data and much smaller than what is reported by the extant literature as summarized in Table \ref{Table:Literature}. The cross-sectional averages by category convey the message most clearly: Using tick data, we find that jumps account for a mere 1.3\% of total return variation for the DJIA constituents, 0.4\% for the currency pairs, and -1.3\% for the equity indexes (under the null of no jumps, the distribution of $RV^\ast-BV^\ast_\tau$ is symmetric around zero, so small negative numbers can occur due to measurement error). The corresponding figures for the widely used five- (and 15-) minute data are in line with the literature at a much more substantial 7.3\% (11.9\%), 8.3\% (9.2\%), and 3.9\% (8.5\%) for individual equity, currency, and equity index data, respectively. Also, for each asset there is a strict ordering in JV by sampling frequency in which the tick data give the lowest (and arguably most accurate) value and the 15-minute data the highest.

The results in Table \ref{Table:Descriptive} use a pre-averaging parameter of $\theta=1$. Regarding the robustness of our results, Figure \ref{Figure:Theta_Sigplot} draws $RV^\ast$ and $BV^\ast_\tau$ as a function of $\theta$ together with the implied jump proportion for equity and currency data separately. This so-called $\theta$-signature plot shows that for values of above $0.5$ (the lower bound also identified in the simulation study above), the results are remarkably insensitive to the choice of $\theta$. Only when $\theta$ is set to very small values, the pre-averaging is insufficient to mitigate the impact of microstructure noise and a pronounced downward bias results [see \citet*{hautsch-podolskij:13a} for a more detailed study].

\begin{figure}[ht!]
\begin{center}
\caption{Jump proportion, quarter-by-quarter.}
\label{Figure:JpmQuarter}
\begin{tabular}{cc}
\footnotesize{Panel A: Equity data} & \footnotesize{Panel B: Foreign exchange rate data} \\
\includegraphics[height=8cm,width=0.48\textwidth]{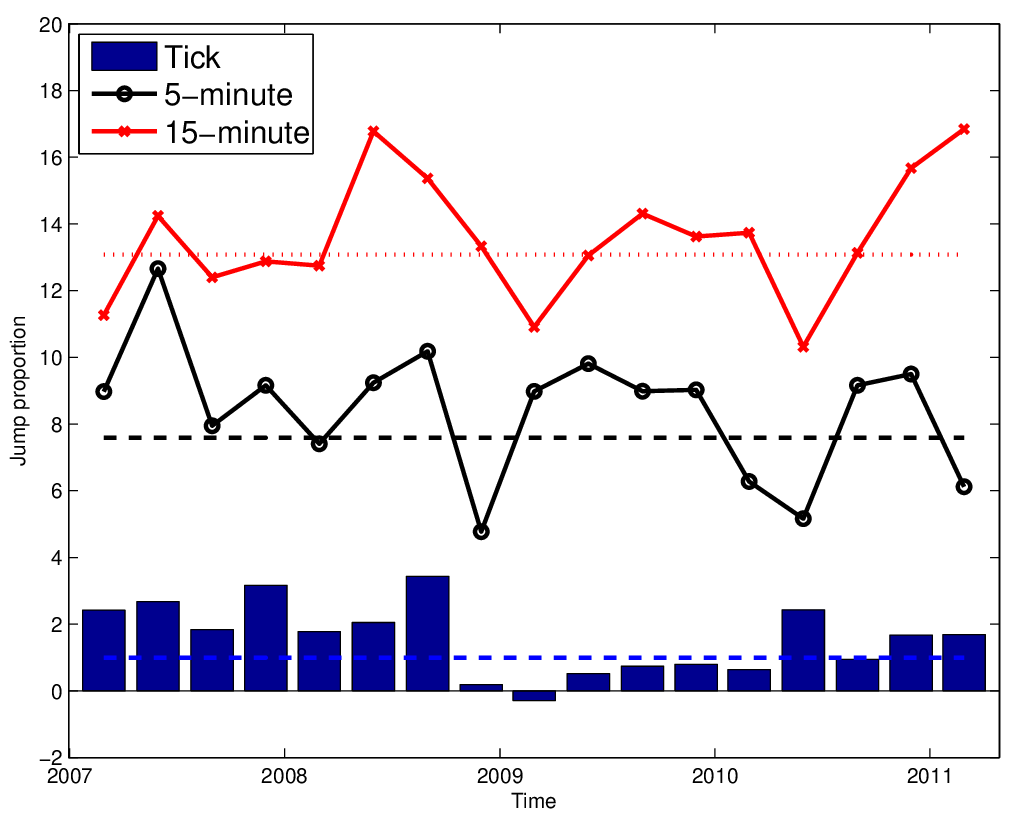} &
\includegraphics[height=8cm,width=0.48\textwidth]{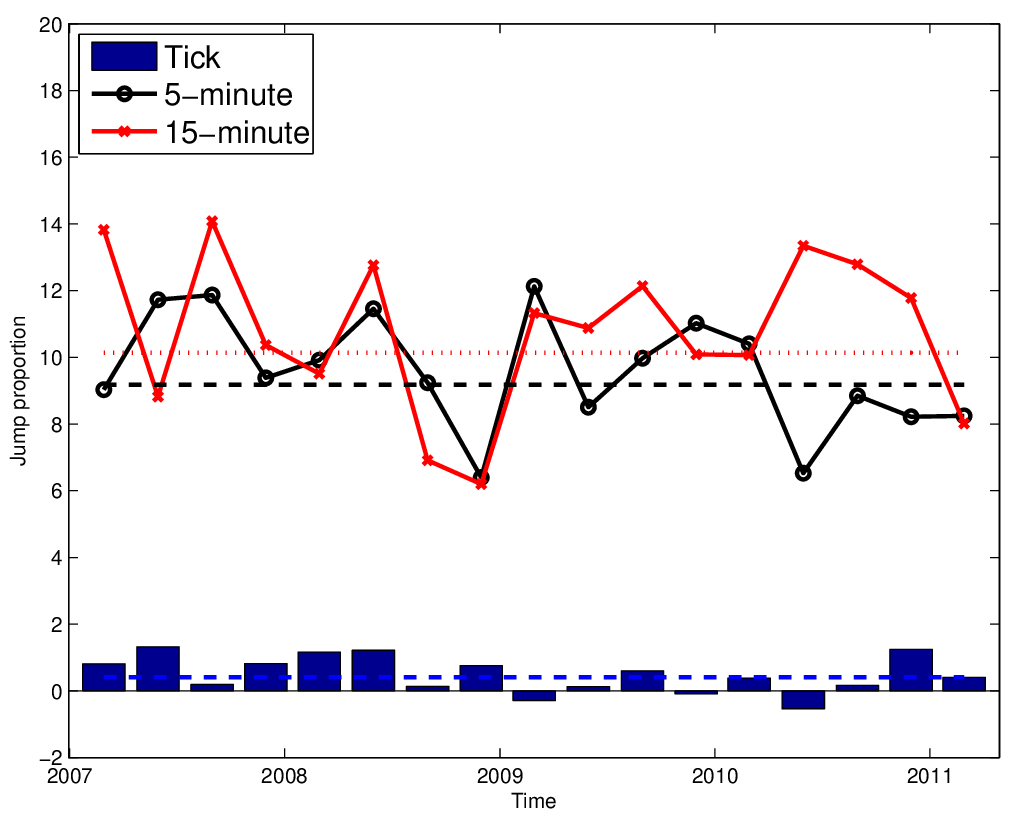} \\
\end{tabular}
\begin{scriptsize}
\parbox{\textwidth}{\emph{Note.} This figure draws the quarterly jump variation in percent as implied by tick data (tick) using the noise-robust $RV^\ast$ defined in Eq. \eqref{Eqn:RVast} and $BV^\ast_\tau$, which is a truncated version of $BV^\ast$ from Eq. \eqref{Eqn:BVast}. Truncation details are in the simulation study. We also report the jump variation as measured by the $RV$ and $BV$ from Eqs. \eqref{Eqn:RVn} -- \eqref{Eqn:BVn} based on five- and 15-minute data. The dashed line is the average across the full sample 2007:Q1 -- 2011:Q1.}
\end{scriptsize}
\end{center}
\end{figure}

Next, regarding the robustness of our finding over time, Figure \ref{Figure:JpmQuarter} plots the quarterly jump variation estimates over the full four-year sample period. Again, the results appear stable and the JV calculated from tick data is small and consistently below the JV estimates calculated from lower-frequency data.

To conclude, an alternative way to summarize and visualize our main finding is via a scatter plot of daily estimates of the integrated variance versus the quadratic variation (Figure \ref{Figure:OlsPlot}). In the absence of jumps, the integrated variance equals the quadratic variation and the estimates are expected to center around a 45-degree line through the origin. Panels A and B show the results for the $RV^\ast$ and $BV_\tau^\ast$ estimates based on tick data. A standard linear regression indicates that the slope coefficient is insignificantly different from one for both the equity and the currency data. By comparison, Panels C and D report the results based on five-minute data. Consistent with the results in Table \ref{Table:Descriptive}, the slope coefficient is significantly smaller than one.

\begin{figure}[ht!]
\begin{center}
\caption{Regression analysis of bi-power variation against realized variance.}
\label{Figure:OlsPlot}
\begin{tabular}{cc}
\small{Panel A: Equity tick data} & \small{Panel B: Foreign exchange rate tick data} \\
\small{$BV_\tau^\ast = \underset{(0.442)}{0.012} + \underset{(-1.737)}{0.987}RV^\ast$} &
\small{$BV_\tau^\ast = \underset{(0.132)}{0.001} + \underset{(-0.354)}{0.994}RV^\ast$} \\
\includegraphics[height=8cm,width=0.48\textwidth]{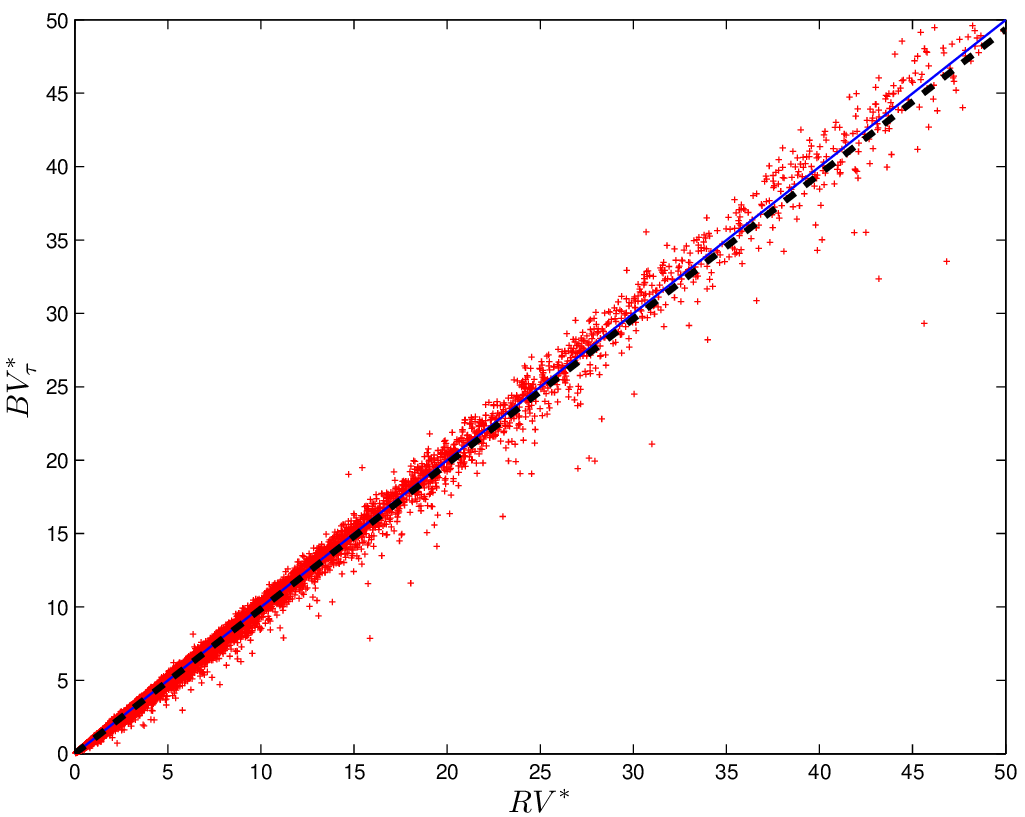} &
\includegraphics[height=8cm,width=0.48\textwidth]{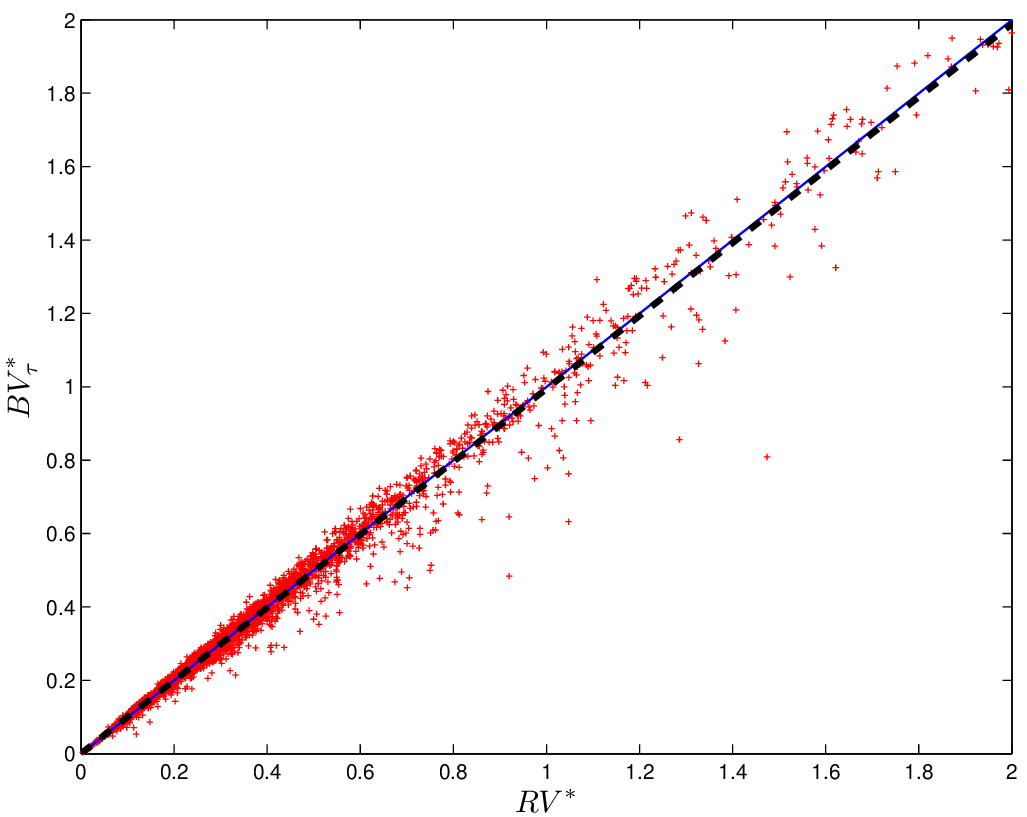} \\
\small{Panel C: Equity five-minute data} & \small{Panel D: Foreign exchange rate five-minute data} \\
\small{$BV = \underset{(1.075)}{0.067} + \underset{(-5.194)}{0.913}RV$} &
\small{$BV = \underset{(1.225)}{0.014} + \underset{(-3.498)}{0.890}RV$} \\
\includegraphics[height=8cm,width=0.48\textwidth]{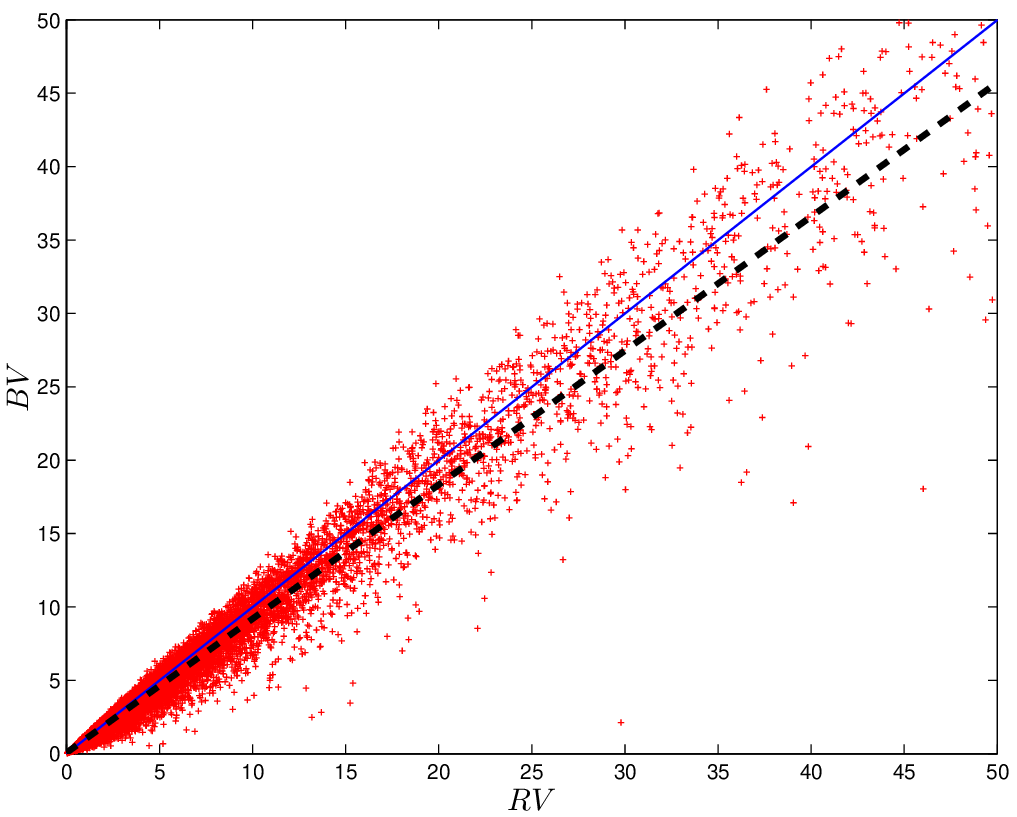} &
\includegraphics[height=8cm,width=0.48\textwidth]{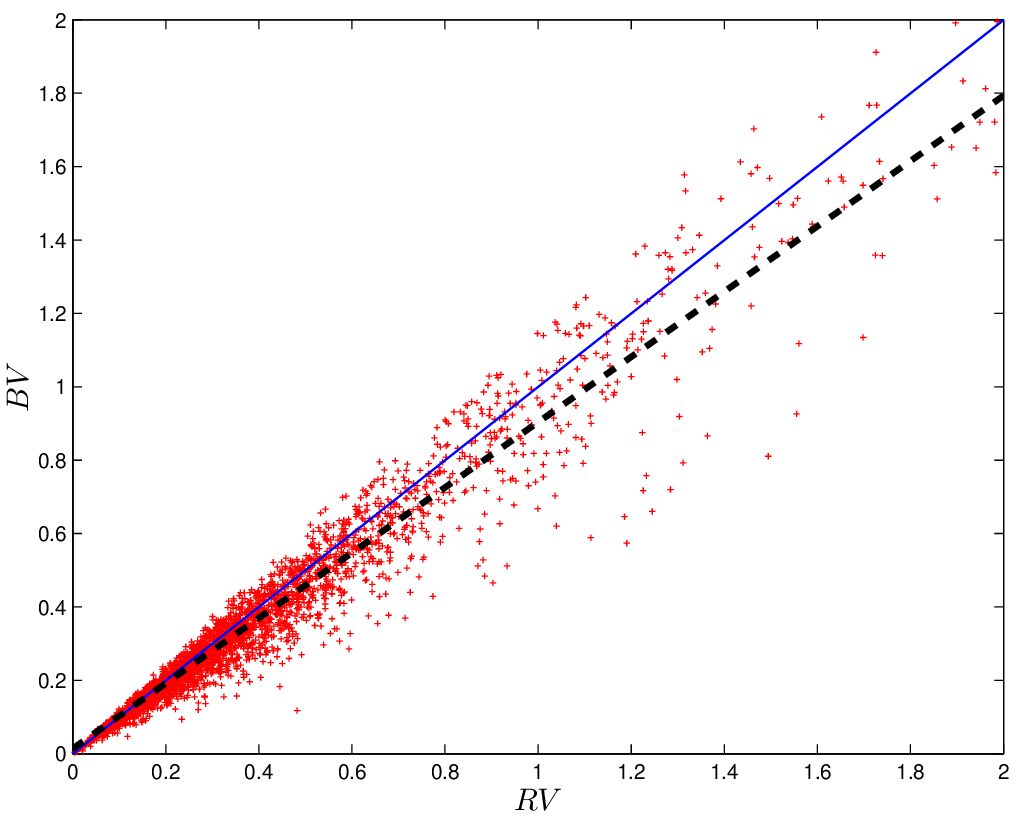} \\
\end{tabular}
\begin{scriptsize}
\parbox{\textwidth}{\emph{Note.} Panels A and B plot the pairwise values of $RV^\ast$ and $BV^\ast_\tau$ based on equity and currency pair tick data, respectively. $RV^\ast$ is defined in Eq. \eqref{Eqn:RVast}, while $BV^\ast_\tau$ is a truncated version of $BV^\ast$ from Eq. \eqref{Eqn:BVast}. Truncation details are in the simulation study. Panels C and D plot the corresponding results for $RV$ and $BV$ from Eqs. \eqref{Eqn:RVn} -- \eqref{Eqn:BVn} based on five-minute data. A plus (\textcolor{red}{\textbf{+}}) indicates an observation for a given day and instrument. We fit a linear regression to the data and conduct a test for the hypothesis that the intercept is zero and the slope is one. The estimated regression equation is reported in subpanels along with $t$-statistics (in parenthesis) based on White's heteroskedasticity-consistent standard errors. The dashed line is based on the estimated regression, while the solid 45-degree line offers a reference point to the null hypothesis.}
\end{scriptsize}
\end{center}
\end{figure}

\subsection{The burst of volatility hypothesis}

The results above point towards a much-reduced role for jumps in explaining the total return variation for equity and currency data. We find that the jump variation computed from tick data is an order of magnitude smaller than the consensus view in the literature that uses five-minute or lower-frequency data. How can these findings be reconciled?

\begin{figure}[ht!]
\begin{center}
\caption{USDJPY earthquake episode: Jump or burst in volatility?}
\label{Figure:Quake}
\begin{tabular}{cc}
\footnotesize{Panel A: Time around earthquake} & \footnotesize{Panel B: Closeup March 16 -- 18} \\
\includegraphics[height=8cm,width=0.48\textwidth]{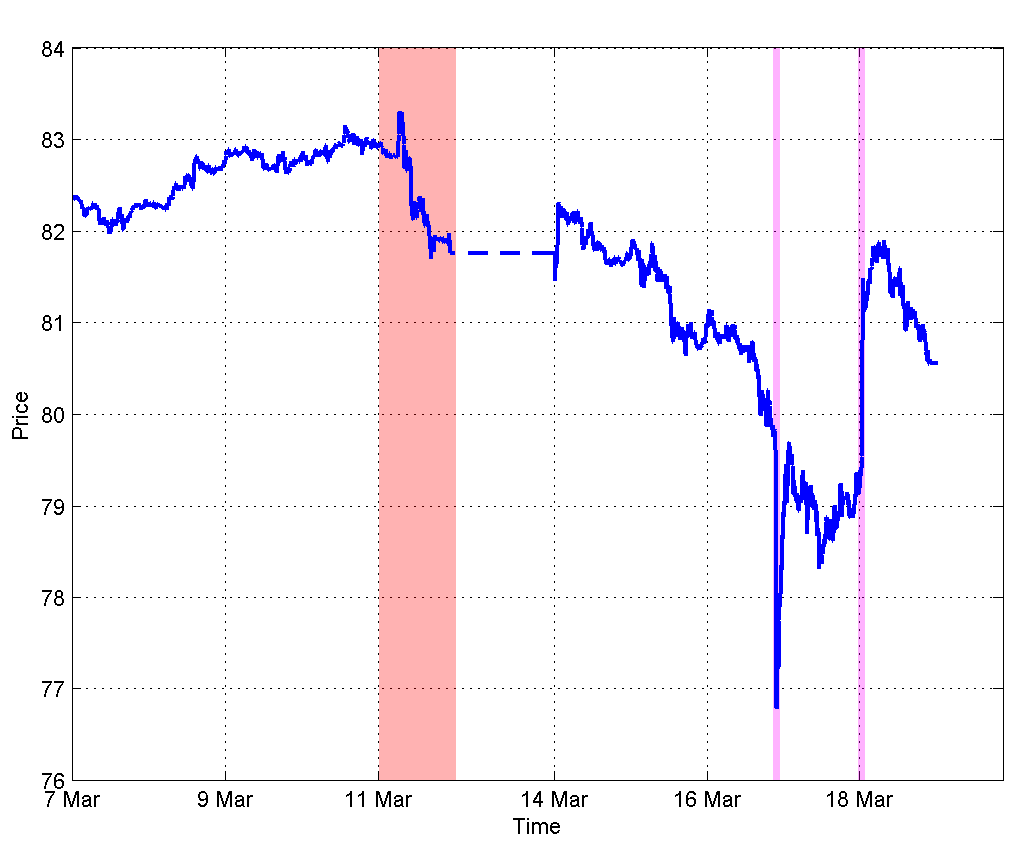} &
\includegraphics[height=8cm,width=0.48\textwidth]{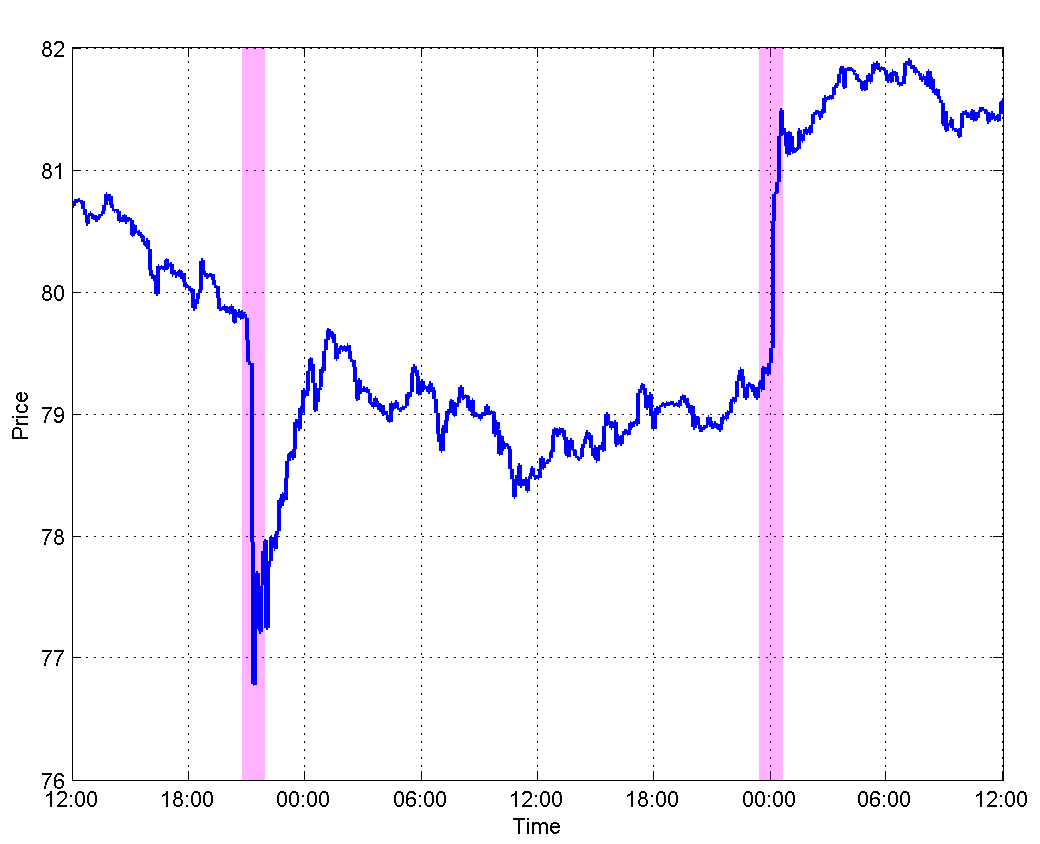} \\
\footnotesize{Panel C: ``Flash-crash'' (trade-by-trade)} & \footnotesize{Panel D: Intervention (trade-by-trade)} \\
\includegraphics[height=8cm,width=0.48\textwidth]{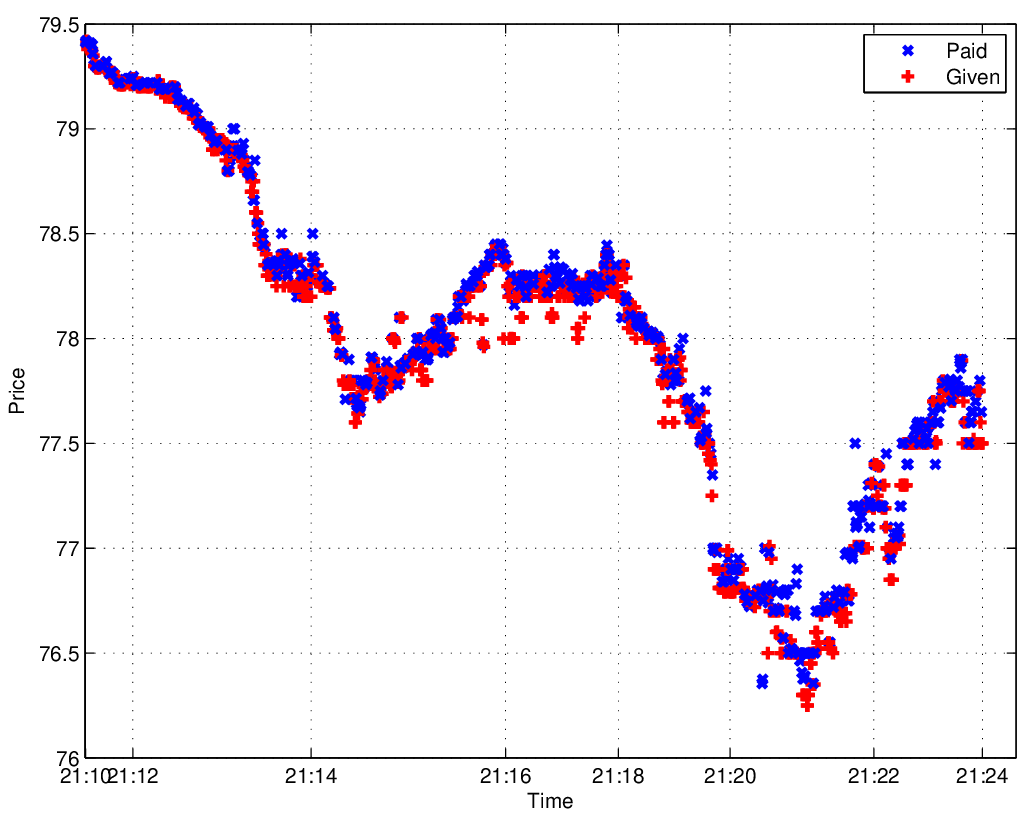} &
\includegraphics[height=8cm,width=0.48\textwidth]{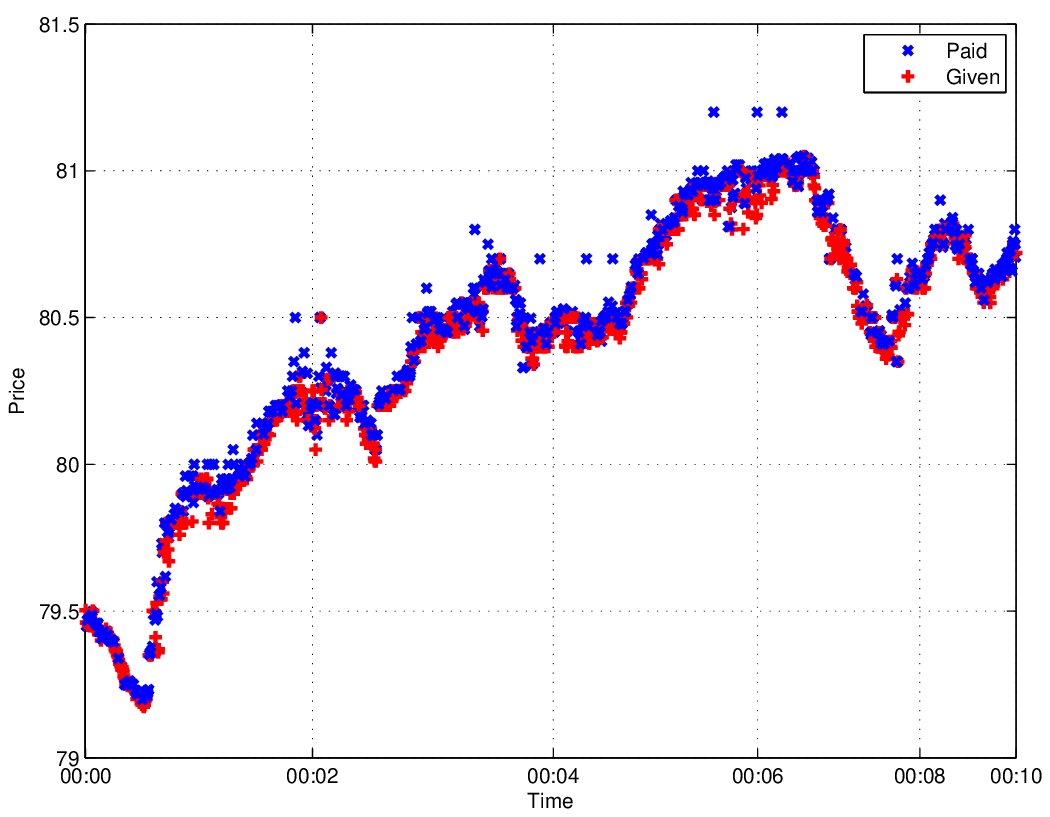} \\
\end{tabular}
\begin{scriptsize}
\parbox{\textwidth}{\emph{Note.} We draw USDJPY spot data for March 7 -- 18, 2011 with time reported in Greenwich Mean Time. The data were collected from EBSLive. The upper panels are based on five-minute data, while the lower panels use tick-by-tick data. The day of the Fukushima earthquake (March 11) is highlighted in Panel A, as are the two five-minute intervals covering the ``flash-crash'' episode on March 16 and the central bank intervention on March 18 that we discuss in the main text. In Panels C and D, ``Paid'' and ``Given'' denote aggressive buy- and sell-orders, respectively.}
\end{scriptsize}
\end{center}
\end{figure}

One explanation is that our ability to distinguish true discrete jumps from continuous diffusive variation diminishes as we lower the sampling frequency $N$. Specifically, a short-lived burst of volatility is likely to be mistakenly identified as a jump with less frequent sampling. The flash-crash episode is a prominent example of such a scenario. The Japanese earthquake provides another for the currency data. Figure \ref{Figure:Quake} plots the evolution of the USDJPY rate in Panels A and B for the period in question. From the five-minute data, two very substantial jumps are evident. One on March 16 when the USDJPY experienced a flash-crash type drop in the rate and another on March 18 following a coordinated central bank intervention aimed at devaluing the JPY. Panels C and D draw the tick data for the respective episodes and tell a very different story in which price jumps remain elusive. 

\subsubsection{A burst of volatility simulation}

To provide support for the burst of volatility hypothesis, we start by considering the simplest of simulation studies. We draw noise-free prices from a scaled Brownian motion, $\text{d}X_t = \sigma_t \text{d}W_t$ for $t \in [0,1]$, where $\sigma_t = 3\sigma^*$ for $t \in [16/32, 17/32]$ and $\sigma_t = \sigma^\ast$ otherwise. $\sigma^\ast$ is fixed at a level corresponding to 40\% in annualized terms. In this scenario, $\sigma_t$ is piecewise constant and increases three-fold over a short interval of the day (equivalent to a 15-minute interval based on an eight-hour trading session). However, because the price path is still continuous, the true JV is zero. We simulate noisy log-prices as above, $Y = X + u$, using i.i.d. noise with a noise ratio parameter of $\gamma = 0.5$ and then round $Y$ to the nearest cent to induce price discreteness, based on a starting price of \$50 in levels. Finally, we calculate RV and BV across a range of sampling frequencies and report the average implied JV over 10,000 independent simulation runs. Panel A of Figure \ref{Figure:JVsign} plots the result on a log-log scale. While the true JV is zero, the estimated JV from low-frequency data is positive and increases with a decrease in the sampling frequency. When applied to noisy data, i.e., $Y$, the effects of microstructure noise invalidate the traditional variation measures at the highest sampling frequencies. For comparison, Panel B plots the corresponding results for the equity data (averaged across all instruments over the full sample) and a pattern similar to that of the simulation emerges thereby providing support for the burst of volatility hypothesis.

\begin{figure}[ht!]
\begin{center}
\caption{Jump variation signature plot.}
\label{Figure:JVsign}
\begin{tabular}{cc}
\small{Panel A: Simulated data (true JV = 0)} & \small{Panel B: Equity data (true JV unknown)} \\
\includegraphics[height=8cm,width=0.48\textwidth]{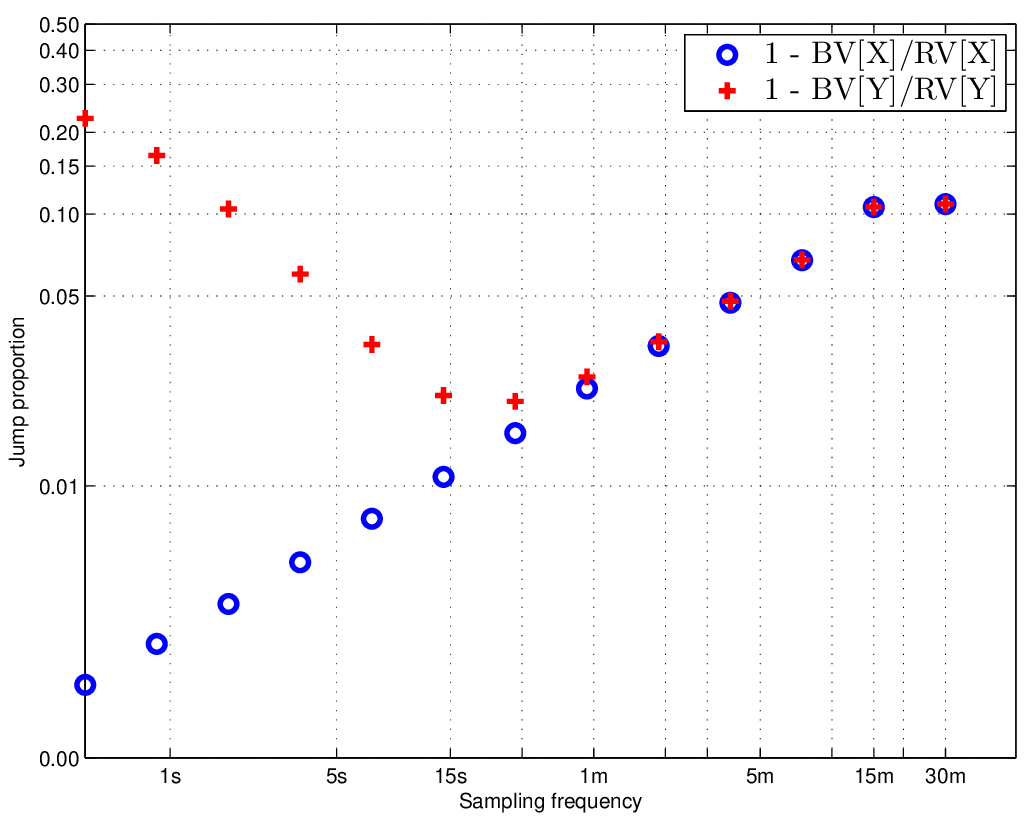} &
\includegraphics[height=8cm,width=0.48\textwidth]{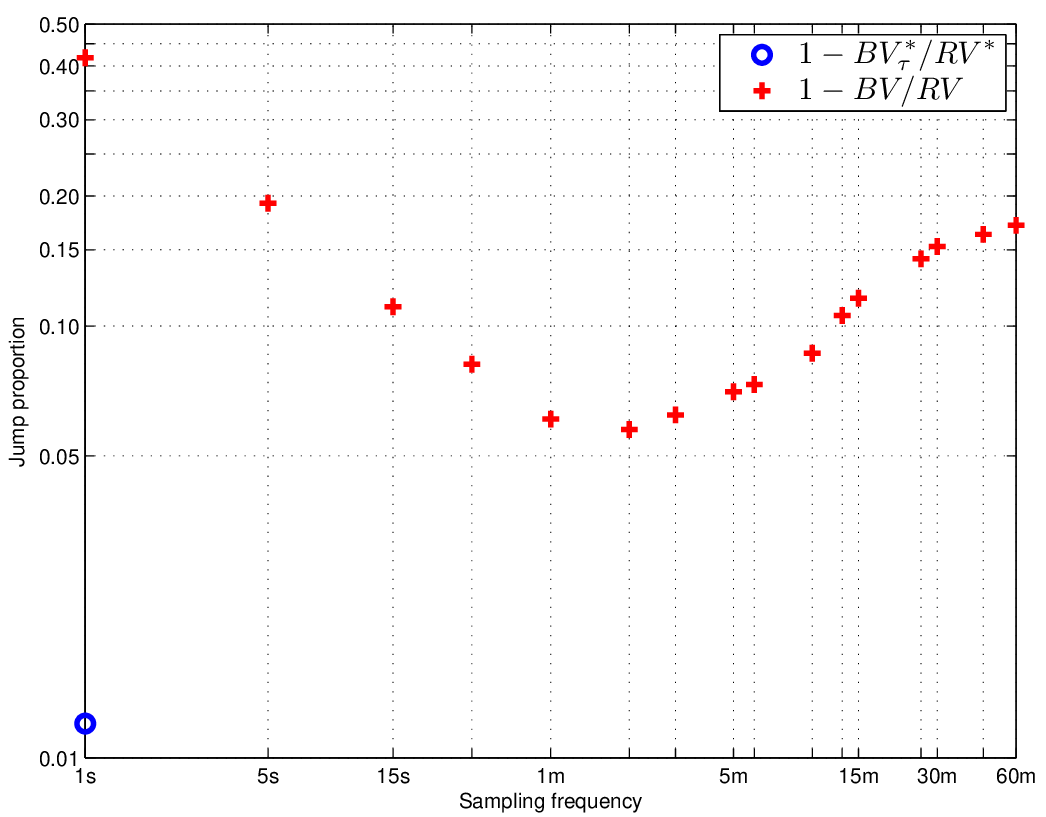} \\
\end{tabular}
\begin{scriptsize}
\parbox{\textwidth}{\emph{Note.} The graph displays average jump proportion estimates as a function of the sampling frequency. The left panel is based on simulated data with a true JV of zero, while the right panel covers the full sample of equity instruments (both cross-sectionally and over time). In both panels, the plus marks (\textcolor{red}{\textbf{+}}) denote estimates based on noisy prices by comparing the standard RV and BV measures defined in Eqs. \eqref{Eqn:RVn} -- \eqref{Eqn:BVn}. These are not valid in the presence of noise. In Panel A, the circles (\textcolor{blue}{\textbf{o}}) denote the JV computed with RV and BV by using the simulated ``unobserved'' noise-free prices. In Panel B, the circle indicates our noise-robust JV estimate based on the real equity tick data using $RV^\ast$ and $BV^\ast_\tau$. $RV^\ast$ is defined in Eq. \eqref{Eqn:RVast}, while $BV^\ast_\tau$ is a truncated version of $BV^\ast$ from Eq. \eqref{Eqn:BVast}. Truncation details are in the simulation study. The latter are both valid in the presence of noise.}
\end{scriptsize}
\end{center}
\end{figure}

The approximately linear relation between the implied JV and the sampling frequency on a log-log scale has some theoretical foundation. In \ref{Sec:BVburst}, we derive the expression for the unconditional bias of BV:
\begin{equation} \label{Eqn:loglog}
E\left(BV - \int_0^1 \sigma_s^2 \text{d}s \right) = -\frac{1}{N} E \left( \frac{1}{12} \int_0^1 \frac{\upsilon_{s}^{2}}{\sigma_{s}^2} \text{d}s \right) + o(N^{-1}),
\end{equation}
where $\upsilon$ is the volatility of volatility. From here, with time-varying volatility, and in the absence of noise, BV is downward biased in finite samples, translating into an inflated JV measure, the effect is stronger with a high volatility of volatility $\upsilon$ or a lowering of the sampling frequency $N$, and the log bias is approximately linear in the log sampling frequency. Consequently, Eq. \eqref{Eqn:loglog} formally illustrates the point that a short-lived burst of volatility can be spuriously attributed to the jump variation component and that its magnitude grows as the sampling frequency of the data is lowered. This line of thought is also consistent with \citet*{ait-sahalia-jacod:09a} who emphasize that jumps can be identified only by increasing the sampling frequency to the limit. This is precisely what our paper does and leads us to argue for a much-reduced role of jumps and by implication, an elevated role for the volatility process.

\subsubsection{Further empirical results} \label{Sec:Femp}

The realised measures discussed in this paper can be used to gauge the magnitude of the jump component but they are not able to identify actual jumps. So to complement the empirical results above and the discussion surrounding the burst of volatility hypothesis, we conduct some further analysis that allows us to study the properties of the jumps in more detail.

To locate jumps from five-minute data, we employ the nonparametric test of \citet*{lee-mykland:08a}, i.e.,
\begin{equation} \label{Eqn:LM}
\mathcal{L}_{i} = \frac{r_{i}}{\widehat{ \sigma}_{i}} \qquad \text{with} \qquad \widehat{ \sigma}_{i}^{2} = \frac{ \pi}{2} \frac{1}{M - 2} \sum_{j = i - M + 2}^{i - 1}|r_{j}||r_{j - 1}|.
\end{equation}
Intuitively, $\mathcal{L}_i$ compares each five-minute return to a local estimate of the instantaneous volatility and identifies a jump when the ratio is ``large''.\footnote{Throughout this section, we use a 1\% level of significance to define ``large.'' Moreover, we set the parameter $M$ in Eq. \eqref{Eqn:LM} equal to the value suggested by \citet*{lee-mykland:08a}.} We should note that $\widehat{\sigma}_i^2$ is robust to jumps. \citet*{lee-mykland:08a} show that the distribution of $\mathcal{L}_i$ is asymptotically standard normal in the absence of jumps and they specify a threshold, say $c$, for the jump detection statistic based on Gaussian extreme value theory. In addition, they assume that when a jump occurs, the size of it dominates $r_{i}$, i.e., when $|\mathcal{L}_i| > c$, a jump $\mathcal{J}_i$ is established at time $i$ with size $r_{i}$. We follow that route here.

The first column of Table \ref{Table:maxgap} reports the results of the test for every security in our universe. It shows that the number of jumps identified by this approach ($\mathcal{J}_{\#}$) translates to an average of roughly one and a half jumps per month for the equity indexes, slightly less than two jumps per week for the foreign exchange rates, and about one jump per week for the DJIA constituents. The average absolute jump size ($\mathcal{J}_{avg}$) is around $1\%$ for individual equities with the largest jumps ($\mathcal{J}_{max}$) measured at over $10\%$ (e.g., AA = Alcoa, MRK = Merck, and TRV = Travelers). The jumps in currencies and equity indexes are smaller in magnitude. With jump sizes at hand, it is straightforward to calculate the implied jump variation (JV) and interestingly, we find that the \citet*{lee-mykland:08a} procedure yields estimates that are broadly in line with those obtained from the low-frequency realised measures reported in Table \ref{Table:Descriptive}.

To apply this test to tick data, we need to modify it to correctly account for microstructure noise. To see this, note that as we move to tick frequency, the noise process starts to dominate $r_{i}$ and, asymptotically,
\begin{equation}
\displaystyle \mathcal{L}_i \simeq \frac{u_{i} - u_{i - 1}}{ \displaystyle \sqrt{ \frac{ \pi}{2} \frac{1}{M - 2} \sum_{j = i - M + 2}^{i - 1}|u_{j} - u_{j - 1}| |u_{j - 1} - u_{j - 2}|}}.
\end{equation}
The random variables $(\mathcal{L}_i)$ are $1$-dependent and identically distributed. In general, it is not clear how the rule $|\mathcal{L}_{i}| > c$ works here. If the distribution of $u_{i} - u_{i-1}$ is close to Gaussian, the test has little power and will fail to identify jumps. On the other hand, if $u_i - u_{i-1}$ has heavier tails than a normal, the statistic will be oversized and identify more jumps than expected under the null. The bottom line is that the test only depends on the distribution of the noise and not on the physical absence or presence of jumps.\footnote{We verified this empirically by computing the above test on tick data and indeed found its behavior to be very erratic.}

We can, however, apply the \citet*{lee-mykland:08a} test to tick data after pre-averaging, i.e.,
\begin{equation} \label{Eqn:LMstar}
\mathcal{L}_{i}^{ \ast} = \frac{r^{ \ast}_{i, K}}{ \widehat{ \sigma}_{i, K}^{ \ast}} \qquad \text{where} \qquad \widehat{ \sigma}_{i, K}^{ \ast2} = \frac{ \pi}{2} \frac{1}{M - 2} \sum_{j = i - M + 2}^{i - 1} |r^{ \ast}_{j, K}||r^{ \ast}_{j - K, K}|,
\end{equation}
for $i = M - 2 + K, M - 2 + 2K, \ldots$. The constant $M$ is again chosen as recommended by \citet*{lee-mykland:08a}, but here with respect to the number of non-overlapping returns following pre-averaging. The sequence $(\mathcal{L}_i^{\ast})$ is asymptotically i.i.d. standard normal if there are no jumps in the data and we can therefore apply the above critical value to identify jumps, i.e., $|\mathcal{L}_{i}^{ \ast}| > c$.

The second column of Table \ref{Table:maxgap} reports the results of the pre-averaged \citet*{lee-mykland:08a} test applied to tick data. It confirms our main finding: both the number and magnitude of identified jumps are substantially diminished. This leads to implied jump variation estimates of just below $1\%$ across the board, which agrees with the results reported in Table \ref{Table:Descriptive}. The close alignment between the (implied) JV estimates calculated using bi-power variation and the \citet*{lee-mykland:08a} jump test on five-minute as well as tick data confirms that the conclusions we draw in this paper are robust to the choice of test statistic.

To verify that our findings are also robust to the use of pre-averaging, we conduct a final piece of analysis that takes a closer look at the raw tick data. Specifically, for every jump identified by the \citet*{lee-mykland:08a} procedure using five-minute data, we search the tick data over the corresponding time interval and compare the identified jump ($\mathcal{J}_i$) to the maximum increment (in percent) of any two consecutive prices observed at this level of granularity. We refer to this maximum tick increment as the ``maxgap'' measure and denote it by $\mathcal{G}_i$.\footnote{To reduce the impact of microstructure noise, we measure price discontinuity at the tick level using a maxgap measure defined in \ref{Sec:ApndxGraphs}, where we set the buffer size $\delta$ such that it covers five observations on each side. The reported results are fairly insensitive to variations in $\delta$.} If an identified jump is genuine, we would expect to observe a discontinuity of similar magnitude in the tick data, i.e., $\mathcal{J}_i \simeq \mathcal{G}_i$. If, on the other hand, a jump is spuriously identified as a result of a burst in volatility, then we expect to observe a continuous, albeit volatile, price evolution at the tick frequency, i.e., $| \mathcal{G}_{i}| \ll | \mathcal{J}_{i}|$.

The third column of Table \ref{Table:maxgap} reports the maxgap results. The magnitude of the price discontinuity observed at the finest resolution is an order of magnitude smaller than the jumps identified from low-frequency data. The average absolute maxgap measure ($\mathcal{G}_{avg}$) implies that price discontinuities observed at tick frequency are in the region of $0.1\%$ for individual equities, whereas the five-minute jumps average about $1\%$. The maximum maxgaps at tick frequency ($\mathcal{G}_{max}$) are also substantially smaller than the maximum five-minute jumps ($\mathcal{J}_{max}$). Figure \ref{Figure:maxgap1} illustrates this further by drawing the size of every jump identified by the \citet*{lee-mykland:08a} procedure using five-minute data against the corresponding maxgap measure. For both the equity and foreign exchange rate data, the slope coefficient of a linear regression of $\mathcal{G}$ on $\mathcal{J}$ is estimated at less than 0.1. Figure \ref{Figure:spuriousjumps} provides some further illustrations of this point by highlighting a few representative examples of jumps identified using low-frequency data that vanish (to a large extent) at high frequency. All taken together, these results reaffirm our earlier finding that jumps tend to dissipate as the sampling frequency increases.

\begin{sidewaystable}
\setlength{\tabcolsep}{0.115cm}
\caption{Jump and maxgap summary statistics}
\label{Table:maxgap}
\begin{center}
\begin{small}
\begin{tabular}{lcccrccccrcclllcccrccccrccll}
\hline
                & \multicolumn{4}{c}{$\mathcal{L}$ on five-minute data} && \multicolumn{4}{c}{$\mathcal{L}^\ast$ on tick data} && \multicolumn{2}{c}{maxgap} && \multicolumn{4}{c}{$\mathcal{L}$ on five-minute data} && \multicolumn{4}{c}{$\mathcal{L}^\ast$ on tick data} && \multicolumn{2}{c}{maxgap} \\
\cline{2-5} \cline{7-10}  \cline{12-13} \cline{15-18} \cline{20-23}  \cline{25-26}
                & $\mathcal{J}_{\#}$ & $\mathcal{J}_{avg}$ & $\mathcal{J}_{max}$  & JV~~ && $\mathcal{J}_{\#}$ & $\mathcal{J}_{avg}$ & $\mathcal{J}_{max}$  & JV~~ && $\mathcal{G}_{avg}$ & $\mathcal{G}_{max}$ & & $\mathcal{J}_{\#}$ & $\mathcal{J}_{avg}$ & $\mathcal{J}_{max}$  & JV~~ && $\mathcal{J}_{\#}$ & $\mathcal{J}_{avg}$ & $\mathcal{J}_{max}$  & JV~~ && $\mathcal{G}_{avg}$ & $\mathcal{G}_{max}$ \\
\hline
\multicolumn{13}{l}{\emph{Panel A : Equity indexes}}                                                      &\multicolumn{12}{l}{\emph{Panel C (cont'd)}} \\
QQQQ           &      80 &   0.69 &   2.53 &   2.35 &&     73 &   0.30 &   1.13 &   0.43 &&   0.02 &   0.22 &HD             &     191 &   1.15 &   4.27 &   6.65 &&     63 &   0.38 &   2.66 &   1.24 &&   0.09 &   0.64 \\
SPY            &      78 &   0.71 &   3.76 &   3.80 &&     91 &   0.22 &   0.87 &   0.87 &&   0.05 &   1.20 &HPQ            &     194 &   1.08 &   5.45 &   8.66 &&     84 &   0.31 &   1.07 &   0.77 &&   0.08 &   0.89 \\
\emph{Average} &      79 &   0.70 &   3.15 &   3.07 &&     82 &   0.26 &   1.00 &   0.65 &&   0.04 &   0.71 &IBM            &     152 &   0.84 &   3.87 &   5.69 &&     73 &   0.29 &   1.08 &   2.16 &&   0.06 &   1.92 \\
\multicolumn{13}{l}{}                                                                                       &INTC           &     118 &   1.21 &   4.82 &   5.56 &&     88 &   0.33 &   1.74 &   0.79 &&   0.04 &   0.46 \\
\multicolumn{13}{l}{\emph{Panel B : Foreign exchange rates}}                                                &JNJ            &     221 &   0.67 &   4.24 &   9.39 &&     70 &   0.21 &   0.88 &   0.23 &&   0.07 &   1.24 \\
EURUSD         &     281 &   0.27 &   1.13 &   6.33 &&     55 &   0.13 &   0.66 &   0.40 &&   0.01 &   0.15 &JPM            &     142 &   1.56 &   7.57 &   5.69 &&     85 &   0.48 &   2.13 &   1.00 &&   0.05 &   0.62 \\
USDCHF         &     311 &   0.32 &   1.94 &   9.09 &&     70 &   0.15 &   0.48 &   0.93 &&   0.03 &   0.25 &KFT            &     238 &   0.94 &   6.21 &  12.72 &&     81 &   0.31 &   0.97 &   1.12 &&   0.14 &   1.66 \\
USDJPY         &     354 &   0.36 &   1.77 &  12.23 &&     65 &   0.14 &   0.39 &   0.43 &&   0.03 &   0.39 &KO             &     214 &   0.84 &   6.17 &  12.56 &&     64 &   0.29 &   1.37 &   0.93 &&   0.08 &   1.80 \\
\emph{Average} &     315 &   0.32 &   1.61 &   9.21 &&     63 &   0.14 &   0.51 &   0.59 &&   0.03 &   0.26 &MCD            &     206 &   0.88 &   5.29 &  10.64 &&     74 &   0.29 &   1.02 &   0.59 &&   0.09 &   1.96 \\
\multicolumn{13}{l}{}                                                                                       &MMM            &     174 &   0.90 &   4.52 &   7.61 &&     64 &   0.29 &   0.83 &   0.43 &&   0.09 &   0.79 \\
\multicolumn{13}{l}{\emph{Panel C : DJIA constituents}}                                                     &MRK            &     250 &   1.28 &  12.07 &  17.40 &&     75 &   0.33 &   1.09 &   0.02 &&   0.11 &   1.27 \\
AA             &     184 &   1.67 &  11.19 &   7.44 &&     70 &   0.52 &   1.74 &   2.53 &&   0.10 &   0.99 &MSFT           &     130 &   1.00 &   4.18 &   4.81 &&     71 &   0.33 &   1.41 &   1.16 &&   0.03 &   0.20 \\
AXP            &     148 &   1.59 &   7.95 &   6.82 &&     54 &   0.46 &   1.44 &   0.68 &&   0.11 &   1.04 &PFE            &     162 &   0.97 &   3.56 &   6.47 &&     73 &   0.32 &   1.01 &   1.04 &&   0.07 &   0.78 \\
BA             &     216 &   1.10 &   4.22 &   8.42 &&     71 &   0.42 &   1.68 &   1.76 &&   0.09 &   0.91 &PG             &     205 &   0.81 &   3.39 &  10.33 &&     89 &   0.26 &   1.44 &   0.64 &&   0.07 &   1.63 \\
BAC            &     173 &   1.94 &   9.89 &   8.37 &&     78 &   0.55 &   3.45 &   0.05 &&   0.08 &   1.76 &T              &     181 &   1.04 &   5.87 &   7.58 &&     80 &   0.26 &   1.16 &   0.08 &&   0.10 &   1.14 \\
CAT            &     161 &   1.28 &   4.76 &   6.25 &&     86 &   0.40 &   2.08 &   0.79 &&   0.07 &   0.56 &TRV            &     223 &   1.47 &  10.14 &  14.67 &&     88 &   0.35 &   1.88 &   1.01 &&   0.22 &   5.39 \\
CSCO           &     122 &   1.08 &   4.94 &   4.40 &&     90 &   0.31 &   1.44 &   0.20 &&   0.04 &   0.92 &UTX            &     195 &   0.85 &   3.44 &   6.34 &&     90 &   0.29 &   1.61 &   1.25 &&   0.09 &   2.15 \\
CVX            &     107 &   0.86 &   4.67 &   2.91 &&     77 &   0.28 &   1.09 &   0.12 &&   0.05 &   0.48 &VZ             &     200 &   1.06 &   3.47 &   8.58 &&     90 &   0.31 &   1.63 &   1.49 &&   0.10 &   1.37 \\
DD             &     217 &   1.20 &   4.16 &   9.15 &&     95 &   0.37 &   1.78 &   1.30 &&   0.11 &   0.95 &WMT            &     198 &   0.82 &   4.20 &   7.85 &&     84 &   0.25 &   0.96 &   0.05 &&   0.06 &   0.70 \\
DIS            &     174 &   1.04 &   4.67 &   6.63 &&     79 &   0.31 &   1.18 &   0.51 &&   0.09 &   0.68 &XOM            &      89 &   0.96 &   5.30 &   4.13 &&     89 &   0.26 &   1.03 &   0.45 &&   0.06 &   1.11 \\
GE             &     149 &   1.35 &   7.68 &   7.58 &&     69 &   0.41 &   2.29 &   0.85 &&   0.07 &   0.70 &\emph{Average} &     178 &   1.11 &   5.74 &   8.04 &&     78 &   0.34 &   1.51 &   0.84 &&   0.08 &   1.22 \\
\hline
\end{tabular}
\end{small}
\medskip
\begin{scriptsize}
\parbox{\textwidth}{\emph{Note.} We report the number of jumps ($\mathcal{J}_{\#}$), the average absolute jump size ($\mathcal{J}_{avg}$), the maximum absolute jump size ($\mathcal{J}_{max}$), and the implied jump variation (JV) for the jumps identified using the non-parametric test of \citet*{lee-mykland:08a}, $\mathcal{L}$ in Eq. \eqref{Eqn:LM}, on five-minute data and its pre-averaged counterpart, $\mathcal{L}^\ast$ in Eq. \eqref{Eqn:LMstar}, applied to tick data. The table also shows the average absolute maxgap ($\mathcal{G}_{avg}$) and the maximum absolute maxgap ($\mathcal{G}_{max}$) for the maxgap measure calculated from tick data over the five-minute window corresponding to a significant jump. The maxgap measure is defined in \ref{Sec:ApndxGraphs}.}
\end{scriptsize}
\end{center}
\end{sidewaystable}

\begin{figure}[ht!]
\begin{center}
\caption{Jumps at five-minute frequency vs. maxgap at tick frequency.}
\label{Figure:maxgap1}
\begin{tabular}{cc}
\small{Panel A: Equity data} & \small{Panel B: Foreign exchange rate data} \\
\includegraphics[height=8cm,width=0.48\textwidth]{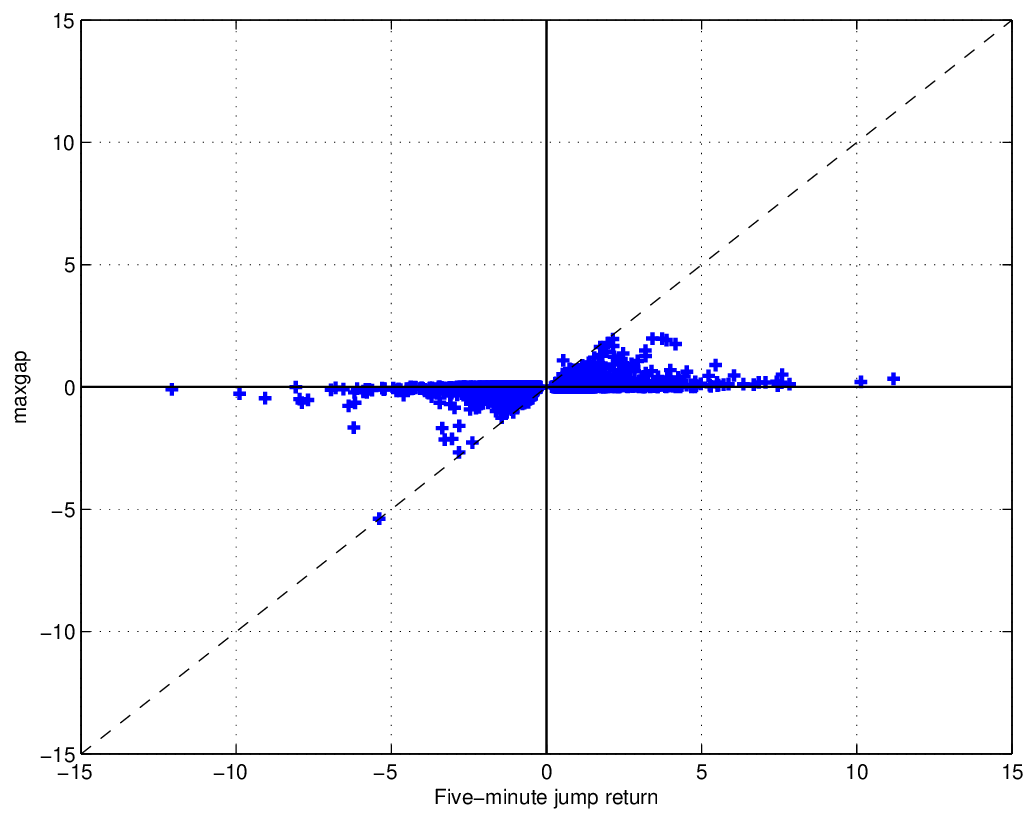} &
\includegraphics[height=8cm,width=0.48\textwidth]{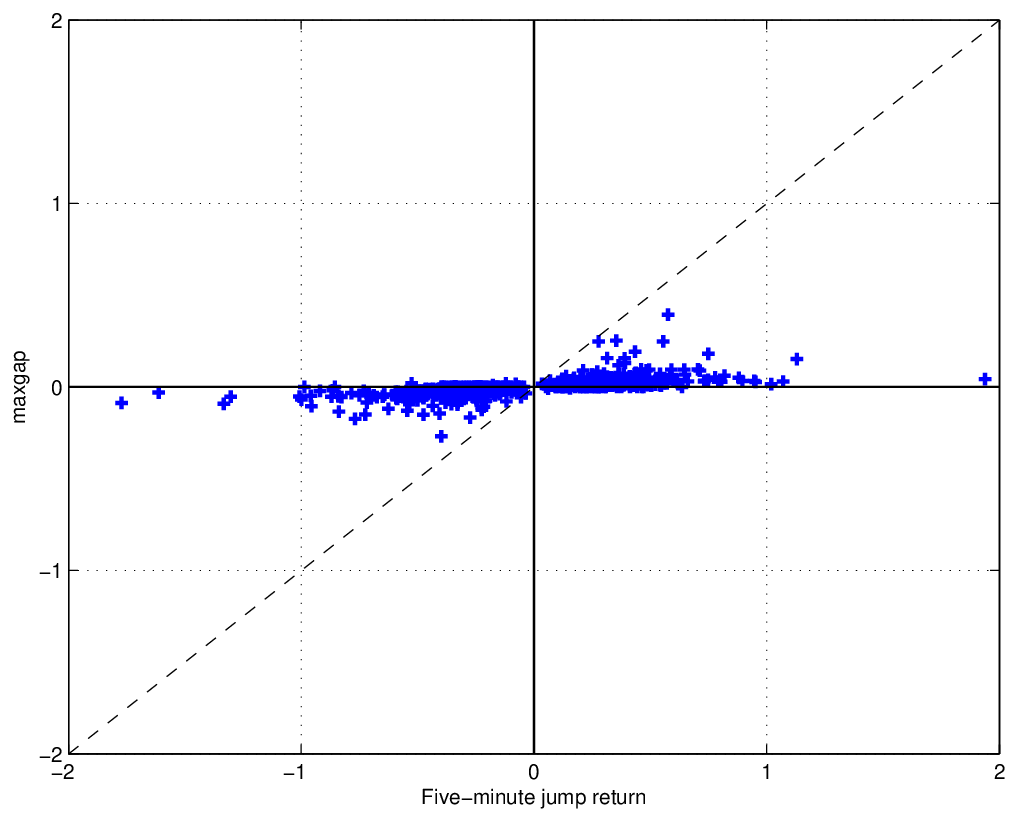} \\
\end{tabular}
\begin{scriptsize}
\parbox{\textwidth}{\emph{Note.} This figure reports a cross-plot of the jumps identified at the five-minute frequency using the nonparametric test of \citet*{lee-mykland:08a} in Eq. \eqref{Eqn:LM} (on the $x$-axis) versus the maxgap measure calculated from tick data over the corresponding time interval (on the $y$-axis). The maxgap measure is formally defined in \ref{Sec:ApndxGraphs}. All numbers are in percentages. The solid 45-degree line is for visual reference only.}
\end{scriptsize}
\end{center}
\end{figure}

\begin{sidewaysfigure}
\setlength{\tabcolsep}{0.10cm}
\begin{center}
\caption{Illustrations of jumps at five-minute frequency compared to the associated maxgap.}
\label{Figure:spuriousjumps}
\begin{tabular}{ccc}
\small{INTC 06 May 10: 5-minute jump $\mathcal{J} = 4.8\%$} & \small{MRK 25 Jan 08: 5-minute jump $\mathcal{J} = -12.1\%$} & \small{KFT 19 Sep 08: 5-minute jump $\mathcal{J} = -6.2\%$} \\
\includegraphics[height=7cm,width=0.32\textwidth]{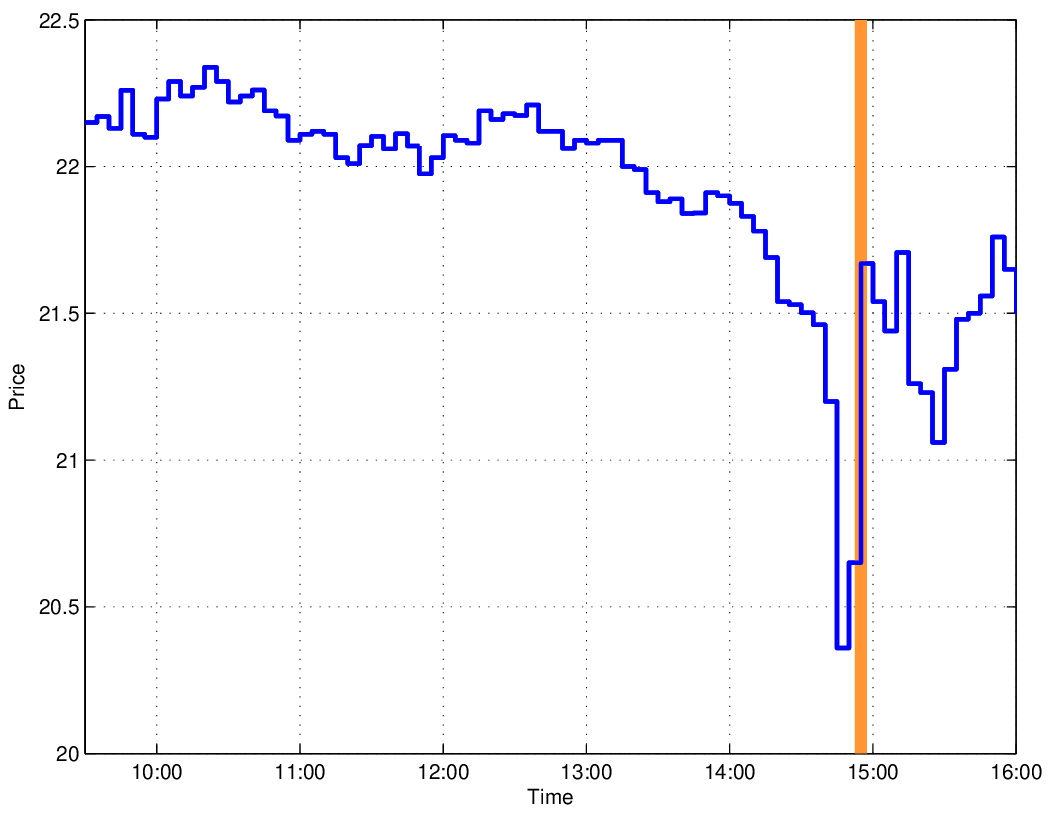} &
\includegraphics[height=7cm,width=0.32\textwidth]{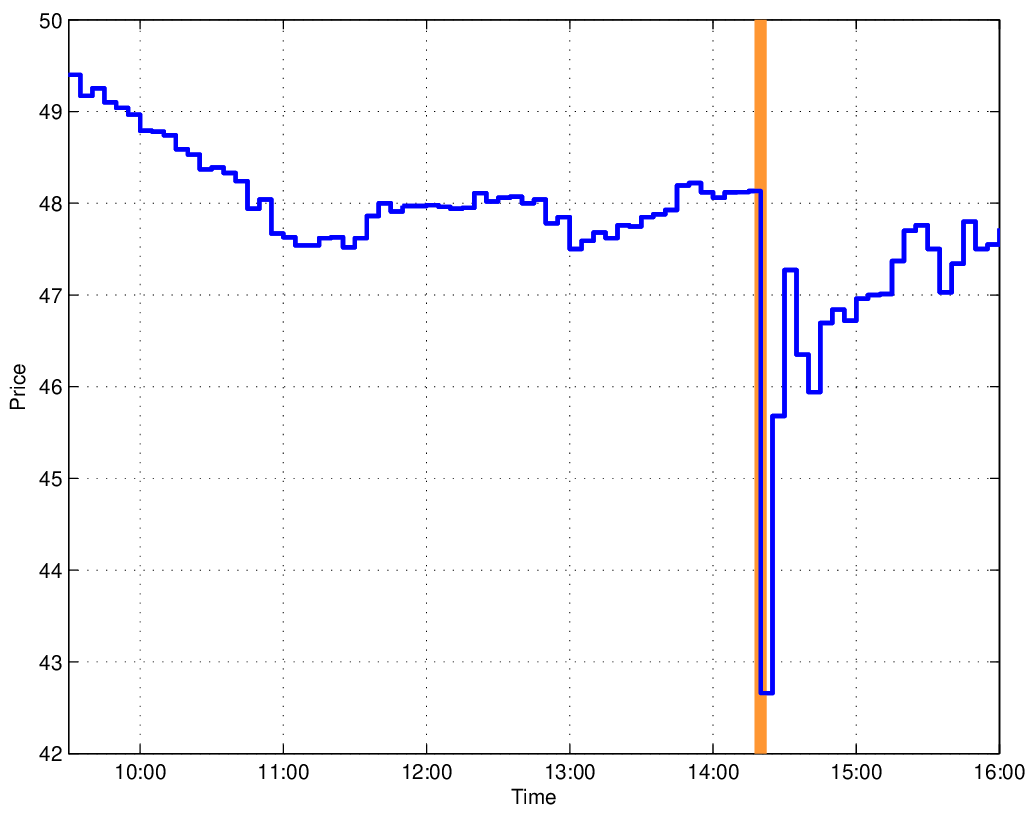} &
\includegraphics[height=7cm,width=0.32\textwidth]{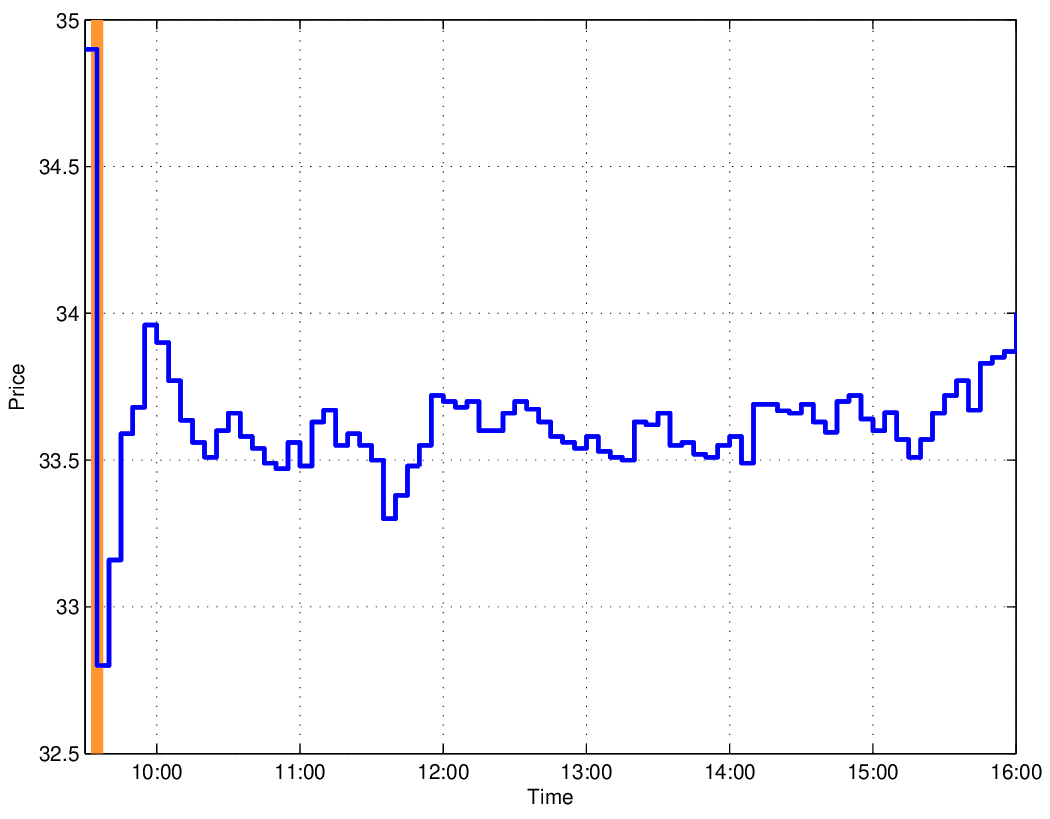} \\
\small{INTC 06 May 10: tick-frequency gap $\mathcal{G} = 0.1\%$} & \small{MRK 25 Jan 08: tick-frequency gap $\mathcal{G} = 0.1\%$} & \small{KFT 19 Sep 08: tick-frequency gap $\mathcal{G} =-1.7\%$} \\
\includegraphics[height=7cm,width=0.32\textwidth]{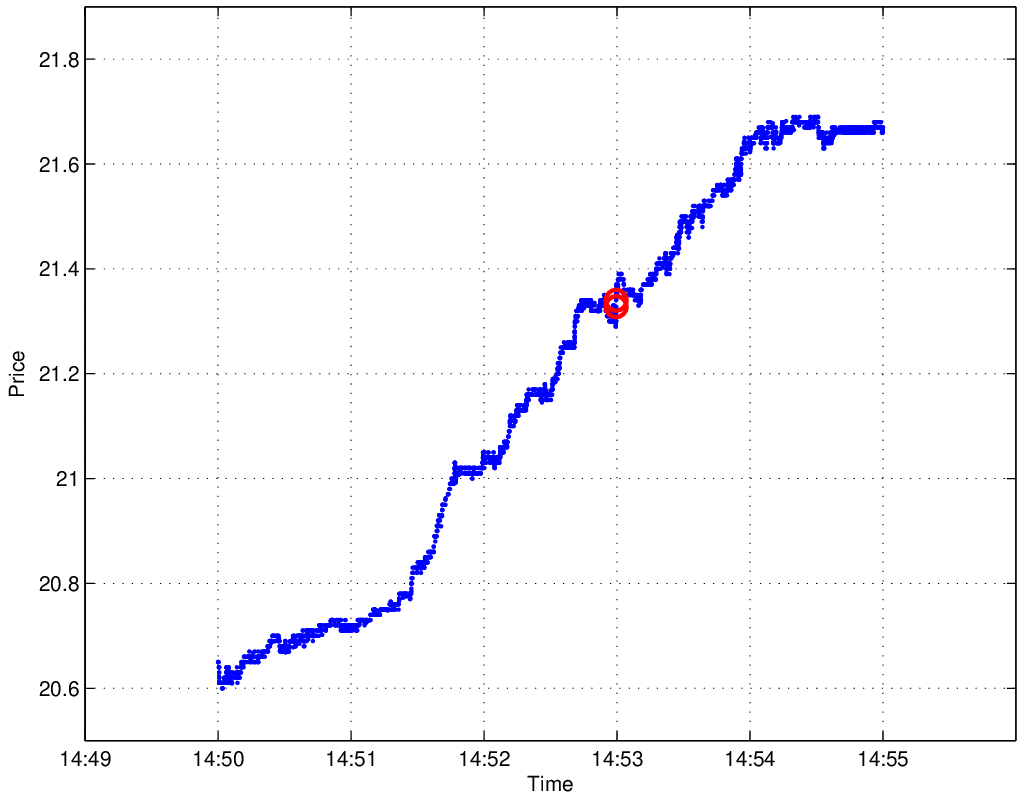} &
\includegraphics[height=7cm,width=0.32\textwidth]{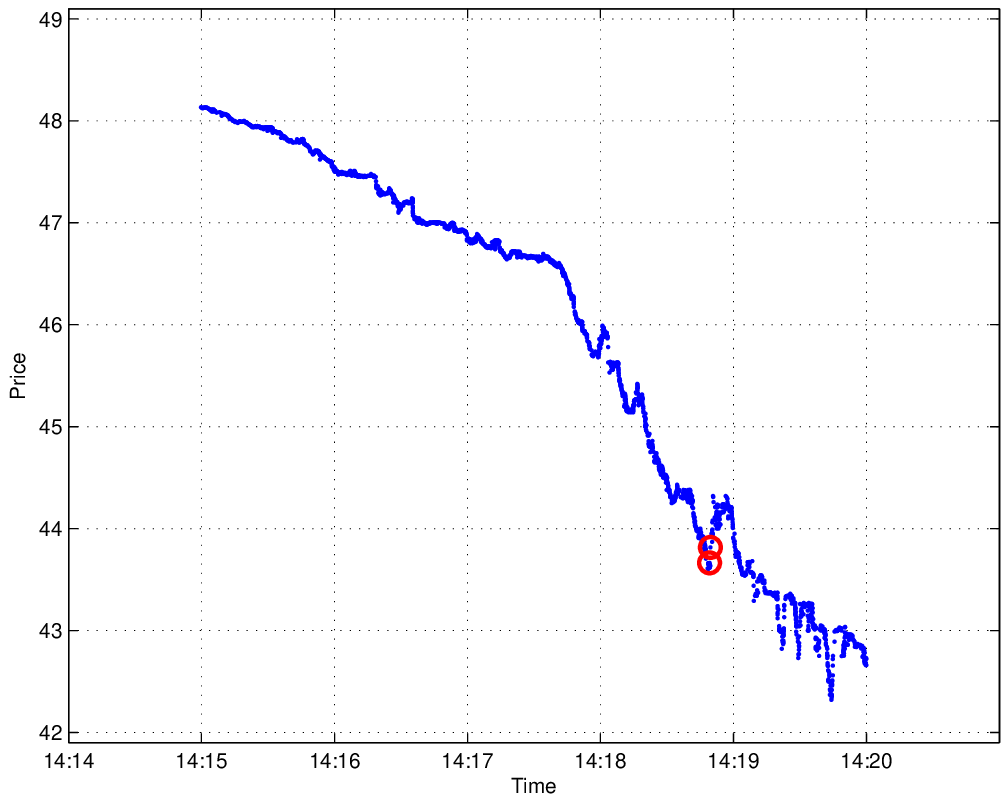} &
\includegraphics[height=7cm,width=0.32\textwidth]{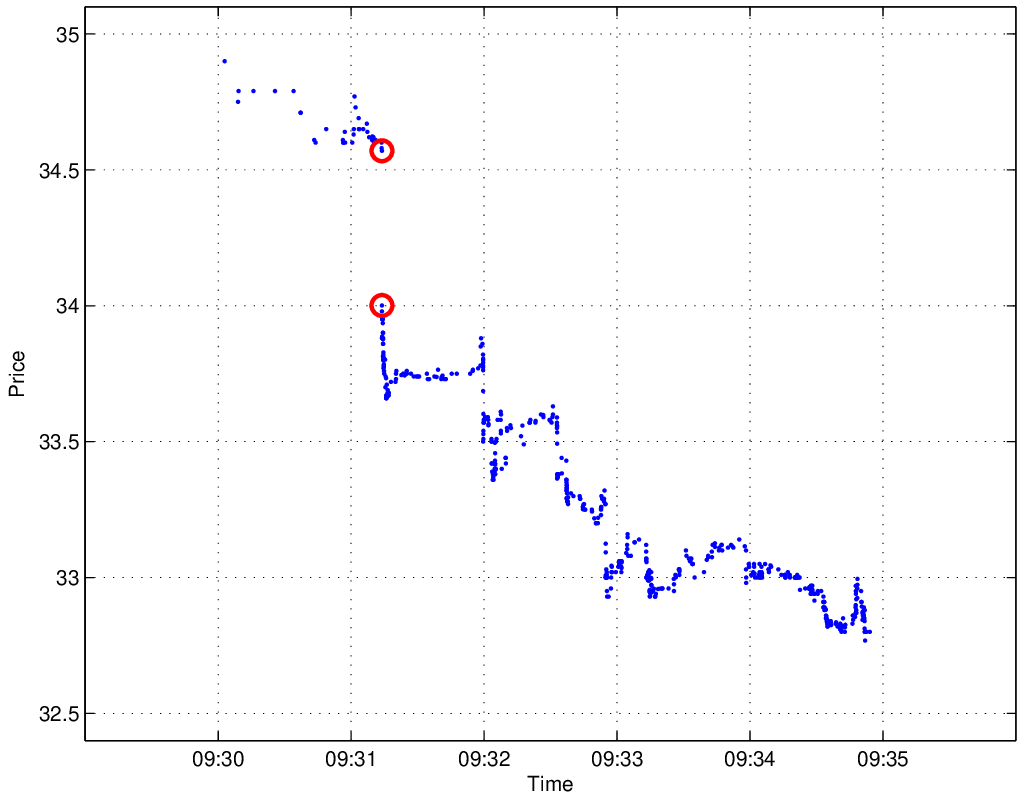} \\
\end{tabular}
\begin{scriptsize}
\parbox{\textwidth}{\emph{Note.} We give three examples of large jumps identified with the \citet*{lee-mykland:08a} test from Eq. \eqref{Eqn:LM}. The title of each subfigure shows the ticker and date of the event (INTC = Intel, MRK = Merck, KFT = Kraft Foods). The upper half shows high-frequency data for the whole day using five-minute intervals along with the jump size in percent. A shaded area is used to highlight the time of the jump. We zoom in on the jump episode in the bottom row, where the associated maxgap measure (defined in \ref{Sec:ApndxGraphs}) calculated from tick data is reported. The circles pin down the location of the maxgap. We only look at a single jump per day, although both the INTC and MRK examples possess multiple jumps according to the \citet*{lee-mykland:08a} criterion.}
\end{scriptsize}
\end{center}
\end{sidewaysfigure}

\section{Concluding remarks} \label{Sec:Conclusion}

This paper uses new econometric techniques for separating the diffusive variation component from the jump variation component, and it applies these to a comprehensive set of tick data covering both equity and foreign exchange rate data to find evidence of a much-reduced role for the jump component in explaining total return variation. Specifically, we find that the jump variation is an order of magnitude smaller than what is widely reported in the literature over the past four decades. The explanation for this can be found in the sampling frequency. The inability of the leading jump variation measures to account for market microstructure noise or market friction has prevented the literature from using tick data. Thus, in recent years a compromise five-minute sampling frequency has been used. However, at this frequency, bursts of volatility are easily mistaken for jumps, thereby obscuring the fact that jumps are not nearly as common as generally thought. A reduced role for jumps, and by implication, an increased role for the volatility process, has many important implications for empirical finance applications such as option pricing, risk management, and portfolio allocation.

\begin{figure}[ht!]
\begin{center}
\caption{USDJPY liquidity shock over the March 16, 2011 sell-off.}
\label{Figure:USDJPYliquidity}
\begin{tabular}{cc}
\footnotesize{Panel A: USDJPY spread} & \footnotesize{Panel B: Regular bid/offer quotes} \\
\includegraphics[height=8cm,width=0.48\textwidth]{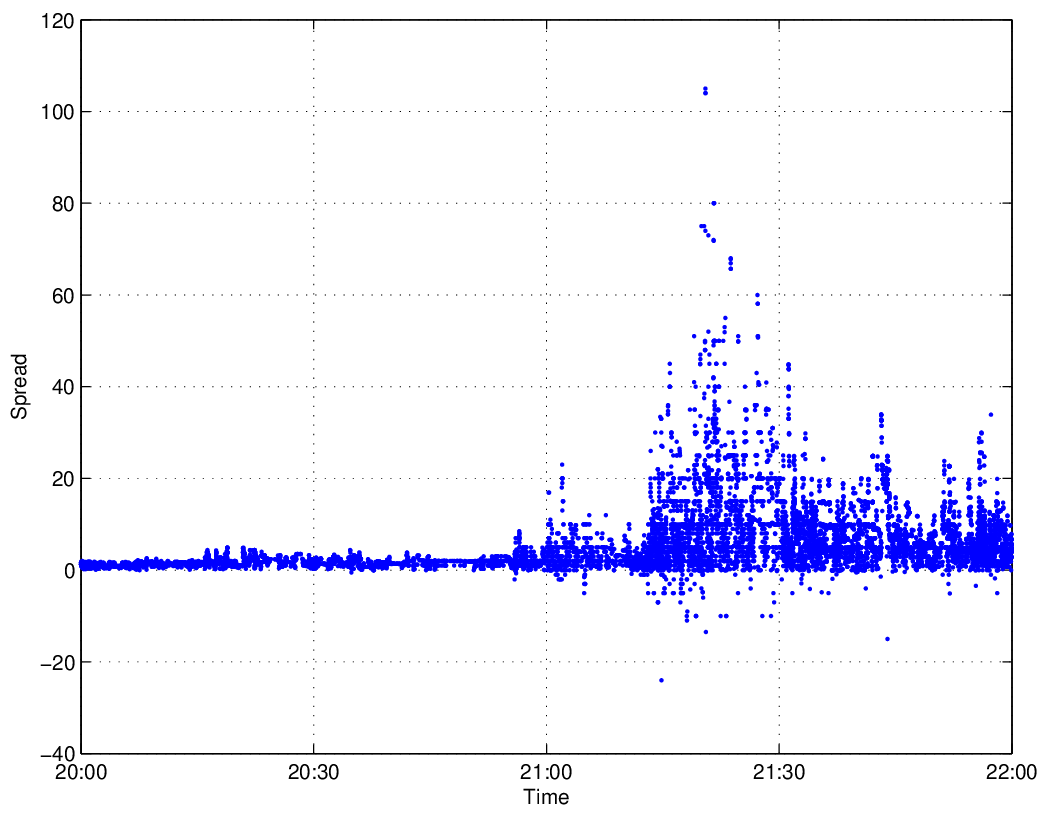} &
\includegraphics[height=8cm,width=0.48\textwidth]{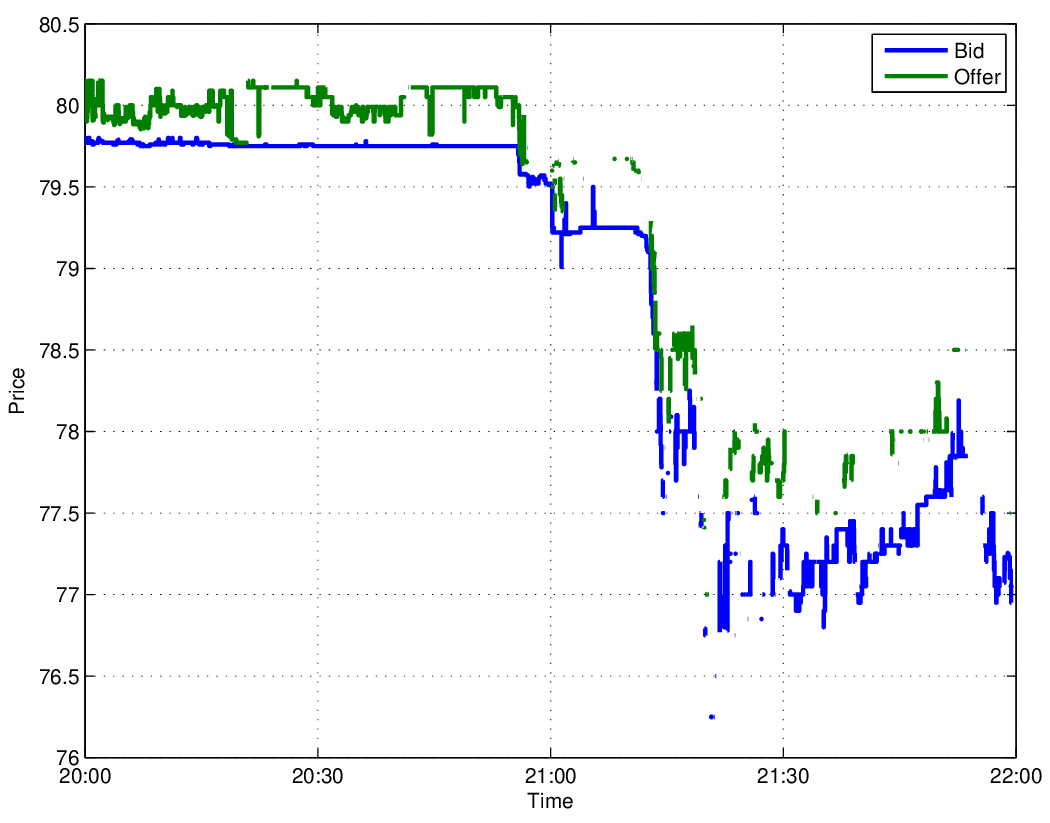} \\
\end{tabular}
\begin{scriptsize}
\parbox{\textwidth}{\emph{Note.} We plot USDJPY tick-by-tick data on March 16, 2011 with time reported in Greenwich Mean Time. The data were provided by EBSLive. In Panel A, we report the inside spread in pips over the course of the two-hour period from 20:00 -- 22:00. Panel B draws the regular quotes during this time, i.e., the best bid and ask price good for up to \$50 million.}
\end{scriptsize}
\end{center}
\end{figure}

To conclude, we emphasize that price continuity as focused on in this paper is a rather narrow concept and that one could instead focus on liquidity as a more meaningful and insightful measure of market state. It is indisputable that markets are often subject to tremendous amounts of stress. But we find this rarely leads to a substantial discontinuity in the price path at a millisecond tick resolution. Instead, what we do find are severe shocks to liquidity, a finding that is consistent with recent work by \citet{jiang-lo-verdelhan:11a}. They show that jumps identified at the five-minute frequency are often preceded by a drop in liquidity provision. Figure \ref{Figure:USDJPYliquidity} illustrates this point using USDJPY data over the March 16, 2011 sell-off. Looking at the price path in Figure \ref{Figure:Quake}, we found little evidence of a price jump. Yet, Panel A in Figure \ref{Figure:USDJPYliquidity} shows that the inside spread that normally is around one or two pips, becomes highly volatile and reaches levels in excess of a hundred pips in the midst of the sell-off.\footnote{The spread also turns substantially negative at times. In the over-the-counter foreign exchange rate market (where counterparties need to provide credit to each other), a negative spread can be observed if the participant that aggresses the market does not have credit against the counterparties providing liquidity at the top of book. Over the sell-off, several participants were forced to pay for liquidity by hitting deep into the opposite side of the order book reflecting their lack of credit. So in this instance, a large spread of either sign can indicate market turbulence. On a centrally cleared market, one would not expect to see this behavior.} Panel B plots the regular quote (i.e., the best bid and offer available for a fixed amount of \$50 million) over the same period. From here it is clear that the book is very sparse over the sell-off and that the ability to quickly trade out of large positions was severely impaired. Both these observations highlight that while price continuity was preserved over this episode of extreme volatility, liquidity was severely deteriorated. This is a point also recently emphasized in \citet*{ohara:10a} and we believe it indicates an important avenue for future research.

\newpage

\appendix

\section{Proofs} \label{Sec:ApndxProof}

\begin{proof}[Proof of Proposition \ref{Thm:CLT}]\label{Sec:ApndxCLT} The proof for the case without outliers is given in \citet*{podolskij-vetter:09a}, conditioned by Assumption (V). \\[0.10cm]

\noindent \textbf{Assumption (V).} \emph{$\sigma$ does not vanish (V$_1$) and it satisfies the equation:
\begin{equation*}
\tag{V$_2$} \sigma_{t} = \sigma_{0} + \int_{0}^{t} a_{s}^{ \prime} \text{\upshape{d}}s +
\int_{0}^{t} \sigma_{s}^{ \prime} \text{\upshape{d}}W_{s} + \int_{0}^{t} v_{s}^{ \prime}
\text{\upshape{d}}B_{s}^{ \prime}, \quad t \geq 0,
\end{equation*}
where $a^{ \prime} = \left( a_{t}^{ \prime} \right)_{t \geq 0}$, $\sigma^{ \prime} = \left( \sigma_{t}^{ \prime} \right)_{t \geq 0}$ and $v^{ \prime} = \left( v_{t}^{ \prime} \right)_{t \geq 0}$ are adapted c\`{a}dl\`{a}g, $B^{ \prime} = \left( B_{t}^{ \prime} \right)_{t \geq 0}$ is a Brownian motion, and $W \Perp B^{ \prime}$ (here, $A \Perp B$ means that A and B are stochastically independent)}.\footnote{This assumption says that $\sigma$ is of continuous semimartingale form. Note the appearance of $W$ in $\sigma$, which is the Brownian motion from the price equation in Eq. \eqref{Eqn:Xprocess}. This allows for leverage effects \citep*[e.g.,][]{christie:82a}. If $X$ itself is a continuous process, then Assumption (V) is a weak regularity condition, which is fulfilled for many financial models. For example, if $X$ is a unique, strong solution of a stochastic differential equation, then under some smoothness conditions on the volatility function $\sigma = \sigma(t, X_{t})$, Assumption (V$_{2}$) (with $v_{s}' = 0$ for all $s$) is a direct consequence of It\^{o}'s Lemma. If $X$ is not continuous, as in Eq. \eqref{Eqn:Xprocess}, then $\sigma$ is potentially discontinuous as well. In fact, there is some empirical support for allowing $\sigma$ to jump, e.g., \citet*{eraker-johannes-polson:03a}. We could include that case here at the cost of substantial extra technical rigor. Thus, even though Assumption (V) is not a necessary condition, it simplifies some of our proofs considerably. A more general treatment in the high-frequency setting, including the case where $\sigma$ jumps, is covered by \citet*{barndorff-nielsen-graversen-jacod-podolskij-shephard:06a}. We rule out these technical details here, as they are not important to our exposition.} \\[0.10cm]

Thus, what remains to be shown is that Proposition \ref{Thm:CLT} is robust to the presence of outliers. For consistency with \citet*{podolskij-vetter:09a}, we provide the proof for a more general case of pre-averaging, namely,
\begin{equation*}
r^\ast_{i,K} = \sum_{j = 1}^{K - 1} g( j/K ) r^\ast_{i+j},
\end{equation*}
where $r^\ast_i = Y_{i/N} - Y_{(i-1)/N}$, and $g: [0,1] \to \mathbb{R}$ is a weighting function assumed to be continuous, piecewise continuously differentiable such that its derivative $g'$ is piecewise Lipschitz, and $g(0) = g(1) = 0$ and $\int_0^1 g^2(u) \text{d}u > 0$. We define the following constants that are associated with $g$:
\begin{equation*}
\psi _1 = \int_0^1 \left(g'\left( s \right) \right) ^2\text{d}s,
\qquad
\psi_2 = \int_0^1 g^2\left( s \right)\text{d}s.
\end{equation*}
The noise-robust variation measures are now defined as follows:
\begin{align*}
RV^\ast &= \frac{N}{N - K + 2} \frac{1}{ K \psi_{2}^{K}} \sum_{i = 0}^{N - K + 1} |r^\ast_{i,K}|^2 - \frac{\widehat{\omega}^{2}}{\theta^{2}} \frac{\psi_{1}^{K}}{\psi_{2}^{K}}, \\
BV^\ast &= \frac{N}{N - 2K + 2} \frac{1}{ K \psi_{2}^{K}}\frac{\pi}{2} \sum_{i = 0}^{N - 2K + 1} |r^\ast_{i,K}| |r^\ast_{i+K,K}| - \frac{\widehat{\omega}^{2}}{\theta^{2}} \frac{\psi_{1}^{K}}{\psi_{2}^{K}},
\end{align*}
where $\psi_1^K$ and $\psi_2^K$ are Riemann approximations to $\psi_1$ and $\psi_2$ that improve finite sample accuracy:
\begin{equation*}
\psi_1^K = K \sum_{j = 1}^K \left[ g\left( \frac{j}{K} \right)- g \left( \frac{j-1}{K} \right)\right]^2, \qquad \psi_2^K = \frac{1}{K} \sum_{j = 1}^{K - 1} g^2 \left( \frac{j}{K} \right).
\end{equation*}
In the main text we have simplified the above setup by selecting a specific choice of weight function $g(x) = \min (x, 1-x)$ and assumed that $K$ is even. In that case, $\psi_1^K = 1$ and $\psi_2^K = \psi_K$.

To show robustness of the result to outliers, we first recall that the definition in Eq. \eqref{An} implies that there are only finitely many $i$'s such that $O_{i/N} \not = 0$. Here, we confine attention to proving the outlier robustness of the $RV^\ast$ estimator. The robustness property of $BV^\ast$ can be proved using almost identical tools.

First, we note that the suggested estimator for $\omega^2$ is robust to (finite activity) outliers, i.e.,
\begin{equation*}
\hat{\omega}_{\text{AC}}^{2} = - \frac{1}{N - 1} \sum_{i=2}^N (u_{i}-u_{i-1})(u_{i-1}-u_{i-2}) + O_p(N^{-1}).
\end{equation*}

Next, we know that with probability converging to one, any interval of the form $[i/N, (i+K)/N]$ contains at most a single outlier. Using that $g(0) = g(1) = 0$, we can re-express the return on pre-averaged prices as follows:
\begin{equation*}
r^\ast_{i,K} = \sum_{j = 1}^K \Big[ g\Big( \frac{j-1}{K} \Big) - g\Big( \frac{j}{K} \Big) \Big] Y_{(i+j-1)/N}.
\end{equation*}
This shows that an outlier $O_{l/N}$ in $[i/N, (i+K)/N]$ induces a bias of the form $\Big[g\Big(\frac{l-i}{K} \Big) - g\Big(\frac{l-i+1}{K} \Big)\Big] O_{l/N}$ and, as there are finitely many outliers, only $O(K)$ of the returns $r^\ast_{i,K}$ are affected by these. Putting the parts together, we conclude that
\begin{equation*}
RV^\ast[Y] = RV^\ast[Z] + O_p(K^{-1})
\end{equation*}
where $Z = X + u$. Thus, the law of large numbers and central limit theorem for $RV^\ast$ are robust to the presence of outliers. \qed
\end{proof}

\begin{proof}[Derivation of Eq. \eqref{Eqn:loglog}] \label{Sec:BVburst}
In a simple model, it is possible to work out an approximate expression for the finite sample bias of the BV defined by Eq. \eqref{Eqn:BVn} in the main text in the paper. To this end, consider the model
\begin{equation*}
X_t = \int_0^t \sigma_s \text{d}W_s,
\end{equation*}
where $\sigma \Perp W$. Moreover, we assume that the variance process $\sigma^2$ is bounded away from zero and has the form
\begin{equation*}
\sigma^2_t = \sigma_0^2+ \int_0^t \eta_s \text{d}s + \int_0^t \upsilon_s \text{d}B_s,
\end{equation*}
where $B$ is another Brownian motion with $B \Perp W$. Let
\begin{equation*}
BV = \frac{\pi}{2} \sum_{i = 2}^{N} |r_{i - 1}| |r_{i}|.
\end{equation*}
Here, we initially leave out the correction $N/(N-1)$ used to define BV in the main text. We add it back later.

Then, the conditional bias of $BV$ is given as
\begin{equation*}
E\left(BV - \int_{0}^{1} \sigma_s^2 \text{d}s \mid \sigma\right) = \sum_{i=1}^{N-1} \left \{ \Bigg( \int_{\frac{i-1}{N}}^{\frac{i}{N}} \sigma_s^2 \text{d}s \Bigg)^{1/2} \Bigg( \int_{\frac{i}{N}}^{\frac{i+1}{N}} \sigma_s^2 \text{d}s \Bigg)^{1/2} - \int_{\frac{i-1}{N}}^{\frac{i}{N}} \sigma_s^2 \text{d}s \right\} - \int_{\frac{N-1}{N}}^{1} \sigma_s^2 \text{d}s.
\end{equation*}
To simplify notation, we define $f(x)=\sqrt{x}$ and set
\begin{equation*}
a_i = N \int_{\frac{i-1}{N}}^{\frac{i}{N}} \sigma_s^2 \text{d}s.
\end{equation*}
Hence,
\begin{equation*}
E\left(BV - \int_0^1 \sigma_s^2 \text{d}s \mid \sigma \right) = \frac{1}{N} \Bigg( a_N + \sum_{i = 1}^{N - 1} f(a_i) (f(a_{i + 1}) - f(a_i)) \Bigg).
\end{equation*}
By Burkholder's inequality, we get
\begin{equation*}
E(|a_{i+1} - a_{i}|^p)\leq C N^{-p/2}
\end{equation*}
for any $p>0$. Thus,
\begin{equation*}
E\left(BV - \int_{0}^{1} \sigma_s^2 \text{d}s \mid \sigma\right) = \frac{1}{N} \Bigg( a_N + \sum_{i=1}^{N-1} f(a_i) \{f'(a_i)(a_{i+1}-a_i) + \frac{f''(a_i)}{2} (a_{i+1}-a_i)^2 \} \Bigg) + o_p(N^{-1}).
\end{equation*}
Note that $f'(x)= \frac{1}{2} x^{-1/2}$ and $f''(x)=-\frac{1}{4} x^{-3/2}$. Using this, we deduce that
\begin{equation*}
\frac{1}{N} \sum_{i = 1}^{N - 1} f(a_i) f'(a_i)(a_{i+1}-a_i)= \frac{1}{2N} (a_N - a_1) = \frac{1}{2N} (\sigma_1^2 - \sigma_0^2) + o_p(N^{-1}).
\end{equation*}
On the other hand,
\begin{equation*}
\frac{1}{N} \sum_{i = 1}^{N - 1} f(a_i) \frac{f''(a_i)}{2} (a_{i+1} - a_i)^2 = - \frac{1}{N} \sum_{i = 1}^{N - 1} \frac{1}{8} \frac{(a_{i+1} - a_i)^2}{a_i}.
\end{equation*}
Because $\sigma^2$ is a continuous semimartingale, a standard approximation scheme yields
\begin{equation*}
-\frac{1}{N} \sum_{i = 1}^{N - 1} \frac{1}{8} \frac{(a_{i+1} - a_i)^2}{a_i} = -N \sum_{i = 1}^{N - 1} \frac{1}{8} \frac{\upsilon_{\frac{i - 1}{N}}^{2}}{\sigma_{{\frac{i - 1}{N}}}^2} \Bigg(\int_{\frac{i}{N}}^{\frac{i + 1}{N}} B_s \text{d}s - \int_{\frac{i - 1}{N}}^{\frac{i}{N}} B_s \text{d}s\Bigg)^2 + o_p(N^{-1}).
\end{equation*}
Moreover, as
\begin{equation*}
E\Bigg( \Bigg[ \int_{\frac{i}{N}}^{\frac{i + 1}{N}} B_s \text{d}s - \int_{\frac{i - 1}{N}}^{\frac{i}{N}} B_s \text{d}s \Bigg]^2 \Bigg) = \frac{2}{3N^3},
\end{equation*}
we conclude that
\begin{equation*}
-\frac{1}{N} \sum_{i = 1}^{N - 1} \frac{1}{8} \frac{(a_{i + 1} - a_i)^2}{a_i} = -\frac{1}{N} \frac{1}{12} \int_0^1 \frac{\upsilon_{s}^{2}}{\sigma_{s}^2} \text{d}s + o_p(N^{-1}).
\end{equation*}
Putting everything together, we find that
\begin{equation*}
E\left(BV - \int_0^1 \sigma_s^2 \text{d}s \mid \sigma \right) = -\frac{1}{N} \left(\frac{1}{12} \int_0^1 \frac{\upsilon_{s}^{2}}{\sigma_{s}^2} \text{d}s + \frac{1}{2} (\sigma_1^2 + \sigma_0^2) \right) + o_p(N^{-1}).
\end{equation*}
Now, applying the finite sample correction $N / (N-1)$ cancels out the effect of the missing summand and adds an additional $\int_0^1 \sigma_s^2 \text{d}s / N$  term to the conditional bias:
\begin{equation*}
E\left(\frac{N}{N-1} BV - \int_0^1 \sigma_s^2 \text{d}s \mid \sigma \right) = -\frac{1}{N} \left(\frac{1}{12} \int_0^1 \frac{\upsilon_{s}^{2}}{\sigma_{s}^2} \text{d}s + \frac{1}{2} (\sigma_1^2 + \sigma_0^2) - \int_0^1 \sigma_s^2 \text{d}s \right) + o_p(N^{-1}).
\end{equation*}
Note that this expression is not negative in general, but that the sign and magnitude of the dependence on $\upsilon$ are unaffected by the adjustment. Finally, taking unconditional expectations and assuming that $\sigma^2$ is a stationary process, we find that
\begin{equation*}
E\left(\frac{N}{N-1}BV - \int_0^1 \sigma_s^2 \text{d}s \right) = -\frac{1}{N} E \left( \frac{1}{12} \int_0^1 \frac{\upsilon_{s}^{2}}{\sigma_{s}^2} \text{d}s \right) + o(N^{-1}).
\end{equation*}
This is the expression given in the main text, which is negative up to terms of order $o(N^{-1})$. \qed
\end{proof}

\section{The explicit form of $\Sigma^\ast$} \label{Sec:ApndxSigma}

Proposition \ref{Thm:CLT} does not specify the exact form of the asymptotic covariance matrix $\Sigma^\ast$. Here, a formula for it is provided. We set $f_i: \mathbb R^2 \rightarrow \mathbb R$, $i=1,2$, equal to
\begin{equation*}
f_1(x)=x_1^2, \qquad f_2 (x)= \frac{|x_1||x_2|}{\mu^2}.
\end{equation*}
Then, for $x\in \mathbb{R}$, $u\in [0,1]$ and $l = -1, \ldots, 2$, we define
\begin{equation*}
F_{l,x,u}^{(ij)} = \text{\upshape{cov}} (f_i(S), f_j(T) ) \qquad 1\leq i,j \leq 2,
\end{equation*}
where $S=(S_1,S_2)'$, $T=(T_1,T_2)'$ are centered and jointly normal with
\begin{itemize}
\item[(i)] $S_i\bot S_j$, $T_i\bot T_j$ for all $i \not= j$.
\item[(ii)] $\text{\upshape{var}}(S_i)= \text{\upshape{var}}(T_i)=\theta  \psi_2  x^2 +
    \frac{\psi_1}{\theta}  \omega^2$  for all $i$.
\item[(iii)] $\text{\upshape{cov}}(S_{i+l-1},T_i)=\theta  w_{g}(u) x^2 + \frac{1}{\theta}
    w_{g'} (u) \omega^2$ for all $i$.
\item[(iv)] $\text{\upshape{cov}}(S_{i+l},T_i)=\theta w_{g}(1-u) x^2 + \frac{1}{\theta} w_{g'}
    (1-u) \omega^2$ for all $i$.
\item[(v)] $\text{\upshape{cov}}(S_i,T_j)=0$ for all $|i+l-j-1|>1$.
\end{itemize}
Here, the function $w_g(u)$ is given by
\begin{equation*}
w_{g} \left( u \right) = \int_{0}^{1 - u} g \left( y \right) g \left( y + u \right)
\text{\upshape{d}}y.
\end{equation*}
Finally, we get that
\begin{equation*}
\Sigma^\ast = \frac{1}{\theta \psi_2^2}\sum_{l = -1}^2 \int_0^1 \int_0^1 F_{l,\sigma_s,u} \text{\upshape{d}}s\text{\upshape{d}}u.
\end{equation*}

\section{Consistency with irregular sampling} \label{Sec:ApndxSampling}

We show robustness of the $BV^{*}$ estimator to irregular sampling (the corresponding result for $RV^{*}$ is almost identical) of the following type: The sampling times are given by $0 = t_0 < \ldots < t_N = 1$ with $t_{i} = f(i/N)$ for a strictly increasing continuously differentiable function $f:[0,1]\rightarrow [0,1]$ such that $f(0)=0$ and $f(1)=1$.

The pre-averaged returns $r^{*}_{i,K}$ are defined as in the equidistant setting. Following the proofs in \citet*{podolskij-vetter:09a}, we know that (a) $r^{*}_{i,K}$ and $r^{*}_{i+K,K}$ are asymptotically uncorrelated and (b) conditionally on $\mathcal F_{t_i}$, $r^{*}_{i,K}$ is asymptotically distributed according to
\begin{equation*}
\mathcal{N} \left(0, \sigma_{t_i}^2 \sum_{j=1}^{K} g^2 \Big( \frac{j}{K}\Big) (t_{i+j}- t_{i+j-1}) +\omega^2 K\sum_{j=1}^{K} \Big[g \Big( \frac{j}{K}\Big) - g \Big( \frac{j-1}{K}\Big)\Big]^2 \right).
\end{equation*}
As in \citet*{podolskij-vetter:09a}, we conclude that
\begin{equation*}
BV^{*} - \left( \frac{1}{K \psi_K} \frac{\pi}{2} \sum_{i = 0}^{N -2K + 1} E \Big(|r^{*}_{i,K}||r^{*}_{i+K,K}| \mid \mathcal F_{t_i}\Big) - \frac{\hat \omega^2}{\theta^2 \psi_K}\right) \overset{p}{\to} 0.
\end{equation*}
Applying the mean value theorem  $t_j -t_{j-1}\sim \frac 1N f'(j/N)$ and using the continuity of $f'$ and $\sigma$, we find that
\begin{equation*}
\frac{1}{K\psi_K} \frac{\pi }{2} \sum_{i=0}^{N-2K+1} E\Big(|r^{*}_{i,K}||r^{*}_{i+K,K}| \mid \mathcal F_{t_i}\Big) - \frac{\hat \omega^2}{\theta^2 \psi_K} = \frac 1N \sum_{i=0}^{N-2K+1} f'(i/N) \sigma_{t_i}^2 + o_{p}(1).
\end{equation*}
Thus,
\begin{equation*}
BV^{*} \overset{p}{\to} \int_0^1 f_{s}' \sigma_{f_{s}}^{2} \text{d}s = \int_{0}^{1} \sigma_{s}^{2} \text{d}s,
\end{equation*}
where we have used integration by parts and the fact that $f(0)=0$, $f(1)=1$.

\section{BFM filter and the maxgap measure} \label{Sec:ApndxGraphs}

\setcounter{figure}{0} \renewcommand{\thefigure}{D\arabic{figure}}
\setcounter{table}{0} \renewcommand{\thetable}{D\arabic{table}}

Figure \ref{Figure:tradefilter} provides an illustration in which the quote and trade feed go temporarily out of sync and shows the distinct ways in which the data cleaning filters operate.\footnote{While in recent data sets this phenomenon in which trades lag quotes is rare and restricted to short periods of heightened market activity when the exchange networks are at maximum capacity, it is routinely observed in the earlier high-frequency data sets \citep*[see, e.g.,][]{lee-ready:91a}.} The BNHLS procedure ends up removing large chunks of data, leaving a large but spurious filtering-induced jump in the cleaned price process. In contrast, the BFM filter matches most of the trades over the backward-looking window and the resulting filtered trade price series is relatively smooth. Table \ref{Table:filter} reports summary statistics by instrument on the number of outliers identified, matched, and removed by the BFM filter.

In Section \ref{Sec:Femp} we use the ``maxgap'' $\mathcal{G}_i$ to measure the price discontinuity observed in tick data over every five-minute episode where the \citet*{lee-mykland:08a} test identifies a jump. The construction is as follows. For a jump $\mathcal{J}_i$ identified at the $i^{th}$ five-minute interval, we take the tick data series $\{Y_{t_j}\}$ for all $j$ where $t_j\in [(i-1)/78,i/78]$ and use this to compute:
\begin{equation*}
\mathcal{G}_i = g_{j^\ast} \quad \text{where}\quad j^\ast = \argmax_j |g_j|
\end{equation*}
where
\begin{equation*}
g_j = \begin{cases}
\max(0, \min_{k \in F_{j}} Y_{t_{k}} - \max_{k \in B_{j}} Y_{t_{k}}) &
\text{for } Y_{t_{j + 1}} > Y_{t_{j}} \\
\min(0, \max_{k \in F_{j}} Y_{t_{k}} - \min_{k \in B_{j}} Y_{t_{k}}) &
\text{for } Y_{t_{j + 1}} < Y_{t_{j}}
\end{cases}
\end{equation*}
and $B_j = \{k:t_k\in [t_{j}-\delta,t_{j}]\}$, $F_j = \{k:t_k\in [t_{j+1},t_{j+1}+\delta]\}$ for $\delta\geq 0$. In words, for a candidate jump from $Y_{t_j}$ to $Y_{t_{j+1}}\gg Y_{t_j}$, $g_{j}$ is defined as the difference between the lowest price observed in a local window $F_j$ of length $\delta$ just after the candidate jump and the highest price observed in a local window $B_j$ just before the candidate jump. As such, the gap measure $g_j$ identifies the largest range of price points that are jumped over in a short timespan around the $j^{th}$ observation. Based on this, the maxgap measure $\mathcal{G}_i$ then simply identifies the largest tick price gap over the interval. Compared to simply calculating tick returns (a special case with $\delta = 0$), the maxgap measure reduces the impact of noise (see below for an illustration) and while $|g_j| \leq |Y_{t_{j+1}}-Y_{t_j}|$, in a genuine jump scenario the difference is negligible and they can be interpreted alike.

Figure \ref{Figure:maxgap} provides an illustration of the maxgap calculation for four different scenarios. In Scenario I, the maxgap and max return identify the same jump location. The difference in their value is attributable to microstructure noise (maxgap is largely robust to this) but it should be negligible for genuine jumps. In Scenario II, the max return selects a different location and takes the value -4 while maxgap is +3. In Scenario III, maxgap is zero whereas the max return is +3. In Scenario IV, maxgap and max return both identify the same two jumps. The genuine jumps in Scenarios I, II, and IV are identified correctly by maxgap, as is the no-jump in Scenario III. On the other hand, the max return only works well in Scenarios I and IV but picks up noise in the other two scenarios resulting in the incorrect identification of the jump in Scenario II and a false positive in Scenario III.

\begin{figure}[ht!]
\begin{center}
\caption{Trade filter illustration (SPY, 18 Sep 2007).}
\label{Figure:tradefilter}
\begin{tabular}{cc}
\small{Panel A: Full trading day} & \small{Panel B: Episode where trades lag quotes} \\
\includegraphics[height=8cm,width=0.48\textwidth]{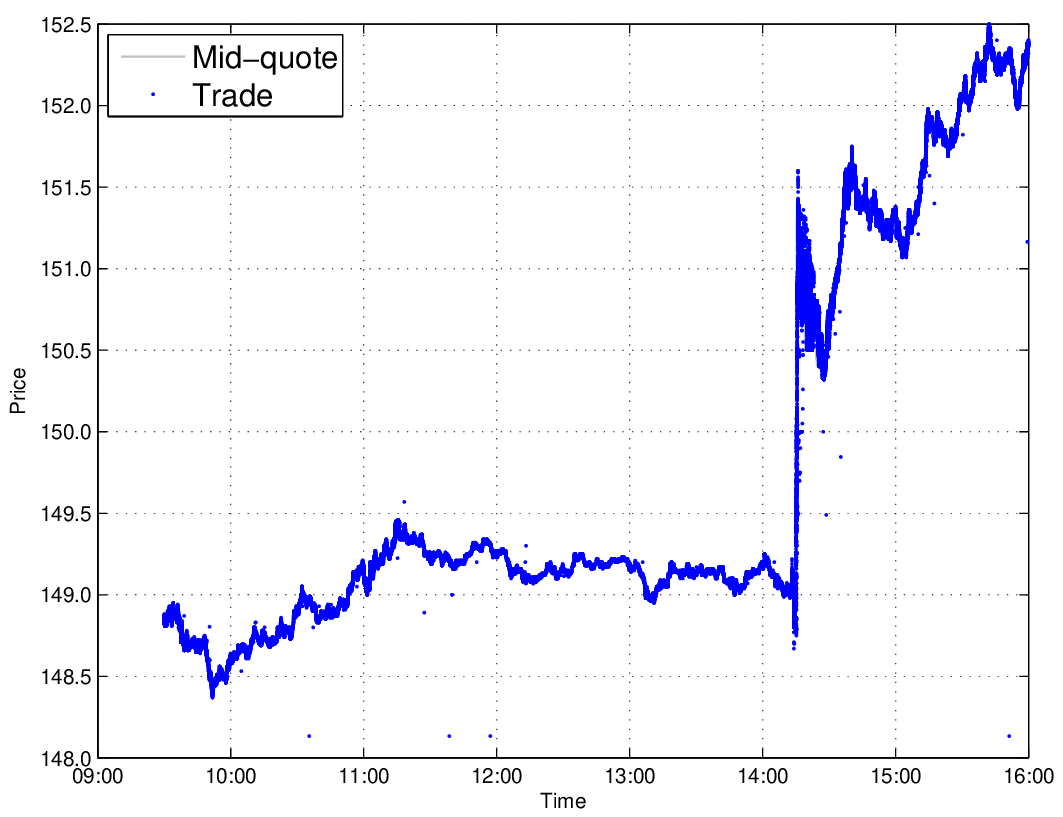} &
\includegraphics[height=8cm,width=0.48\textwidth]{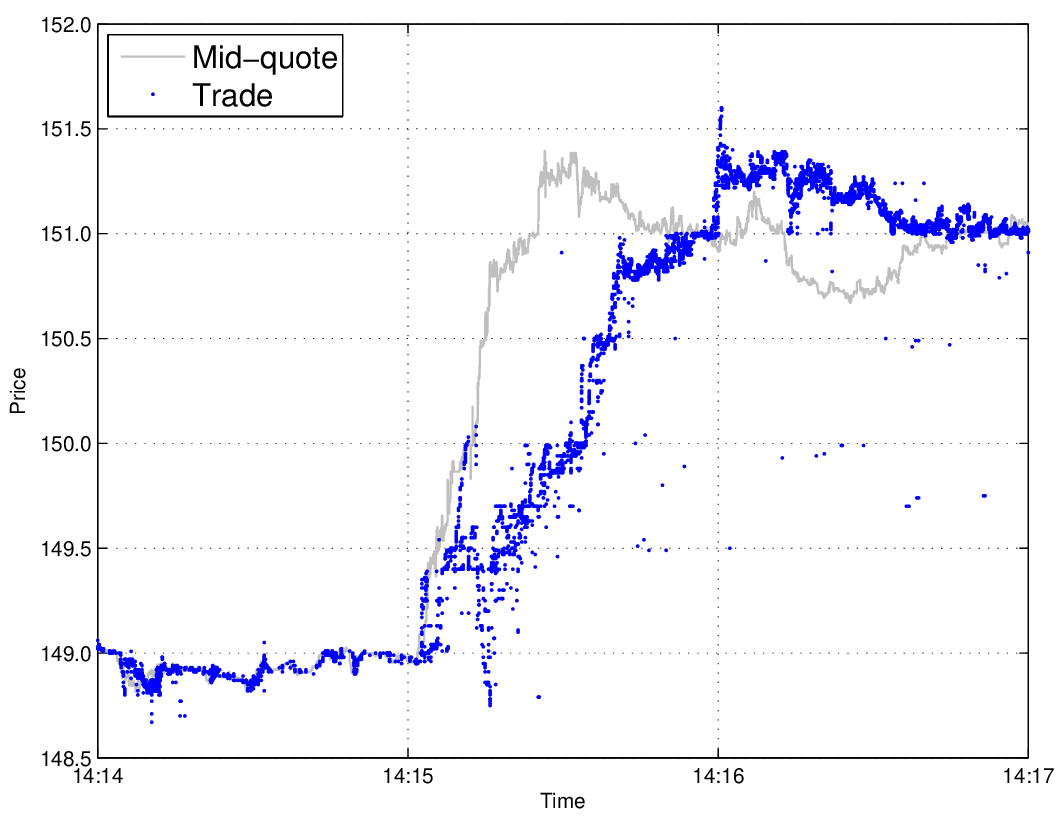} \\
\small{Panel C: Trades post BNHLS filter} & \small{Panel D: Trades post BFM filter} \\
\includegraphics[height=8cm,width=0.48\textwidth]{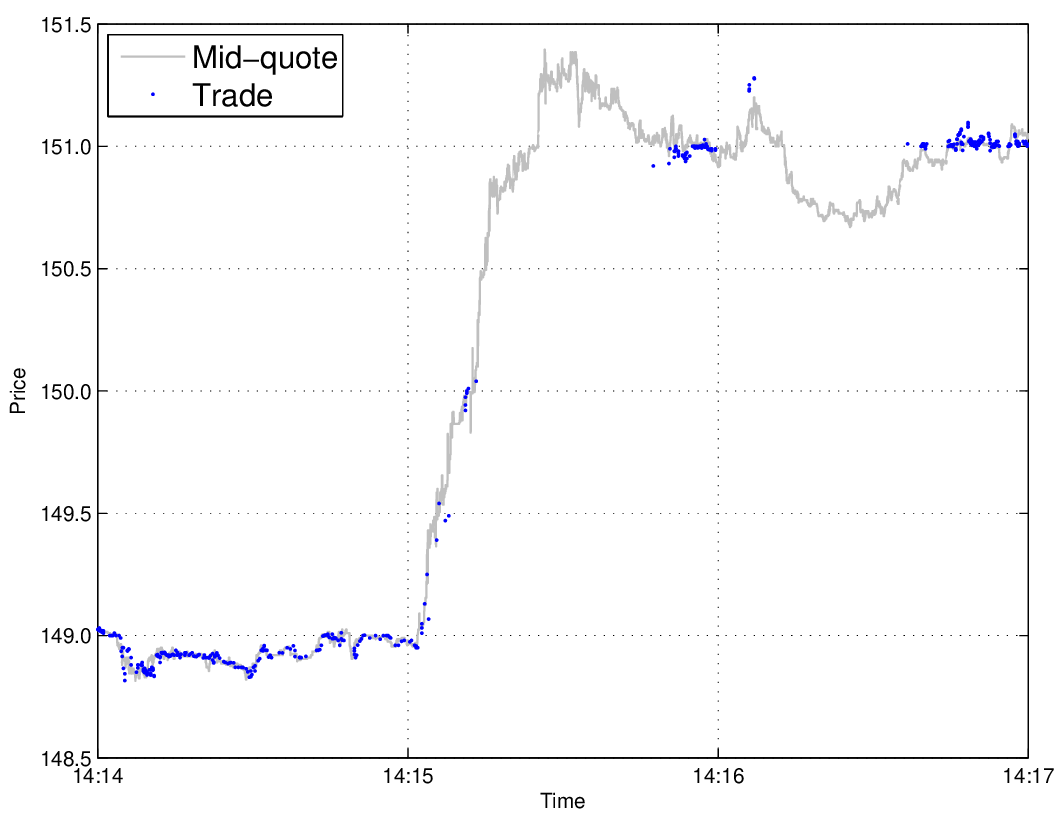} &
\includegraphics[height=8cm,width=0.48\textwidth]{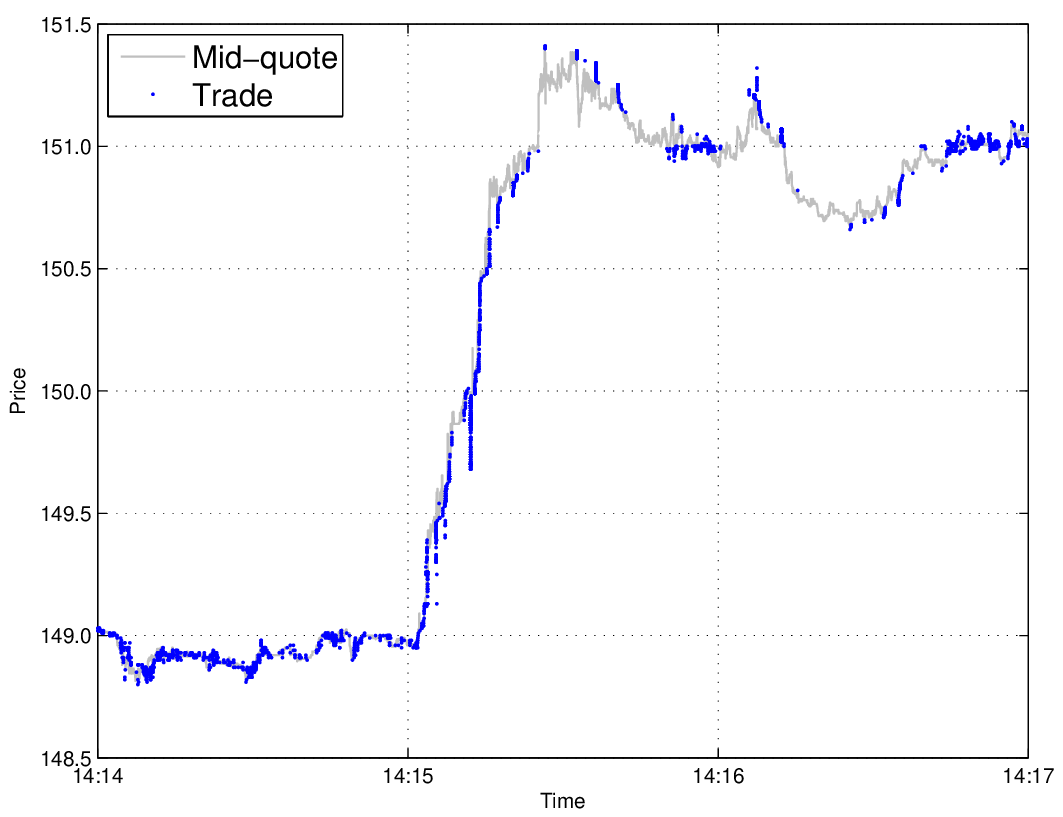} \\
\end{tabular}
\begin{scriptsize}
\parbox{\textwidth}{\emph{Note.} Panel A draws the mid-quote and unfiltered trades for SPY on September 18, 2007. Panel B zooms in on the volatile episode beginning at 14:15 following an unexpected interest rate cut by the FED. Panel C (D) plots the mid-quote and remaining trades over this period after filtering the data with the BNHLS (BFM) procedure.}
\end{scriptsize}
\end{center}
\end{figure}

\begin{sidewaystable}
\setlength{\tabcolsep}{0.195cm}
\caption{BFM trade filter statistics}
\label{Table:filter}
\begin{center}
\begin{small}
\begin{tabular}{lrrrrrrrrclrrrrrrrr}
\hline
    & raw~      & cand.~~ & \multicolumn{2}{c}{frwd match} && \multicolumn{2}{c}{bwrd match}&    &&& raw~ & cand.~~ & \multicolumn{2}{c}{frwd match} && \multicolumn{2}{c}{bwrd match}&\\
\cline{4-5}\cline{7-8} \cline{14-15}\cline{17-18}
    & nobs  & outliers  & nobs & $\tau$(s) && nobs & $\tau$(m)  & outliers                              &&& nobs  & outliers  & nobs & $\tau$(s) && nobs & $\tau$(m)  & outliers\\
\hline
\multicolumn{9}{l}{\emph{Panel A : Equity indexes}}                     &&    \multicolumn{9}{l}{\emph{Panel B (cont'd)}}\\
QQQ & 137,007&  1,720&   233 & 0.21 &&   817 & 2.18 &  670              &&    INTC& 137,138&   667&   109 & 0.17 &&   336 & 3.08 &  223 \\
SPY & 434,809& 18,705&  7,657 & 0.15 && 10,005 & 0.80 & 1,043           &&    JNJ &  53,799&   632&   154 & 0.20 &&   421 & 1.72 &   57 \\
\emph{Average} & 285,908& 10212&  3945 & 0.18 &&  5,411 & 1.49 &  856   &&    JPM & 166,702&  2,673&   895 & 0.19 &&  1,621 & 1.17 &  157 \\
\multicolumn{9}{l}{}                                                    &&    KFT &  40,369&   249&    44 & 0.19 &&   165 & 1.70 &   39 \\
\multicolumn{9}{l}{\emph{Panel B : DJIA constituents}}                  &&    KO  &  44,270&   445&   104 & 0.20 &&   294 & 1.75 &   47 \\
AA  &  71,173&   606&   118 & 0.20 &&   422 & 1.54 &   67               &&    MCD &  38,217&   433&    95 & 0.17 &&   294 & 1.48 &   44 \\
AXP &  57,510&   623&   165 & 0.18 &&   400 & 1.46 &   58               &&    MMM &  23,767&   175&    38 & 0.18 &&   113 & 1.92 &   24 \\
BA  &  31,515&   284&    52 & 0.17 &&   192 & 1.76 &   41               &&    MRK &  58,191&   588&   112 & 0.19 &&   421 & 1.36 &   55 \\
BAC & 250,433&  1,985&   474 & 0.21 &&  1,265 & 1.66 &  247             &&    MSFT& 148,724&  1,209&   226 & 0.19 &&   620 & 2.46 &  363 \\
CAT &  43,157&   408&    95 & 0.17 &&   267 & 1.64 &   46               &&    PFE &  96,037&   312&    44 & 0.17 &&   204 & 3.02 &   64 \\
CSCO& 129,261&   852&   161 & 0.19 &&   461 & 2.29 &  229               &&    PG  &  54,424&   679&   175 & 0.18 &&   447 & 1.57 &   57 \\
CVX &  59,509&   846&   226 & 0.18 &&   546 & 1.57 &   75               &&    T   &  84,676&   818&   122 & 0.21 &&   585 & 1.50 &  111 \\
DD  &  33,857&   273&    66 & 0.17 &&   178 & 1.65 &   29               &&    TRV &  21,815&   114&    24 & 0.15 &&    71 & 1.95 &   19 \\
DIS &  48,352&   317&    65 & 0.20 &&   201 & 2.11 &   51               &&    UTX &  28,421&   210&    46 & 0.17 &&   135 & 1.86 &   29 \\
GE  & 153,779&  1,031&   210 & 0.23 &&   654 & 2.17 &  167              &&    VZ  &  59,018&   456&    90 & 0.19 &&   300 & 1.70 &   67 \\
HD  &  64,837&   491&    88 & 0.19 &&   339 & 2.08 &   64               &&    WMT &  72,407&   870&   223 & 0.19 &&   564 & 1.40 &   84 \\
HPQ &  70,454&   808&   209 & 0.18 &&   526 & 1.57 &   73               &&    XOM & 116,572&  2,996&   926 & 0.19 &&  1,882 & 1.08 &  188 \\
IBM &  37,931&   474&    94 & 0.18 &&   321 & 1.81 &   59               &&    \emph{Average} &  76,544&   751&   182 & 0.19 &&   475 & 1.80 &   94 \\
\hline
\end{tabular}
\end{small}
\medskip
\begin{scriptsize}
\parbox{\textwidth}{\emph{Note.} This table reports the average daily unfiltered number of observations (``raw nobs''), the number of candidate outliers as identified by the filtering procedure (``cand. outliers''), the number of candidate outliers matched forwards (``frwd match'') and backwards (``bwrd match'') in time together with the average displacement in seconds (``$\tau$(s)'') and minutes (``$\tau$(m)''), and the number of remaining outliers that are removed from the sample (``outliers''). Because the filtering step is applied to the data before aggregation by millisecond timestamp, the number of observations here is substantially higher than what is reported in Table \ref{Table:Descriptive} of the paper. Also note that the filter is not applied to the currency pair data and hence, these are omitted from this table.}
\end{scriptsize}
\end{center}
\end{sidewaystable}

\begin{figure}[ht!]
\begin{center}
\caption{An illustration of the maxgap calculation based on simulated data.}
\label{Figure:maxgap}
\begin{tabular}{cc}
\small{Scenario I} & \small{Scenario II} \\
\includegraphics[height=8cm,width=0.48\textwidth]{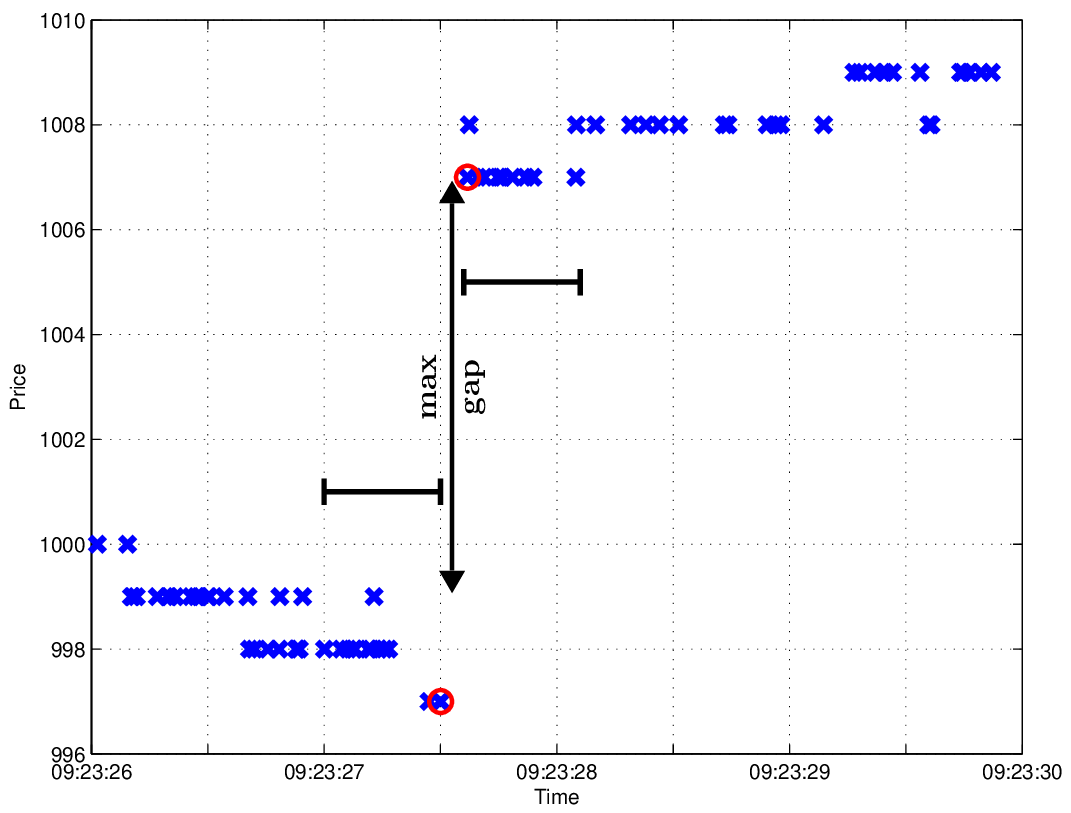} &
\includegraphics[height=8cm,width=0.48\textwidth]{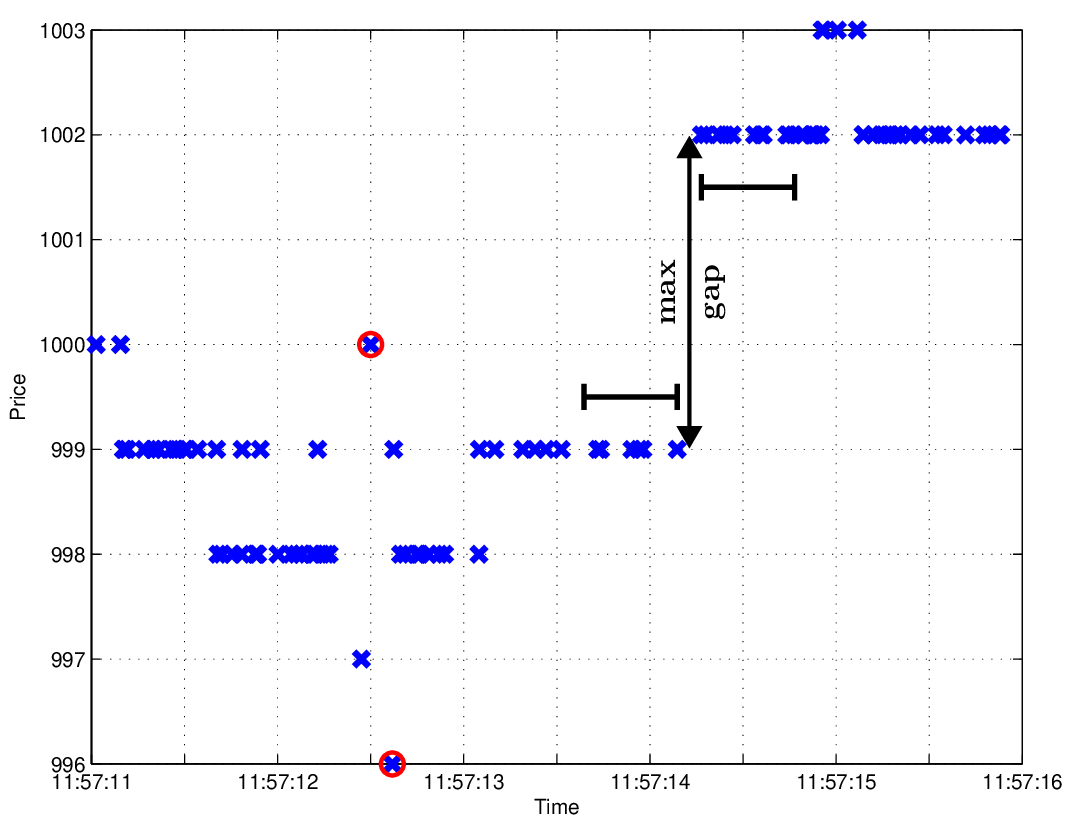} \\
\small{Scenario III} & \small{Scenario IV} \\
\includegraphics[height=8cm,width=0.48\textwidth]{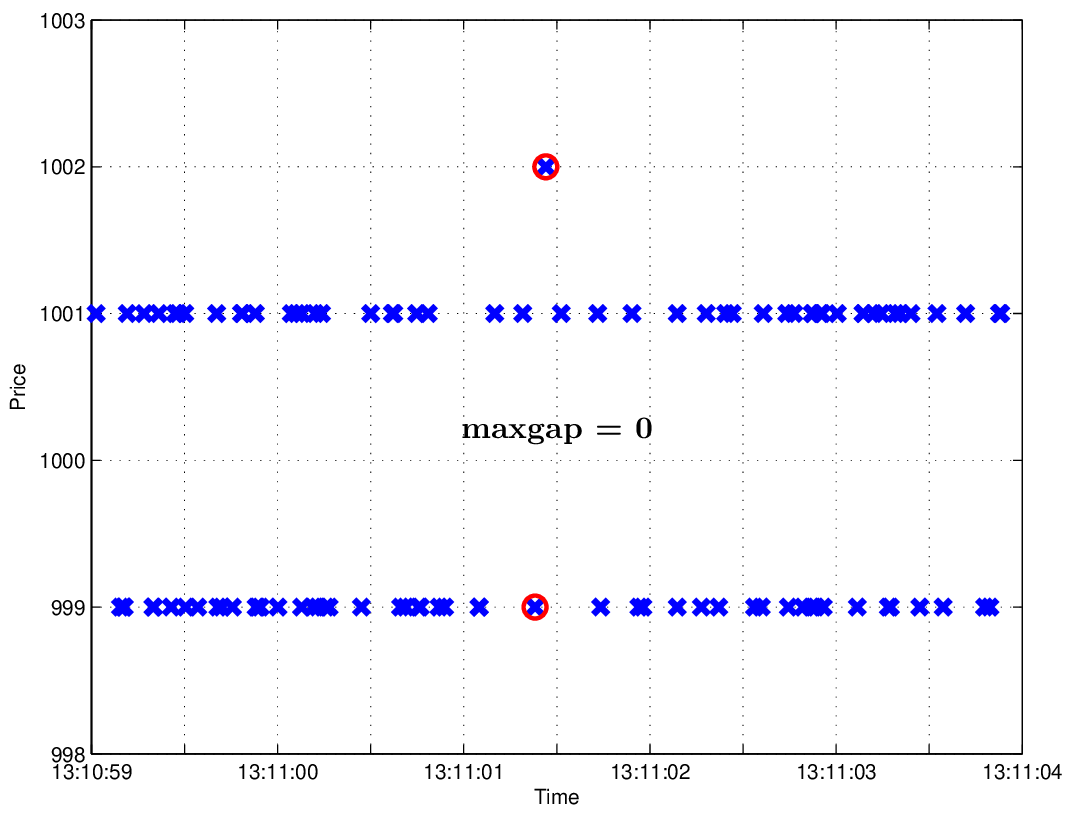} &
\includegraphics[height=8cm,width=0.48\textwidth]{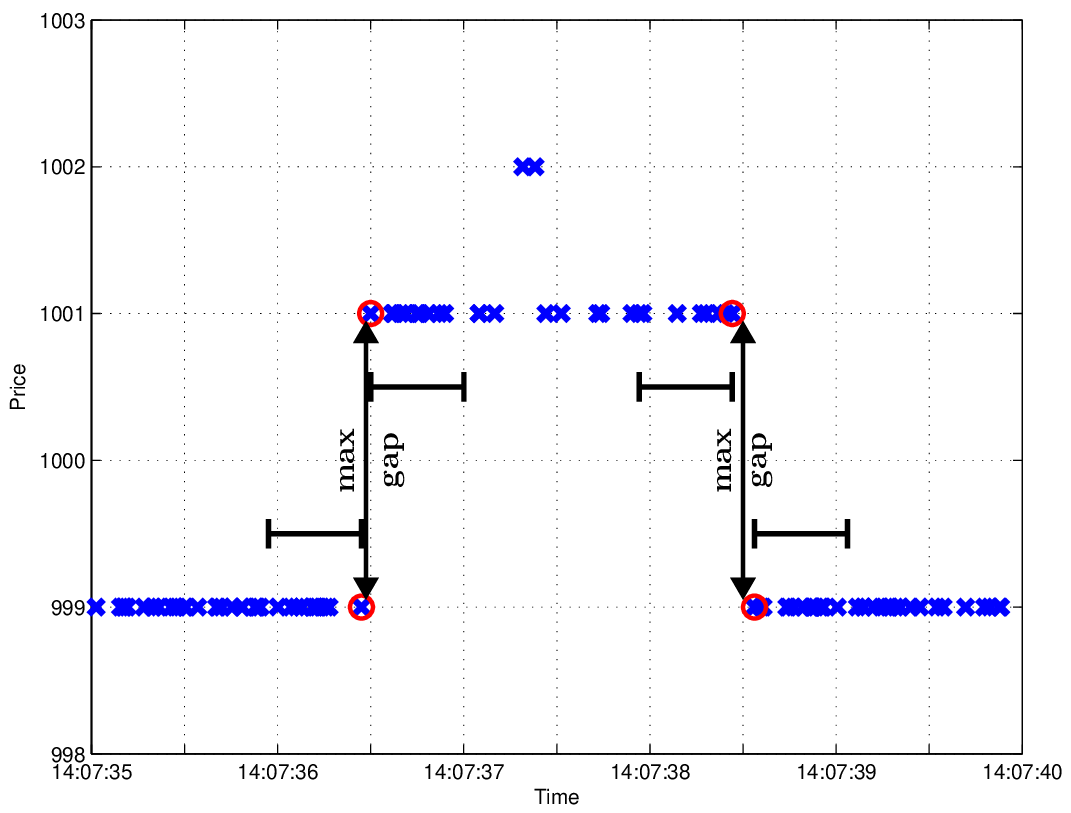} \\
\end{tabular}
\begin{scriptsize}
\parbox{\textwidth}{\emph{Note.} The crosses represent tick price observations, the circles identify the maximum price increment, the horizontal lines indicate the local window width used by the maxgap procedure. The vertical line segment is the actual maxgap (not shown in Panel C, because the maxgap is zero).}
\end{scriptsize}
\end{center}
\end{figure}

\clearpage


\small
\bibliography{userref}
\bibliographystyle{rfs}

\end{document}